\documentclass
[preprint,11pt]{elsarticle}
\usepackage[left=25mm, top=20mm, right=20mm, bottom=20mm, nohead, footskip=10mm]{geometry}
\usepackage[colorlinks=true, linkcolor=black, citecolor=black, urlcolor=black, filecolor=black, pdfborder={0 0 0}]{hyperref}
\usepackage{color}
\usepackage{graphicx}

\usepackage{float}

\usepackage{listings}

\usepackage{dsfont}

\usepackage[labelfont=bf]{caption}
\usepackage{subcaption}

\renewcommand{\gamma}{\delta}
\renewcommand{\varrho}{\eta}

\usepackage{amsmath,amsfonts,amssymb,amsthm,booktabs,  algorithm, algpseudocode}

\usepackage{mathrsfs}
\RequirePackage[OT1]{fontenc}

\allowdisplaybreaks 

\numberwithin{equation}{section}

\theoremstyle{plain}
\newtheorem{theorem}{Theorem}
\newtheorem{lemma}{Lemma}
\newtheorem{corollary}{Corollary}
\newtheorem{proposition}{Proposition}
\newtheorem{Assumption}{Assumption}
\theoremstyle{remark}
\newtheorem{remark}{Remark}

\newtheorem{definition}{Definition}

\newcommand{\ignore}[1]{} 

\numberwithin{equation}{section} %

\newcommand\cA{{\cal A}}

\newcommand\cE{{\cal E}}

\newcommand\cH{{\cal H}}

\newcommand\cF{{\cal F}}
\newcommand\cL{{\cal L}}

\newcommand\cN{{\cal N}}

\newcommand\cX{{\cal X}}

\newcommand\KL{{\bf KL}}

\newcommand\cV{{\cal V}}

\newcommand\ve{\varepsilon}

\newcommand\vphi{\varphi}

\def\bbr{{\mathbb R}}

\def\text#1{\hbox{#1}}
\def\proof{{\noindent \bf Proof. }}
\def\endproof{\mbox{\ $\qed$}}

\def\E{\mathbb{E}}
\def\P{\mathbb{P}}

\def\C{{\bf C}}

\def\H{{\bf H}}

\def\R{{\bf R}}

\def\V{{\bf V}}

\def\a{{\bf a}}
\def\b{{\bf b}}

\def\q{{\bf q}}

\def\Chi{{\bf 1}}

\def\d{\mathrm{d}}
\def\build #1_#2{\mathrel{\mathop{\kern 0pt #1}\limits_\zs{#2}}} 

\newcommand{\bro}[2]{(#1_\zs{#2})_\zs{0 \le t \le T}}
\newcommand{\rr}[2]{\mathbb{#1}^{#2}}

\newcommand{\gr}[1]{\textbf{#1}}
\newcommand{\zt}[2]{\zeta_\zs{#1}^{#2}}

\newcommand{\zs}[1]{{\mathchoice{#1}{#1}{\lower.25ex\hbox{$\scriptstyle#1$}}
{\lower0.25ex\hbox{$\scriptscriptstyle#1$}}}}

\usepackage{setspace}
\usepackage{indentfirst}

\usepackage[english]{babel}
\usepackage{verbatim}

\def\pc#1{{\textcolor{blue}{#1}}}

\begin{document}


\begin{frontmatter}
	
	\title{Optimal Investment and Entropy-Regularized Learning Under Stochastic Volatility Models with Portfolio Constraints}
	
	\author{Thai Nguyen\fnref{f1}}
\ead{thai.nguyen@act.ulaval.ca}

\author{Pertiny Nkuize\fnref{f2}}
\ead{pwnkk@ulaval.ca}

\fntext[f1]{École d'Actuariat, Université Laval, 2325 Rue de l'Université, Québec, QC G1V 0A6, Canada}

\fntext[f2]{École d'Actuariat, Université Laval, 2325 Rue de l'Université, Québec, QC G1V 0A6, Canada}

\begin{abstract}
We study the problem of optimal portfolio selection under stochastic volatility within a continuous-time reinforcement learning framework with portfolio constraints. Exploration is modeled through entropy-regularized relaxed controls, where the investor selects probability distributions over admissible portfolio allocations rather than deterministic strategies. Using dynamic programming arguments, we derive the associated entropy-regularized Hamilton–Jacobi–Bellman (HJB) equation, whose Hamiltonian involves optimization over probability measures supported on a compact control set. We show that the optimal exploratory policy takes the form of a truncated Gaussian distribution characterized by spatial derivatives of the solution of the resulting quasilinear parabolic partial differential equation.
Under suitable structural conditions on the model coefficients, we prove the existence of classical solutions to this nonlinear HJB equation for the value function.
We then establish a verification theorem and analyze the policy-improvement structure induced by the entropy-regularized Hamiltonian, showing how the resulting sequence of Partial Differential Equations (PDE) provides a continuous-time interpretation of actor–critic learning dynamics. Finally, our PDE analysis with a semi-closed form of optimal value and optimal policy enables the design of an implementable reinforcement learning algorithm by recasting the optimal problem in a martingale framework. {We numerically illustrate the results and provide quantitative comparisons of the learned policy with classical buy-and-hold policy and with that of the conventional plug-in method.}

\end{abstract}
	\begin{keyword}
		Optimal investment \sep entropy regularized \sep reinforcement learning \sep exploration \sep stochastic optimal control \sep portfolio constraint \sep quasilinear parabolic PDE
		 \MSC[2020] 93E20 \sep 93E35 \sep 91G10 \sep 49L20 \sep 35K59
	\end{keyword}
	
\end{frontmatter}


\section{Introduction}

The problem of optimal portfolio selection in continuous time originates from the seminal works of \cite{merton1969lifetime,merton1975optimum}, which established expected-utility maximization as a cornerstone of mathematical finance. In this framework, an investor dynamically allocates wealth between a risk-free asset and one or several risky assets to maximize the expected utility of terminal wealth or intertemporal consumption. The associated Hamilton–Jacobi–Bellman equation yields a nonlinear parabolic partial differential equation  whose explicit solution can be derived under power or logarithmic preferences in complete diffusion markets. The mathematical foundations of this theory were rigorously developed in \cite{KaratzasLehoczkyShreve1987,yong1999stochastic,fleming2006controlled,pham2009continuous} and further systematized in \cite{karatzas1998methods}. In these classical formulations, the model coefficients are assumed to be perfectly known,
and the optimal strategy is a deterministic feedback control. 

In realistic investment environments, however, portfolio decisions are subject to institutional and regulatory constraints. Short-selling prohibitions, borrowing limits, and leverage caps restrict the admissible set of strategies and fundamentally alter the geometry of the optimization problem. Constrained portfolio optimization has been extensively studied in \cite{cvitanic1992convex,cvitanic1993hedging,shreve1994optimal,karatzas2003optimal,cuoco1997optimal}, where convex-analytic and duality techniques were employed to characterize optimal strategies under compact or convex control sets. The presence of constraints often induces boundary solutions and transforms the associated HJB equation into a fully nonlinear variational inequality; see also \cite{bensoussan2011applications}. The interaction between state constraints and control compactness plays a central role in the analytical treatment of such problems. Another essential extension concerns stochastic volatility. Empirical evidence that volatility varies over time and is persistent motivated
diffusion-based stochastic volatility models, most notably \cite{heston1993closed}. Optimal portfolio selection in stochastic volatility environments has been investigated in \cite{kraft2005optimal,chacko2005dynamic,liu2007portfolio,kim1996dynamic,campbell1999consumption}. In such models, volatility evolves as an additional state variable, possibly correlated with asset returns, and markets typically become incomplete when volatility risk cannot be perfectly hedged. The resulting value function depends jointly on wealth and volatility factors, and the HJB equation becomes multi-dimensional and fully nonlinear. Explicit or semi-explicit solutions are rare and generally rely on specific structural assumptions.

Despite these structural refinements, classical portfolio theory relies heavily on the assumption that drift and volatility parameters are known. In practice, however, estimating expected returns is notoriously difficult. This issue was emphasized early by \cite{merton1980estimating} and later analyzed empirically in \cite{pastor2000portfolio}. Small estimation errors in the drift parameter may generate substantial deviations in optimal portfolio weights, a phenomenon often referred to as the mean-estimation problem. Robust control approaches, pioneered in \cite{hansen2001robust} and developed further in \cite{maenhout2004robust,maenhout2006robust}, address model misspecification by incorporating worst-case distortions of probability measures. Risk-sensitive control formulations, as in \cite{bielecki1999risk}, provide an alternative representation in which ambiguity is embedded directly into the objective functional.

While robust and Bayesian approaches remain model-based, recent developments in reinforcement learning (RL) advocate a model-free paradigm. RL, systematically developed in \cite{sutton1998reinforcement}, enables an agent to learn optimal behavior through interaction with the environment via exploration and exploitation. Continuous-time perspectives were proposed in \cite{doya2000reinforcement}, suggesting a connection between stochastic differential equations and learning dynamics.

A major breakthrough in continuous-time reinforcement learning is the entropy-regularized framework introduced in \cite{wang2020continuous,wang2020reinforcement}. 
The theoretical foundations of continuous time RL were substantially advanced in \cite{Jia2022policya,jia2022policyb}, where martingale methods were used to derive policy evaluation, policy gradient, and actor–critic algorithms in continuous time. These results place reinforcement learning firmly within the framework of stochastic analysis and dynamic programming. In the context of portfolio optimization, \cite{dai2023learning} studied entropy-regularized learning of Merton-type strategies in incomplete markets and identified biased Gaussian exploratory policies. {In a similar vein, \cite{bo2025optimal} developed a continuous-time q-learning algorithm for an optimal tracking portfolio problem in incomplete markets.} Complementarily, \cite{chau2025continuous} analyzed constrained entropy-regularized investment in complete markets and demonstrated the emergence of truncated Gaussian policies under compact action sets. Extensions to jump-diffusion environments were obtained in \cite{gao2024reinforcement}{, and to optimal stopping problems in \cite{dong2024randomized}}, while \cite{bender2024random} showed that sampling the control independently at every
instant averages the randomization out, and derived instead the limit dynamics of
controls sampled on a time grid, whose law coincides with that of the exploratory SDE.

The idea of replacing pointwise controls with measure-valued controls originates in the classical theory of relaxed stochastic control, notably in \cite{kushner1990numerical} and in the monograph \cite{fleming2006controlled}. In that literature, relaxed controls were introduced primarily as a compactification device: enlarging the admissible control set to probability measures restores convexity and compactness properties that may fail in the original formulation. This enlargement is essential for proving the existence of optimal controls and the stability of value functions. The entropy-regularized RL formulation builds upon this relaxed-control framework but introduces a structurally significant ingredient. This term, represented by a differential entropy in the performance criterion, discourages overly concentrated control distributions and promotes systematic exploration of the action space{. Other regularizers, such as the Choquet regularizers of \cite{han2023choquet} serve the same purpose}. Unlike classical relaxed control, where measure-valued controls are mainly technical tools, the entropy-regularized framework interprets them as intrinsic to the learning mechanism. Exploration thus becomes part of the optimization objective itself, implying important structural consequences. In particular, the associated HJB equation involves optimization over probability measures, and the entropy term induces additional convexity and smoothing effects in the Hamiltonian. When the exploration weight vanishes, the formulation converges to the classical deterministic control problem~\pc{\cite{tang2022exploratory}}, thereby establishing a rigorous bridge between reinforcement learning and stochastic optimal control. Related connections between entropy regularization and stochastic control have also been discussed in \cite{todorov2009efficient,rawlik2013stochastic}.

Despite substantial advances in continuous-time RL, several important gaps still remain in the literature. First,  entropy-regularized RL in continuous time has not been systematically studied for incomplete markets with stochastic volatility under portfolio constraints and constant risk aversion (CRRA) preferences. Second, existing studies often assume the existence of solutions to associated HJB equations without rigorous justifications \cite{Jia2022policya,jia2022policyb,dai2023learning}. While explicit solutions are available for certain models with logarithmic or quadratic preferences, CRRA preferences combined with state-dependent volatility lead to a quasilinear PDE whose classical solvability is not straightforward, and the constrained optimization over exploratory policies induces additional nonlinearities into the HJB equation. The present paper closes those gaps by developing a rigorous framework for entropy-regularized optimal investment under stochastic volatility and explicit constraints. Building on the entropy-regularized exploration paradigm of \cite{wang2020continuous}, which captures the correlation between the risky asset and the factor variable \cite{dai2023learning}, and the martingale-based framework of \cite{Jia2022policya}, we integrate the constrained exploration in \cite{chau2025continuous} into a stochastic volatility (SV) learning framework. Our main theoretical contribution is to establish existence and regularity of classical solutions to the resulting quasilinear HJB equations, rather than assuming their existence as in \cite{dai2023learning}. To this end, we employ analytical techniques from the theory of quasilinear parabolic PDEs  \cite{ladyzhenskaia1968linear,krylov1987nonlinear,lieberman1996second,crandall1983viscosity,crandall1992user}. These results provide the mathematical foundation for a rigorous verification theorem and establish the validity of the associated policy improvement procedure. 

{Our contributions can be summarized as follows. First, we derive the entropy-regularized HJB equation of the constrained problem and show that the optimal exploratory distribution is a Gaussian truncated to the action set, whose variance is proportional to the temperature and inversely proportional to risk aversion and to the squared volatility. Second, in contrast to the existing literature, we rigorously prove existence of  classical solution to the resulting HJB PDE.  In particular, by applying a homothetic ansatz, it can be reduced to a one-dimensional quasilinear parabolic equation with nonlinear gradient terms and a logarithmic entropy contribution, for which growth conditions of Ladyzhenskaya--Solonnikov type \cite{ladyzhenskaia1968linear} can be explicitly verified. Third, we show a verification theorem identifying the PDE solution with the value function and establish optimality of the truncated Gaussian in the class of relaxed entropy-regularized controls. Fourth, we obtain a policy-improvement theorem at the PDE level, in which the entropy-regularized Hamiltonian induces an explicit update of the exploratory distribution, generating a sequence of parabolic equations along which the performance functional increases monotonically. This scheme recovers the Merton feedback control as the temperature parameter vanishes. Fifth, we turn the analysis into an interpretable actor-critic algorithm that learns the value function and the exploratory policy from real S\&P~500 data.}

{Our paper differs from the exploratory SV framework of \cite{dai2023learning} in two important respects. First, the portfolio constraint restricts the admissible action space set to a compact set, leading to a truncated Gaussian policy rather than the Gaussian policy obtained in \cite{dai2023learning}. Thus, portfolio constraints  are incorporated directly into the exploratory optimization problem, rather than imposed through a post-processing step as in \cite{dai2023learning}. Second, whereas \cite{dai2023learning} assumes existence and sufficient regularity of classical solutions to the HJB equations, we rigorously establish existence and regularity of classical solutions. This is particularly challenging because optimizing over constrained exploratory policies induces additional nonlinearities into the HJB equation, leading to more complex PDE than in \cite{dai2023learning}. Our work also extends the constrained exploratory problems of \cite{chau2025continuous}, developed for a complete market and logarithmic or quadratic utility, to an incomplete market with SV and CRRA utility.}

The remainder of the paper is organized as follows: 
Section 2 introduces the stochastic-volatility market model and the admissible exploratory policy class. We also establish the well-posedness of the exploratory wealth dynamics and the admissibility conditions for truncated Gaussian policies. 
Section 3 is devoted to the analytical study of the reduced partial differential equation and establishes the existence of classical solutions under suitable structural growth conditions, together with the regularity required for feedback optimality. 
Sections 4 and 5 provide the verification theorem, prove the optimality of the truncated Gaussian exploratory policy, and analyze the limiting regime near the zero-temperature limit. 
Section 6 develops the policy evaluation and policy-improvement framework associated with the entropy-regularized Hamiltonian. 
Section 7 illustrates the learning results with S\&P~500 data. Section 8 concludes and discusses possible extensions. Finally, extra arguments can be found in the Appendix.
\paragraph*{Notations} 
Throughout this paper, $X_t$ denotes the value of any continuous-time stochastic process $X = (X_t)_{t \ge 0}$ at time $t$. We denote by $\vphi(\cdot)$ and $\Phi(\cdot)$ the probability density function (PDF) and cumulative distribution function (CDF) of the standard normal distribution $\mathcal{N}(0,1)$, respectively, while $\pi_e \approx 3.14159$ represents the mathematical constant pi (to avoid confusion with portfolio strategies $\pi$). 

For a sufficiently differentiable function $f(t, y)$, partial derivatives with respect to time $t$ are denoted interchangeably by $\partial_t f$ or $\frac{\partial f}{\partial t}$. For spatial or auxiliary variables (such as volatility $y$), partial derivatives are denoted either by subscripts, $f_y = \frac{\partial f}{\partial y}$ and $f_{yy} = \frac{\partial^2 f}{\partial y^2}$, or explicitly using the operator notation $\frac{\partial f}{\partial y}$ and $\frac{\partial^2 f}{\partial y^2}$.

\section{Market Model}
In this section, we assume that standard independent Wiener processes $(W_\zs{t})_\zs{0\le t \le T}$
	 and $(\bar{W}_\zs{t})_\zs{0\le t \le T}$ are defined on the probability space
 $(\Omega\,, \cF_\zs{T}^0\,, (\cF_t^0)_\zs{0\le t \le T}\,, \mathbb{P}) ,$
 with the fields $\cF_\zs{t}^0=\sigma\{W_\zs{v}\,,\bar{W}_\zs{v}\,,\,0\le v \le t\}$.
Our financial markets consist of one riskless bond $B=(B_\zs{t})_\zs{0\le t \le T}$ defined as
$
\d B_\zs{t}=r B_\zs{t}\d t\,,  B_\zs{0}=1,
$ 
and a risky asset $S=\bro{S}{t}$ governed by the following system
\begin{equation}
	\label{scv1}
	\begin{cases}
		\d S_\zs{t}=\mu S_\zs{t}\d t+ \sigma(y_\zs{t}) S_\zs{t}\d W_\zs{t},\qquad S_\zs{0}>0,\\
		\d y_\zs{t}=\varpi (y_\zs{t})\d t+ \gamma(y_\zs{t})\d U_\zs{t}\,,
	\end{cases}
 \end{equation} 
 where $U_\zs{t}=\rho W_\zs{t}+\sqrt{1-\rho^2}\bar{W}_\zs{t}$, with $ -1 <\rho < 1$. Here, $\mu$ and $\sigma$ are the drift and volatility parameters, respectively, and $y_\zs{0}$ is a fixed non-random initial value.  Consider an investor endowed with an initial wealth $x_0$ that will be used for investments. We assume that the investor splits her initial wealth between the two assets given above.
 We use $\pi=\{\pi_\zs{t}\}_\zs{t \in [0,T]}$ to denote the fraction of wealth that the investor invests in the risky asset. The remaining money is invested in the risk-free asset. We assume that the process $\{\pi_\zs{t}\}_\zs{t \in [0,T]}$ is $\cF^0$-adapted. The investor wealth process $X^{\pi}_t$ satisfies the following stochastic differential equation:
 \begin{equation}
     \d X_{t}^{\pi}=(r+\pi_\zs{t}(\mu-r))X_\zs{t}^{\pi}\,\d t + \sigma(y_\zs{t})\pi_\zs{t}X_\zs{t}^{\pi}\, \d W_\zs{t}, \quad X_\zs{0}=x_0>0 .\label{eq:wealth}
 \end{equation}
We assume that the investor chooses an investment strategy from the following admissible set
 \begin{align*}
     \mathcal{A}_\zs{[a,b]}:&=\bigg\{\pi |a \le \pi \le b\, \text{is progressively measurable}, X_\zs{t}^{\pi}\ge 0 \quad \,\text{for all} \,t \in [0,T], \int_{0}^{T} \pi_\zs{t}^{2}\d t < \infty
     \bigg\},
 \end{align*}
 {where $a,b$ are the portfolio bounds. Note that the short-selling constraint can be achieved by setting $a=0<b$, while the case $a<b=1$ is nothing other than the borrowing constraint.} 

Below, we suppose that the investor's preference is given by a CRRA utility function $U(x)=\frac{x^{1-\eta}-1}{1-\eta}$, where $0<\eta<1$ is the investor's risk aversion. Her objective is then to choose an admissible strategy that maximizes the terminal expected utility,
 \begin{equation}\label{classic_EU}
 	\underset{\pi \in \mathcal{A}_\zs{[a,b]}}{\max} \E[U(X_\zs{T}^{\pi})],
 \end{equation}
 subject to dynamics (\ref{eq:wealth}). {The following condition is accepted throughout the paper.}
 
\vspace{2mm}

\noindent$(\gr{SV})$ {\em The functions $\varpi$ and $\gamma$ are deterministic and continuously differentiable, such that the second equation in \eqref{scv1} has a unique strong solution.}

\noindent We remark that the exact forms of these functions are unknown to the investor.

\subsection{Exploratory wealth dynamics}

\noindent In RL settings where the parameters of the underlying model are unknown, dynamic learning becomes crucial. Dynamic learning refers to the ability of an agent to adapt and learn from her interactions with an environment where the underlying model is not known. In such cases, the agent must employ exploration strategies to interact with the environment, gather data, and learn through trial and error. In particular, at each time $t$, an action (control) $\pi_t$ is generated or sampled from the probability density function $\lambda=\{\lambda_\zs{t}(\pi), 0\le t\le T\}$. By making use of the law of large numbers,  \cite{wang2020reinforcement} considers the following exploratory version of the wealth dynamics \eqref{eq:wealth} 
  \begin{equation}\label{exploratry_dynamics_0}
  	 \d X_\zs{t}^{\lambda}=(r+(\mu-r)\E(\lambda_\zs{t}))X_\zs{t}^{\lambda}\d t+ \sigma(y_\zs{t})\sqrt{\E(\lambda_\zs{t}^{2})}X_\zs{t}^{\lambda}\d W_{t},
  \end{equation}
  where $\mathbb{E}(\lambda_t)$ is the mean of the distribution $\lambda$, i.e., $
  	\mathbb{E}(\lambda_t)=\int_{a}^{b}\pi \lambda_t(\pi|t,X_t^\lambda,y_t)\d \pi.
  $ 
Here we assume that the policy $\lambda$ is of feedback form; see e.g., \cite{Jia2022policya,jia2022policyb,chau2025continuous}. For further motivations and discussions, we refer e.g., to \cite{wang2020reinforcement, chau2025continuous,dai2023learning}. Remark that when all the parameters are known, exploration is not needed at all, and we return to dynamic \eqref{eq:wealth}.
  We notice that the exploratory wealth depends on the first and second moments of a probability distribution of $\lambda_\zs{t}$.  However, as pointed out in \cite{dai2023learning}, the exploratory wealth \eqref{exploratry_dynamics_0} does not properly take into account the correlation between $W_t$ and $\bar{W}_t$; consequently, the agent behaves as if the second moment remains constant when measuring the noise in exploration with respect to this second moment. As pointed out in \cite[Appendix A]{dai2023learning}, the quadratic variation of $\d X_\zs{t}^{\lambda}$ is given by
  \begin{align*}
  \d <{ X^{\lambda}}>_\zs{t} =\sigma^{2}(y_\zs{t})\E(\lambda_\zs{t}^{2})(X_\zs{t}^{\lambda})^{2}\d t
  =\sigma^{2}(y_\zs{t})(X_\zs{t}^{\lambda})^{2}((\E(\lambda_\zs{t}))^{2}+Var(\lambda_\zs{t}))\d t,
\end{align*}
and the exploratory wealth dynamics is given by 
\begin{equation}\label{exploratory_dynamics}
	\d X_t^{\lambda} = \big(r + (\mu - r)\mathbb{E}[\lambda(\pi \mid t, X_t^\lambda, y_t)]\big) X_t^{\lambda} \d t 
	+ \sigma(y_t) X_t^{\lambda} \big(\mathbb{E}[\lambda(\pi \mid t, X_t^\lambda, y_t)] \d W_t + \sqrt{\text{Var}(\lambda(\pi \mid t, X_t^\lambda, y_t))} \d \hat{W}_t\big),
\end{equation}
 where 
  $\hat{W}_\zs{t}$ is another Brownian motion independent of both $W_t$ and $\bar{W}_\zs{t}$. Here, a third Brownian motion $\hat{W}_t$ is introduced to capture noises resulting from exploration while maintaining an appropriate correlation between the wealth and volatility processes. This new randomness can be seen as a random number generator that the agent utilizes to devise a randomized policy. 
	Below, we expand the fields $\cF^0_\zs{t}$ with $(\hat{W}_t)_\zs{0\le t\le T}$ i.e., $\cF_\zs{t}=\sigma\{W_\zs{v}\,,\bar{W}_\zs{v}\,,\hat{W}_\zs{v},\,0\le v \le t\}$ in order to capture exploration.
 
 {Under such an exploration setting, the investor tries to choose the best exploration policy $\lambda$ that maximizes the following entropy-regularized objective (see, e.g.,\cite{wang2020continuous, chau2025continuous})
 \begin{equation}\label{explo-EU}
 	 \E\bigg[U(X_\zs{T}^{\lambda})+m\int_0^T \cE nt(\lambda_s) \d s\bigg],
 \end{equation}	where  $\cE nt(\cdot)$ is Shannon's differential entropy, defined by $\cE nt(\lambda)=-\int_\zs{a}^{b}\lambda(\pi)\ln \lambda (\pi) \d \pi,$ and $m$ is a parameter that measures the level of exploration. This means that the objective is to achieve the best trade-off between exploration and exploitation. 	Here, $\lambda$ can be seen as a probability density in $[a,b]$ so that the entropy can be properly defined.}
\subsection{Admissible exploration policies}
Let $\cH_{[a,b]}$ be the set of (feedback) admissible policies $\lambda$ that are characterized as follows:
 \begin{definition}\label{def_1}

The policy $\lambda$ is admissible if it satisfies the following properties:
 \begin{itemize}
     \item [(i)] For each $ (s,x,y) \in [t,T]\times\bbr^+\times\bbr$, $\lambda(\cdot|s,x,y)$ is a density function on $[a,b]$ ;
     \item [(ii)] The mapping $(\pi, s,x,y)\mapsto \lambda(\pi|s,x,y)$ is jointly Borel measurable from $[a,b] \times [t,T] \times \mathbb{R}^+ \times \mathbb{R}  \, \,\text{to} \,\, \mathbb{R^+};$
         \item [(iii)] the SDE \eqref{exploratory_dynamics} admits a unique strong, non negative solution for  any initial $(t,x,y)\in [0, T]\times \bbr^+\times \bbr$  and for any $1 < q < \frac{1}{1-\eta}$, $\mathbb{E}[(X^{\lambda}_T)^{q(1-\eta)}]<\infty$;
     \item [(iv)] The mapping $(s,x,y)\mapsto  \cE nt( \lambda(.|s,x,y))$ is bounded on every compact set $K$ of $[0, T]\times (0, +\infty)\times \mathbb{R}$.
 \end{itemize}
  \end{definition}

When the density $\lambda$ is independent of the wealth level, i.e.,  \( \lambda_t = \lambda(\pi |t, y) \), Proposition \ref{Prop1} below shows that the SDE \eqref{exploratory_dynamics} admits a unique strong solution that satisfies \( X^{\lambda}_s \geq 0 \) for \( t \leq s \leq T \). 
	\begin{proposition}\label{Prop1}
Assume that Condition \textbf{(SV)} holds. Let $\lambda(\cdot|t,y)$ be a probability density function on $[a,b]$, jointly Borel in $(\pi,t,y)$, with mean $m(t,y)$ and variance $M(t,y)$. Suppose that $m$ and $M$ are deterministic, continuous in both variables, satisfying $m^2(t,y)+M(t,y)>0$ for all $t$ and $y$, and that $(t,y)\mapsto\cE nt(\lambda(\cdot|t,y))$ is bounded on compact sets. Then $\lambda$ is admissible.
\end{proposition}

\begin{proof}
	We rewrite \eqref{exploratory_dynamics}:
	\begin{equation}\label{exploratory_dynamics_modified}
		\d X_t^{\lambda} = \big(r + (\mu - r)\,m(t, y_t)\big) X_t^{\lambda} \,\d t 
		+ \sigma(y_t) X_t^{\lambda} \big(m(t, y_t) \,\d W_t + \sqrt{M(t, y_t)} \,\d \hat{W}_t\big)\,,
	\end{equation}
which can be rewritten as
\begin{equation}\label{ex-check-coef-1}
		\d X_t^{\lambda} = \check{k}_t X_t^{\lambda} \,\d t 
		+ \check{h}_tX_t^{\lambda} \,\d \check{W}_t,
\end{equation}
where
$$
\check{k}_t=r + (\mu - r)\,m(t, y_t)
\quad\mbox{and}\quad
\check{h}_t=\sigma(y_t) \sqrt{m^{2}(t, y_t)+M(t, y_t)},
$$
and
$$
\check{W}_\zs{u}=\int^{u}_\zs{0} \,\left( \frac{m(t, y_t)}{\sqrt{m^{2}(t, y_t)+M(t, y_t)}} \,\d W_t  +  \frac{ \sqrt{M(t, y_t)}}{\sqrt{m^{2}(t, y_t)+M(t, y_t)}}\,\d \hat{W}_t\right), \quad u\in [0,T],
$$
is a Wiener process. Note that the processes $(\check{k}_t)_\zs{0\le t\le T}$ and $(\check{h}_t)_\zs{0\le t\le T}$ are continuous. Therefore, by Itô's formula, 
there exists a strong solution to \eqref{ex-check-coef-1} which is given by
\begin{equation}\label{ex-check-coef-22}
 X_t^{\lambda}= X_0^{\lambda}\,
 \exp\left\{ \int^{t}_0 \big(\check{k}_u -\check{h}^{2}_u/2\big)\d u +  \int^{t}_0 \check{h}_u \d \check{W}_\zs{u}\right\}.
\end{equation}
To show uniqueness, we set
\begin{equation}\label{st-tim}
	\tau_n := \inf \{ t \geq 0 : |\check{k}_t|+|\check{h}_t| \geq n \} \wedge T,
\end{equation}
which is a stopping time. Taking into account that $\max_\zs{0\le t\le T} \big( |\check{k}_t|+|\check{h}_t|\big) <\infty$ a.s., we get that  $\tau_n\to T$ a.s. as $n\to\infty$. Now, set
$$
\d X_t^{\lambda,n} = \check{k}^{n}_t X_t^{\lambda} \,\d t 
		+ \check{h}^{n}_t X_t^{\lambda} \,\d \check{W}_t,
$$
where 
$\check{k}^{n}_t=\check{k}_{t\wedge \tau_n}$ and 
$\check{h}^{n}_t=\check{h}_{t\wedge \tau_n}$. This equation has a unique solution which is
$$
 X_t^{\lambda,n}= X_0^{\lambda}\,
 \exp\left\{ \int^{t}_0 \big(\check{k}^{n}_u -(\check{h}^{n}_u)^{2}/2\big)\d u +  \int^{t}_0 \check{h}^{n}_u \d \check{W}_\zs{u}\right\}
\,.
$$
Note that on the interval $[0,\tau_n]$, the processes $X_t^{\lambda,m}$
coincide with \eqref{ex-check-coef-22} for any $m\ge n$. 
Let $\zeta_t$ be a  solution of \eqref{ex-check-coef-1}; then, as $n\to\infty$
$$
\P\left(\max_{0\le t\le T}|\zeta_t-X_t^{\lambda,n}|>0 \right)
=\P(\tau_n<T)
\to 0\,.
$$
So, any solution $\zeta$ of the equation \eqref{ex-check-coef-1} can be represented as the uniform limit in probability of the processes $X_t^{\lambda,n}$, so $\zeta$ coincides with \eqref{ex-check-coef-22}; hence, uniqueness. 
It remains to prove that $\mathbb{E}[(X_T^{\lambda})^{q(1-\eta)}] < \infty$ for some $1 < q < \frac{1}{1-\eta}$. From the explicit representation of the wealth process
\eqref{ex-check-coef-22}, we have:
\begin{align*}
(X_T^\lambda)^{q(1-\eta)}
&=
(X_0^\lambda)^{q(1-\eta)}
\exp\left\{
\int_0^T q(1-\eta)\check{k}_u\,\d u
-
\int_0^T \frac{q(1-\eta)}{2}\check{h}_u^2\,\d u
+
\int_0^T q(1-\eta)\check{h}_u\,\d\check{W}_u
\right\}
\\
&=
(X_0^\lambda)^{q(1-\eta)}
\exp\left\{
\int_0^T
\left(
q(1-\eta)\check{k}_u
+
\frac{q^2(1-\eta)^2-q(1-\eta)}{2}
\check{h}_u^2
\right)\d u
\right\}
\\
&\quad\times
\exp\left\{
\int_0^T q(1-\eta)\check{h}_u\,\d\check{W}_u
-
\frac12
\int_0^T
\big(q(1-\eta)\check{h}_u\big)^2\,\d u
\right\}.
\end{align*}
Since $q(1-\eta)<1$, we have
$
q^2(1-\eta)^2-q(1-\eta)
=
-q(1-\eta)\big(1-q(1-\eta)\big)<0.
$ 
Moreover, since the exploratory policy is supported on the compact interval
$[a,b]$, its first moment $m(t,y)$ satisfies $|m(t,y)|\leq \max\{|a|,|b|\}.$ Consequently, $
\check{k}_u
=
r+(\mu-r)m(u,y_u)
\leq
|r|+|\mu-r|\max\{|a|,|b|\}
=:K.
$ 
It follows that
\begin{align*}
\int_0^T
\left(
q(1-\eta)\check{k}_u
-
\frac{q(1-\eta)(1-q(1-\eta))}{2}
\check{h}_u^2
\right)\d u
\leq
q(1-\eta)KT.
\end{align*}
Now define
\[
Z_T
:=
\exp\left\{
\int_0^T q(1-\eta)\check{h}_u\,\d\check{W}_u
-
\frac12
\int_0^T
\big(q(1-\eta)\check{h}_u\big)^2\,\d u
\right\}.
\]
Since $(Z_t)_{t\in[0,T]}$ is a nonnegative local martingale, it is a
supermartingale. Therefore,
\(
\mathbb{E}[Z_T]\leq Z_0=1.
\) Taking expectations, we obtain
\begin{align*}
\mathbb{E}\left[(X_T^\lambda)^{q(1-\eta)}\right]
&\leq
(X_0^\lambda)^{q(1-\eta)}
e^{q(1-\eta)KT}
\mathbb{E}[Z_T]
\leq
(X_0^\lambda)^{q(1-\eta)}
e^{q(1-\eta)KT}
<\infty,
\end{align*}
which completes the proof.\end{proof}

\subsection{Recursive entropy-penalized optimization problem}
 Due to the exploration terms, the entropy-regularized optimization problem \eqref{explo-EU} does not admit closed-form solutions, except in special cases, for example, when $\eta=1$ and the utility function is logarithmic (see \cite{chau2025continuous}). To make learning more efficient, we consider the following recursive entropy-penalized objective  
	\begin{equation}\label{aim-recursive}
		J^{\lambda}(t,X^{\lambda}_t,y_t;m)=\mathbb{E}\left[\int_{t}^{T}m[(1-\varrho)J^{\lambda}(s,X^{\lambda}_s,y_s;m)+1]\cE nt(\lambda_\zs{s})\d s+U(X_\zs{T}^{\lambda})\bigg|\cF_t		\right].
	\end{equation}
Here $\lambda$ is a feedback policy and $(X^{\lambda},y)$ is Markov, so the conditional expectation depends on $\cF_t$ only through $(X^{\lambda}_t,y_t)$; we write $J^{\lambda}(t,X^{\lambda}_t,y_t;m)$ to make this explicit.
The recursive form \eqref{aim-recursive} includes the entropy in an endogenous way; see e.g., \cite{dai2023learning} for a similar  setting.  In particular, it combines the terminal expected utility (the second term) and the utility-dependent (the first term) factor on exploration, which is determined recursively.


Below, we study the recursive problem \eqref{aim-recursive}. Note that  \eqref{aim-recursive} coincides with \eqref{explo-EU} when $\eta=1$. We observe that the recursive form \eqref{aim-recursive} can be connected to a BSDE via the martingale representation theorem as follows:

 \begin{lemma}\label{Lemma1}
 For any given policy $\lambda$, consider $J^{\lambda}$ defined in \eqref{aim-recursive}. Define the martingale 
 	\begin{align}
 		M_t = \mathbb{E}\left[ \int_{0}^{T} m \left[ (1-\varrho) J^{\lambda}(s,X^{\lambda}_s,y_s;m) + 1 \right] \mathcal{E}nt(\lambda_s)  \d s + U(X_T^\lambda)  \bigg| \cF_t \right], \quad t\in [0, T].
 	\end{align}
 	
 	\noindent Then there exists a predictable process $Z_t$
 	such that
 	\begin{align}
 		M_t = M_0 + \int_{0}^{t} Z_s \cdot  \d\gr{W}_s, \quad t \in [0, T],
 	\end{align}
 	
 	\noindent where $\cdot$ denotes the inner product, $\gr{W}_s=(\bar{W}_s,W_s, \hat{W}_s)^{\top}$, and $\top$ denotes the transposition. Furthermore, it holds that
 	\begin{align}
 		J^{\lambda}(t,X^{\lambda}_t,y_t;m) &= M_t - \int_{0}^{t} m \left[ (1-\eta) J^{\lambda}(s,X^{\lambda}_s,y_s;m) + 1 \right] \cE nt(\lambda_s) \d s,
 	\end{align}
 	and
 	\begin{align}\label{BSDE_1}
 		\d J^{\lambda}(t,X^{\lambda}_t,y_t;m) &= -m \left[ (1-\eta) J^{\lambda}(t,X^{\lambda}_t,y_t;m) + 1 \right]\cE nt(\lambda_t) \d t + Z_t \cdot \d\gr{W}_t.
 	\end{align}
 \end{lemma}
 \noindent
 {In such a context of recursive utility, the term \( m[(1-\eta)J^{\lambda} + 1] \) represents the generator, which encapsulates the relationship between immediate utility and future continuation utility, while the process \( Z \) corresponds to the hedging portfolio, capturing the sensitivity of the utility to variations in the underlying stochastic factors; see e.g.,\cite{duffie1992stochastic,el1997backward}. Probabilistic representations through BSDEs and further connections with dynamic programming and nonlinear PDEs can be found in e.g. \cite{peng1991probabilistic,el1997backward,pardoux2005backward}. Observe that \eqref{BSDE_1} is a linear BSDE whose unique solution (see e.g.,  \cite{el1997backward}) can be expressed in the following form
	\begin{equation}\label{Re-Sol}
		J^{\lambda}(t,X^{\lambda}_t,y_t;m)=\E \bigg[\int_{t}^{T}e^{\int_{t}^{s}m(1-\eta)\cE nt(\lambda_u)\d u}m\cE nt(\lambda_s)\d s+e^{\int_{t}^{T}m(1-\eta)\cE nt(\lambda_u)\d u}U(X^{\lambda}_T)\bigg| \cF_t\bigg].
	\end{equation}
	Therefore, the investor's objective is now stated as follows: 
	\begin{equation} \label{exploratory}
\underset{\lambda\in \cH_{[a,b]}}{\sup}\E \bigg[\int_{t}^{T}e^{\int_{t}^{s}m(1-\eta)\cE nt(\lambda_u)\d u}m\cE nt(\lambda_s)\d s+e^{\int_{t}^{T}m(1-\eta)\cE nt(\lambda_u)\d u}U(X^{\lambda}_T)\bigg| X_\zs{t}^{\lambda}=x,y_\zs{t}=y
 	\bigg].
 	\end{equation}
    \begin{remark}\label{Rem1} Under some smoothness assumptions, using Proposition 4.3 in \cite{el1997backward}, we can link the BSDE \eqref{BSDE_1} with the following PDE:
   \begin{equation}\label{PDE_BSDE}
 	 \frac{\partial J}{\partial t}+\mathcal{A^{\lambda}} J(t,x,y)+ m\cE nt(\lambda_t)[(1-\varrho)J(t,x,y)+1] = 0, \quad J(T,x,y)=U(x),
 \end{equation}
where, for each fixed $\pi\in[a,b]$ and $v$, a smooth function, the local operator $\mathcal A^\pi$ is defined by
\[
\mathcal A^\pi v
:=
(r+(\mu-r)\pi)x\,v_x
+\frac12\sigma^2(y)x^2\pi^2\,v_{xx}
+\varpi(y)\,v_y
+\frac12\delta^2(y)\,v_{yy}
+\rho\,\delta(y)\sigma(y)x\pi\,v_{xy},
\]
and the relaxed operator $\mathcal A^\lambda$ is given by
\begin{align}\label{operator_A}
\mathcal A^\lambda v
=
\int_a^b \mathcal A^\pi v\,\lambda(\pi)\,\d\pi.
\end{align}
Then, by the Feynman-Kac formula (see \cite{el1997backward} or Theorem \ref{Thm:FeymanKac}), the solution to the PDE \eqref{PDE_BSDE} can be expressed as:
 \begin{align}\label{eq:new-EU}
 	J(t,x,y)=\E \bigg[\int_{t}^{T}e^{\int_{t}^{s}m(1-\eta)\cE nt(\lambda_u)\d u}m\cE nt(\lambda_s)\d s+e^{\int_{t}^{T}m(1-\eta)\cE nt(\lambda_u)\d u}U(X^{\lambda}_T)\bigg| X^{\lambda}_t=x, y_t=y\bigg].
 \end{align} 
  	\end{remark}
    
\begin{lemma}
    Let $(t,x,y)\in [0, T]\times \mathbb{R^+}\times \mathbb{R}$, and for any $\lambda \in \cH_\zs{[a,b]}$ set 
    \begin{align}\label{lem_2}
        D^{\lambda}_{t,s}:=e^{\int_t^s m(1-\eta)\cE nt(\lambda_u)\d u}, \quad s \in [t, T].
    \end{align}
    Then $
    \int_t^T D^{\lambda}_{t,s}m\cE nt(\lambda_s)\d s=\frac{D^{\lambda}_{t,T}-1}{1-\eta},$ 
    and the objective \eqref{eq:new-EU} admits the following equivalent form    \begin{align}\label{multiplicative_wealth}
         J^{\lambda}(t,x,y;m)=\mathbb{E}\bigg[U\bigg(e^{\int_t^Tm\cE nt(\lambda_s)\d s}X^{\lambda}_T\bigg)\bigg | X^{\lambda}_t=x, y_t=y\bigg].
    \end{align}
\end{lemma}

\begin{proof}
    The proof is straightforward and hence is omitted.
\end{proof}

	 \begin{remark}\label{remark2}
By \eqref{multiplicative_wealth}, the recursive objective is the expected utility of the terminal wealth multiplied by $e^{\int_t^T m\cE nt(\lambda_s)\d s}$, meaning that the exploration reward acts multiplicatively on wealth. 
For $\eta=1$ or $U(x)=\ln x$, \eqref{exploratory} becomes \begin{equation} \label{classic_Merton}
 		\underset{\lambda\in \cH_{[a,b]}}{\sup}\mathbb{E}\left[\int_{t}^{T}m\cE nt(\lambda_\zs{s})\d s+U(X_\zs{T}^{\lambda})|X^{\lambda}_\zs{t}=x,y_\zs{t}=y
 		\right],
 	\end{equation}
 which is already studied in \cite{chau2025continuous}. If $m=0$, then we are back to the classical Merton problem.
 	\end{remark}
\begin{remark}\label{se:Mertonconnection}
Let
\begin{equation*}
\tilde X^{\lambda}_s:=e^{\int_t^s m\cE nt(\lambda_u)\d u}X^{\lambda}_s,\quad s\in[t,T].
\end{equation*}
By \eqref{multiplicative_wealth}, we have $J^{\lambda}(t,x,y;m)=\mathbb{E}[U(\tilde X^{\lambda}_T)\,|\,\tilde X^{\lambda}_t=x,y_t=y]$, so that the exploratory problem is a Merton problem for $\tilde X^{\lambda}$, with no running reward. It\^o's formula and \eqref{exploratory_dynamics} give
\begin{align*}
\d\tilde X^{\lambda}_s&=\tilde X^{\lambda}_s\Big[\big(r+(\mu-r)\mathbb{E}[\lambda(\pi\mid s,X^{\lambda}_s,y_s)]+m\cE nt(\lambda(\cdot\mid s,X^{\lambda}_s,y_s))\big)\d s+\sigma(y_s)\Big(\mathbb{E}[\lambda(\pi\mid s,X^{\lambda}_s,y_s)]\,\d W_s\nonumber\\
&+\sqrt{\text{Var}(\lambda(\pi\mid s,X^{\lambda}_s,y_s))}\,\d\hat W_s\Big)\Big],\qquad \tilde X^{\lambda}_t=x .
\end{align*}
The entropy thus results in additional return at rate $m\cE nt(\lambda(\cdot\mid s,X^{\lambda}_s,y_s))$, which justifies the recursive weight $e^{\int_t^s m(1-\eta)\cE nt(\lambda_u)\d u}$ in \eqref{eq:new-EU}. When the policy depends only on $(s,y_s)$, the coefficients of the last exploratory SDE no longer depend on $X^{\lambda}$. 
\end{remark}

\begin{lemma}\label{lem:concat}
 For any $\lambda, \lambda^{\ve}\in \cH_\zs{[a,b]}$ ,
 we define the concatenated policy
\begin{equation*}
\lambda\otimes_\zs{\bar t}\lambda^{\ve}(\cdot\,|\,s,x,y):=
\lambda(\cdot\,|\,s,x,y)\Chi_\zs{s< \bar t}+\lambda^{\ve}(\cdot\,|s,x,y)\Chi_\zs{s\ge \bar t}
\end{equation*}
Then $\lambda\otimes_\zs{\bar t}\lambda^{\ve}\in  \cH_\zs{[a,b]}$. 
\end{lemma}
\begin{proof}
\noindent (i) For each $s$ and each value of the observations, $\lambda\otimes_\zs{\bar t}\lambda^{\ve}$ coincides with $\lambda(\cdot|s,x,y)$ or with $\lambda^{\ve}(\cdot|s,x,y)$, hence it is a density on $[a,b]$.

\noindent (ii)The map $(\pi,s,x,y)\mapsto\lambda \otimes_\zs{\bar t}\lambda^{\ve}(\pi|s,x,y)$ is measurable as a sum and product of two measurable functions.

\noindent (iii) Fix an initial condition $(t,x,y)$ with $t<\bar t$. On $[t,\bar t]$, the
wealth process is $X^{\lambda}$, which exists, is unique, and is nonnegative by
Definition~\ref{def_1}~(iii) for $\lambda$. On $[\bar t,T]$ it solves the wealth equation
of $\lambda^{\ve}$ started from $X^{\lambda}_\zs{\bar t}$. Since
$\lambda^{\ve}\in\cH_{[a,b]}$, this equation has a unique strong nonnegative solution for
every deterministic initial condition, hence also for the $\cF_\zs{\bar t}$-measurable
initial value $(X^{\lambda}_\zs{\bar t},y_\zs{\bar t})$. The wealth process
$X^{\lambda\otimes_\zs{\bar t}\lambda^{\ve}}$ is therefore the unique strong nonnegative
solution on $[t,T]$. As in the proof of Proposition~\ref{Prop1},
 we can verify that 
$\E\big[(X^{\lambda\otimes_\zs{\bar t}\lambda^{\ve}}_\zs{T})^{q(1-\eta)}\big]\le x^{q(1-\eta)}e^{q(1-\eta)K}<\infty$ for any $1<q<(1-\eta)^{-1}.$
 \noindent (iv) Let $K$ be a compact subset of $[0,T]\times(0,\infty)\times\bbr$. For every $(s,x,y)\in K$,
the density $\lambda\otimes_\zs{\bar t}\lambda^{\ve}(\cdot|s,x,y)$ coincides with $\lambda(\cdot|s,x,y)$ if $s<\bar t$
and with $\lambda^{\ve}(\cdot|s,x,y)$ if $s\ge\bar t$. Hence
\begin{equation*}
\sup_{(s,x,y)\in K}\big|\cE nt\big(\lambda\otimes_\zs{\bar t}\lambda^{\ve}(\cdot|s,x,y)\big)\big|
\le\max\Big(\sup_{K}\big|\cE nt\big(\lambda(\cdot|s,x,y)\big)\big|,\ \sup_{K}\big|\cE nt\big(\lambda^{\ve}(\cdot|s,x,y)\big)\big|\Big)<\infty,
\end{equation*}
by (iv) thanks to the admissibility of $\lambda$ and $\lambda^{\ve}$.\end{proof}

    \begin{proposition}[Dynamic programming principle]
Let the assumptions of Definition \ref{def_1} hold. Then, for any \(\bar t\in[t,T]\),
$$ V(t,x,y) = \sup_{\lambda\in\cH_{[a,b]}} \E\left[ \int_t^{\bar t} e^{\int_t^s m(1-\eta)\cE nt(\lambda_u)\d u}\, m\cE nt(\lambda_s)\d s + e^{\int_t^{\bar t} m(1-\eta)\cE nt(\lambda_u)\d u}\, V(\bar t,X_{\bar t}^\lambda,y_{\bar t}) \,\middle|\, X^{\lambda}_t=x,\,y_t=y\right]. $$
In particular, the control problem is time consistent in the Bellman sense.
\end{proposition}
\begin{proof}
    We recall $D^{\lambda}_\zs{t,s}=e^{\int_t^sm(1-\eta)\cE nt(\lambda_u)\d u}$, which satisfies the multiplicative property $D^{\lambda}_\zs{t,s}=D^{\lambda}_\zs{t,\bar t}D^{\lambda}_\zs{\bar t,s}$ for $t\le \bar t\le s\le T$.
    Let
    \begin{equation*}
    \Theta(t,x,y):=\underset{\lambda\in \cH_\zs{[a,b]}}{\sup}\E\bigg[\int_{t}^{\bar t}D^{\lambda}_\zs{t,s}m\cE nt(\lambda_s)\d s+D^{\lambda}_\zs{t,\bar t}V(\bar t,X^{\lambda}_\zs{\bar t},y_\zs{\bar t})\bigg| X^{\lambda}_t=x,y_t=y\bigg].
    \end{equation*}

    \noindent We first show $V(t,x,y)\le \Theta(t,x,y)$. Fix $\lambda\in \cH_\zs{[a,b]}$. Splitting the time integral at $\bar t$, factorising the discount and conditioning on $\cF_\zs{\bar t}$,
    	\begin{align}
		J^{\lambda}(t,x,y;m)&=\E \bigg[\int_{t}^{\bar t}D^{\lambda}_\zs{t,s}m\cE nt(\lambda_s)\d s+\int_{\bar t}^{T}D^{\lambda}_\zs{t,s}m\cE nt(\lambda_s)\d s+ D^{\lambda}_\zs{t,T}U(X^{\lambda}_T)\bigg| X^{\lambda}_t=x,y_t=y\bigg].\nonumber \\
        &=\E \bigg[\int_{t}^{\bar t}D^{\lambda}_\zs{t,s}m\cE nt(\lambda_s)\d s\bigg| X^{\lambda}_t=x,y_t=y\bigg]\nonumber\\
        &\quad+\E\bigg[D^{\lambda}_\zs{t,\bar t}\E\bigg[\int_{\bar t}^{T}D^{\lambda}_\zs{\bar t,s}m\cE nt(\lambda_s)\d s+ D^{\lambda}_\zs{\bar t,T}U(X^{\lambda}_T)\bigg| \cF_{\bar t}\bigg]\bigg| X^{\lambda}_t=x,y_t=y\bigg]\nonumber\\
        &=\E\bigg[ \int_{t}^{\bar t}D^{\lambda}_\zs{t,s}m\cE nt(\lambda_s)\d s+D^{\lambda}_\zs{t,\bar t}J^{\lambda}(\bar t,X^{\lambda}_\zs{\bar t},y_\zs{\bar t};m)\bigg| X^{\lambda}_t=x,y_t=y\bigg],
	\end{align}
    the last equality holding because $\lambda$ is a feedback policy and $(X^{\lambda},y)$ is Markov, so that the inner conditional expectation is $J^{\lambda}(\bar t,X^{\lambda}_\zs{\bar t},y_\zs{\bar t};m)$. We have that $J^{\lambda}(\bar t,X^{\lambda}_\zs{\bar t},y_\zs{\bar t};m)\le V(\bar t,X^{\lambda}_\zs{\bar t},y_\zs{\bar t})$ and $D^{\lambda}_\zs{t,\bar t}>0$, then we can conclude that
    \begin{equation*}
    J^{\lambda}(t,x,y;m)\le \E\bigg[ \int_{t}^{\bar t}D^{\lambda}_\zs{t,s}m\cE nt(\lambda_s)\d s+D^{\lambda}_\zs{t,\bar t}V(\bar t,X^{\lambda}_\zs{\bar t},y_\zs{\bar t})\bigg| X^{\lambda}_t=x,y_t=y\bigg]\le \Theta(t,x,y),
    \end{equation*}
    and taking the supremum over $\lambda\in \cH_\zs{[a,b]}$ gives $V(t,x,y)\le \Theta(t,x,y)$.

    \noindent Next we show  $\Theta(t,x,y)\le V(t,x,y)$. Fix $\lambda\in \cH_\zs{[a,b]}$ and $\ve>0$. Since $V$ is a supremum, at every state there is a policy which is $\ve$-optimal; as the elements of $\cH_\zs{[a,b]}$ are feedback maps, measurable in their arguments by $(ii)$ of Definition \ref{def_1}, a measurable selection $\lambda^{\ve}\in \cH_\zs{[a,b]}$ can be made for which
    \begin{equation*}
    J^{\lambda^{\ve}}(\bar t,X^{\lambda}_\zs{\bar t},y_\zs{\bar t};m)\ \ge\ V(\bar t,X^{\lambda}_\zs{\bar t},y_\zs{\bar t})-\ve 
    \end{equation*}
    Let $\lambda\otimes_\zs{\bar t}\lambda^{\ve}:=\lambda_s\Chi_\zs{\{s<\bar t\}}+\lambda^{\ve}_s\Chi_\zs{\{s\ge \bar t\}}$ be the concatenated policy, which is admissible by Lemma \ref{lem:concat}. It coincides with $\lambda$ on $[t,\bar t)$, hence it produces the same state path and the same discount factor. The above factorization gives
    \begin{align}
    J^{\lambda\otimes_\zs{\bar t}\lambda^{\ve}}(t,x,y;m)&=\E\bigg[ \int_{t}^{\bar t}D^{\lambda}_\zs{t,s}m\cE nt(\lambda_s)\d s+D^{\lambda}_\zs{t,\bar t}J^{\lambda^{\ve}}(\bar t,X^{\lambda}_\zs{\bar t},y_\zs{\bar t};m)\bigg| X^{\lambda}_t=x,y_t=y\bigg]\nonumber\\
    &\ge \E\bigg[ \int_{t}^{\bar t}D^{\lambda}_\zs{t,s}m\cE nt(\lambda_s)\d s+D^{\lambda}_\zs{t,\bar t}V(\bar t,X^{\lambda}_\zs{\bar t},y_\zs{\bar t})\bigg| X^{\lambda}_t=x,y_t=y\bigg]-\ve\, \E\bigg[D^{\lambda}_\zs{t,\bar t}\bigg| X^{\lambda}_t=x,y_t=y\bigg].\notag
    \end{align}
    The left-hand side is bounded above by $V(t,x,y)$. Moreover, since $\cE nt(\lambda)\le\ln(b-a)$ for every density $\lambda$ on $[a,b]$ and $1-\eta>0$, we have $D^{\lambda}_\zs{t,\bar t}\le C_\zs{D}:=e^{m(1-\eta)T(\ln(b-a))^{+}}$. Therefore
    \begin{equation*}
    \E\bigg[ \int_{t}^{\bar t}D^{\lambda}_\zs{t,s}m\cE nt(\lambda_s)\d s+D^{\lambda}_\zs{t,\bar t}V(\bar t,X^{\lambda}_\zs{\bar t},y_\zs{\bar t})\bigg| X^{\lambda}_t=x,y_t=y\bigg]\le V(t,x,y)+\ve\, C_\zs{D}.
    \end{equation*}
    Letting $\ve\downarrow 0$ at fixed $\lambda$ and taking the supremum over $\lambda\in \cH_\zs{[a,b]}$ yields $\Theta(t,x,y)\le V(t,x,y)$.\end{proof}


 \subsection{ Hamilton-Jacobi-Bellman Equation and optimal control distribution }
In the rest of the paper, we will investigate the exploratory optimization problem \eqref{eq:new-EU} using the stochastic dynamic programming principle. The optimal value function is given by
\begin{align}\label{goal}
V(t,x,y;m):=\underset{\lambda \in \cH_\zs{[a,b]}}{\sup}
J^\lambda(t,x,y;m) 
 \end{align} 
		To this end, we rewrite the process \eqref{exploratory_dynamics} and the second equation in \eqref{scv1} in the following matrix form:
\begin{equation}
	\label{scv5}
	\d\zeta_\zs{t}= \mathbf{a}(\zeta_\zs{t}^{\lambda},\lambda)\d t + \mathbf{b}(\zt{t}{\lambda},\lambda)\d \gr{W}_\zs{t}, \quad \gr{W}_\zs{t} =\left(\begin{matrix}
		\bar{W}_\zs{t}\\
		W_\zs{t}\\
		\hat{W}_\zs{t}
	\end{matrix}\right),
\end{equation}
where the coefficients $\gr{a}$ and $\gr{b}$ are defined for any state vector $\zeta=(x,y)^\top$ and policy $ \lambda$ as 
\begin{equation*}
	\mathbf{a}(\zeta,\lambda)=\left(\begin{matrix}
		(r+(\mu-r)\text{Mean}(\lambda))x\\
		\varpi(y)
	\end{matrix}\right) \quad \text{and} \quad b(\zeta,\lambda)=\left(\begin{matrix}
		0& \sigma(y)x \text{Mean}(\lambda) & \sigma(y)x\sqrt{\text{Var}(\lambda)} \\  \sqrt{1-\rho^{2}}\gamma(y)& \rho \gamma(y)&0  
	\end{matrix}\right).
\end{equation*}
Now, for any vector  $\mathbf{q}=(\mathbf{q}_\zs{1},
\mathbf{q}_\zs{2})^\top \in \rr{R}{2}$ and for a symmetric matrix $\gr{M}=(M_\zs{i,j})_\zs{1 \le i,j, \le 2}$, we set the Hamiltonian function as
\begin{equation}
	\label{scv2}
	H(J,\zeta,\gr{q},\gr{M}):=\sup_\zs{\lambda \in \cH_\zs{[a,b]}} H_\zs{0}(\lambda,J,\zeta,\gr{q},\gr{M}),
\end{equation}
where
\begin{equation}\label{H_0}
	H_\zs{0}(\lambda,\zeta,J,\gr{q},\gr{M}):=\gr{a}^\top(\zeta,\lambda)\gr{q}+ \frac{1}{2}\mathbf{tr}[\gr{b}\gr{b}^\top(\zeta,u)\gr{M}] +m\cE nt(\lambda)[(1-\varrho)J+1].
\end{equation}
The HJB equation associated with Problem \eqref{goal} is given by
 	\begin{equation}
 		\label{scvh}
 		\begin{cases}
 			\partial_tV(\zeta,t)+H(V,\zeta,V_\zeta(\zeta,t),V_\zs{\zeta,\zeta}(\zeta,t))=0, \quad t\in [0, T),\\[2mm]
 			V(\zeta,T)=\frac{x^{1-\eta}-1}{1-\eta}, \quad \zeta \in \bbr^2\,.
 		\end{cases}
 	\end{equation}
 	Here $V_\zs{\zeta}(\zeta,t)=(V_\zs{x},V_\zs{y})^\top \in \rr{R}{2}$ denotes the gradient vector and 
 	$
 	V_\zs{\zeta\zeta}(\zeta,t)=\left(\begin{matrix}
 		V_\zs{xx}&V_\zs{xy}\\
	V_\zs{xy}&V_\zs{yy}
 	\end{matrix}\right)
 	\,,
 	$ denotes the symmetric Hessian matrix.
 	\noindent
 	Note that in this case, the function $H_\zs{0}$ defined in \eqref{H_0} is of the form
 	\begin{align*}
 		H_\zs{0}(\lambda,J,\zeta,\gr{q},\gr{M})&=\int_{a}^{b}((r+(\mu-r)\pi)xq_\zs{1})\lambda(\pi)\d \pi +\varpi(y)q_\zs{2}+\frac12 \sigma^2(y)\int_{a}^{b}x^2\pi^2\lambda(\pi)\d \pi\gr{M}_\zs{11}\\ 
 		&+\int_{a}^{b}\rho \delta(y)\sigma(y)x\lambda(\pi)\d \pi \gr{M}_\zs{12}+\frac12\delta^2(y) \gr{M}_\zs{22}+m\cE nt(\lambda)[(1-\varrho)J+1]
 	\end{align*}
 and the value function $V(t,x,y;m)$ solves the following HJB PDE 
 \begin{align}\label{HJB-1}
  	 \partial_t V+\underset{\lambda \in \cH_{[a,b]}} {\sup}
   \bigg(\mathcal{A^{\lambda}} V(t,x,y)+ m\cE nt(\lambda_t)[(1-\varrho)V(t,x,y)+1] \bigg)= 0, 
 \end{align}
 with terminal condition $V(T,x,y;m)=\frac{x^{1-\varrho}-1}{1-\varrho} $. The optimal policy $\lambda^{*,[a,b]}$ is characterized by the following: 
\begin{lemma}
 \label{lem_2}
 	In the exploratory constrained setting, the optimal density $\lambda^{*,[a,b]}$ to the HJB equation \eqref{HJB-1} 
 is a normal distribution with mean $\alpha$ and variance $\beta^2$, truncated on the interval $[a,b]$
 where \[
 \alpha(t,x,y)=\dfrac{-(\mu-r)xV_\zs{x}-\rho\gamma(y)\sigma(y)xV_\zs{xy}}{\sigma^2(y)x^2V_\zs{xx}},\quad
 \beta^2(t,x,y)=\dfrac{-m[(1-\varrho)V+1]}{\sigma^2(y)x^2V_\zs{xx}}.
 \]
 Moreover, the optimal density  $\lambda_\zs{t}^{*,[a,b]}$ 
  can be written in the following form:
    \begin{align}\label{optimal-density}
    	\lambda_\zs{t}^{*,[a,b]}(\pi|t,x,y;m) &=\dfrac{\frac{1}{\beta}\frac{1}{\sqrt{2\pi_\zs{e}}}\exp\{-\frac{1}{2}(\frac{\pi-\alpha}{\beta})^{2}\}}{\Phi(\frac{b-\alpha}{\beta})-\Phi(\frac{a-\alpha}{\beta})}
    	=\frac1\beta \frac{\vphi(\frac{\pi-\alpha}{\beta})}{\Phi(\frac{b-\alpha}{\beta})-\Phi(\frac{a-\alpha}{\beta})}.
    \end{align}
    \end{lemma}
    \begin{proof}{
Lemma \ref{lem_2} can be shown using a Lagrange approach with a density constraint. Here we make use of the Donsker-Varadhan variational formula. In particular, note first that $\lambda$ is absolutely continuous with respect to the Lebesgue measure. By setting $\d Q=\lambda(\pi)\d \pi $ and $\d P=\d \pi$, we can write the above supremum in the following form:
    	\begin{align*}
    		\underset {\lambda_t \in \cH_{[a,b]}}{\sup}\bigg\{ \int_{a}^{b}[((r&+(\mu-r)\pi)xV_\zs{x})+\frac12\sigma^2(y)\pi^2x^2V_\zs{xx}+\rho\gamma(y)\sigma(y)xV_\zs{xy}\pi]\d Q\\&-m[(1-\varrho)V+1]\int_{a}^{b}\ln ( \frac{\d Q}{\d P})\d Q
    		\bigg\}= \sup_{Q<<P} E^Q[h]-D_{KL}(Q||P),
    	\end{align*}
			where $h=\exp\big\{\frac{1}{m[(1-\varrho)V+1]}[((\mu-r)xV_\zs{x}+\rho\gamma(y)\sigma(y)x V_\zs{xy})\pi+\frac{1}{2}\sigma^{2}(y)\pi^{2}x^{2}V_\zs{xx}]\big\}$ and $D_\zs{KL}(Q||P)=\int \ln (\frac{\d Q}{\d P})\d Q$ is the Kullback-Leibler divergence. Note that to ensure the integrability of h with respect to $\d Q$, we impose $V_\zs{xx}< 0$. By the Donsker--Varadhan variational formula (see Lemma \ref{lem:DV} in  \ref{sec:App-A}), the supremum is attained at
    	\begin{align}\label{optimal-density_proof}
    		\lambda_\zs{t}^{*,[a,b]}(\pi;x,y,m)&=\dfrac{\exp\{\frac{1}{m[(1-\varrho)V+1]}[((\mu-r)xV_\zs{x}+\rho\gamma(y)\sigma(y)x V_\zs{xy})\pi+\frac{1}{2}\sigma^{2}(y)\pi^{2}x^{2}V_\zs{xx}]\}}{\int_{a}^{b}\exp\{\frac{1}{m[(1-\varrho)V+1]}[((\mu-r)xV_\zs{x}+\rho\gamma(y)\sigma(y)x V_\zs{xy})\pi+\frac{1}{2}\sigma^{2}(y)\pi^{2}x^{2}V_\zs{xx}]\}\d \pi},
    	\end{align}}
        which leads to \eqref{optimal-density}.
    \end{proof}

	{The previous section establishes the formulation of the exploratory control
problem and derives the associated Hamilton--Jacobi--Bellman equation.
We now turn to the analysis of this HJB equation and study the structure of its
solutions.}

\section{Solution of the  HJB equation}
\noindent In this section, we study the HJB equation  \eqref{HJB-1}. To this end, recall first  that the first two moments and the entropy of a truncated Gaussian distribution can be
computed explicitly; see for instance \cite{johnson1972continuous} or \ref{Sec:App-Truncated}. Substituting \eqref{optimal-density} back into the HJB \eqref{HJB-1}, we have :
\begin{align}\label{HJB-substitution}
\partial_t V&+rxV_\zs{x}+((\mu-r)xV_\zs{x}+\rho\gamma(y)\sigma(y)xV_\zs{xy})\alpha+\frac{1}{2}\sigma^{2}(y)x^{2}V_\zs{xx}(\alpha^2+\beta^2)\nonumber \\
&+\varpi(y)V_\zs{y}+\frac{1}{2}\gamma^{2}(y)V_\zs{yy}+\frac{m}{2}\ln(2\pi_\zs{e}e\beta^2)\times[(1-\varrho)V+1]\nonumber \\
&+m\ln Z_\zs{a,b}(x,y;m) \times [(1-\varrho)V+1]=0,\quad V(T,x,y;m)=\frac{x^{1-\varrho}-1}{1-\varrho},
\end{align}
 where $Z_\zs{a,b}(x,y;m):=\Phi(\bar{B}(x,y;m))-\Phi(\bar{A}(x,y;m))$, 
 with
 \begin{align}
  \bar{A}(x,y;m)=(a+\dfrac{(\mu-r)xV_\zs{x}+\rho\gamma(y)\sigma(y)xV_\zs{xy}}{\sigma^2(y)x^2V_\zs{xx}})\sqrt{\frac{\sigma^2(y)x^2V_\zs{xx}}{-m[(1-\varrho)V+1]}}
 \end{align}
 and 
 \begin{align}
 	 \bar{B}(x,y;m)=(b+\dfrac{(\mu-r)xV_\zs{x}+\rho\gamma(y)\sigma(y)xV_\zs{xy}}{\sigma^2(y)x^2V_\zs{xx}})\sqrt{\frac{\sigma^2(y)x^2V_\zs{xx}}{-m[(1-\varrho)V+1]}} .\end{align}
 	 \noindent We find a solution of \eqref{HJB-substitution} with the following ansatz.  
 	 \begin{align}\label {Value_solution}
 	 	V(t,x,y;m)=\frac{x^{1-\varrho}e^{u(t,y;m)}-1}{1-\varrho},
 	 \end{align}
 	 where $u$ satisfies the following PDE
 	  \begin{align}\label{sol_u}
 	 	\partial_t u+&r(1-\varrho)+\frac{1}{2}\gamma^2(y)(u_\zs{yy}+u_\zs{y}^2)+\varpi(y)u_\zs{y} 
 	 	+\frac{1-\varrho}{2\varrho}\left[  \frac{(\mu-r)^2}{\sigma^2(y)}+\frac{2\rho (\mu-r)\gamma(y)}{\sigma(y)}u_\zs{y}+\rho^2\gamma^2(y)u_\zs{y}^2
 	 	\right ]\nonumber \\
 	 	+&\frac{(1-\varrho)m}{2}\left[ \ln\frac{2\pi_\zs{e}m}{\eta\sigma^2(y)}+2\ln Z_\zs{a,b}(y,u_y;m)
 	 	\right]=0,
 	 \end{align} 
 	 with the terminal condition $u(T,y)=0$, where $Z_\zs{a,b}(y,u_y;m)=\Phi(D(y,u_y;m))-\Phi(F(y,u_y;m))$ where 
 	 \begin{align}\label{A_y}
 	 	D(y,u_y;m)=\left(a-\dfrac{(\mu-r)+\rho\delta(y)\sigma(y)u_\zs{y}}{\eta \sigma^2(y)}\right)\sqrt{\frac{\eta\sigma^2(y)}{m}},
 	 \end{align}
 	 \begin{align}\label{B_y}
 	 	F(y,u_y;m)=\left(b-\dfrac{(\mu-r)+\rho\delta(y)\sigma(y)u_\zs{y}}{\eta \sigma^2(y)}\right)\sqrt{\frac{\eta\sigma^2(y)}{m}}.
 	 \end{align}

\noindent We now study the PDE \eqref{sol_u}. 
First, by setting $l(y,t)=u(T-t,y)$, Equation \eqref{sol_u} can be rewritten as
\begin{align}\label{Ap_2}
	\partial_t l&-\frac12\delta^2(y)l_{yy}-\frac12\delta^2(y)l^2_y-\varpi(y)l_y-r(1-\eta) 
	-\frac{1-\eta}{2\eta}\big[\frac{(\mu-r)^2}{\sigma^2(y)}+\frac{2\rho(\mu-r)\delta(y)}{\sigma(y)}l_y+\eta^2\delta^2(y)l^2_y\big]\nonumber \\
	&-\frac{(1-\eta)m}{2}\big[\ln(\frac{2\pi_em}{\eta \sigma^2(y)})+2\ln Z_{a,b}(y,l_y;m)\big]=0,
\end{align}
with initial condition $l(y,0)=0$. Observe that \eqref{Ap_2} is a quasilinear parabolic PDE.  
 {The existence and regularity of solutions to the quasilinear parabolic PDE \eqref{Ap_2} have been intensively investigated in the literature, see e.g., \cite{ladyzhenskaia1968linear} and \ref{Sec:paraPDE}.} Note that  \eqref{Ap_2} takes the form of \eqref{sec:CauchyEq.110} with $n=1$, $c_\zs{1}=\frac{\delta^2(y)p}{2}$ and 
\begin{align*}
	c(y,t,l,p)&=-\frac{1}{2}\delta^2(y)p^2-\varpi(y)p-\frac{1-\eta}{2\eta}\big[\frac{(\mu-r)^2}{\sigma^2(y)}+\frac{2\rho(\mu-r)\delta(y)}{\sigma(y)}p+\eta^2\delta^2(y)p^2\big]\\
	&-\frac{(1-\eta)m}{2}\big[\ln(\frac{2\pi_em}{\eta \sigma^2(y)})+2\ln Z_{a,b}(y,p;m)\big] -r(1-\eta),\\
	C_m(y,t,l,p)&=c(y,t,l,p)-\delta(y)\dot{\delta}(y)p.
\end{align*}
Clearly, 
$$	C_m(y,t,l,0)=-\frac{1-\eta}{2\eta}\frac{(\mu-r)^2}{\sigma^2(y)}-\frac{(1-\eta)m}{2}\big[\ln(\frac{2\pi_em}{\eta \sigma^2(y)})+2\ln Z_{a,b}(y,0;m)\big] -r(1-\eta).$$


The existence of the solution for \eqref{Ap_2} is obtained under the following standard assumption.
\vspace{2mm}

\noindent{
$(\gr{SB})$ {\em \(\delta\) is continuously differentiable and bounded, along with their derivatives. In addition, \(\varpi\) and \(\sigma\) are continuous and differentiable, and
\begin{align}  
	\sigma_{*} = \underset{y \in \mathbb{R}}{\inf} \sigma(y) > 0 \quad \text{and} \quad \underset{y \in \mathbb{R}}{\inf} \delta(y) > 0.  
\end{align}
}
In our exploratory framework with stochastic volatility and portfolio constraints, the solution existence of \eqref{Ap_2} is guaranteed by Theorem \ref{CaPrt00_Theorem} (see Appendix), which requires several technical conditions, among which Condition $\C_\zs{2})$ is verified by Proposition \ref{Prop_4} below.
}
\begin{proposition}[Exploration with portfolio constraints]\label{Prop_4}  Suppose that the following two conditions hold:
	\begin{itemize}
		\item [(i)] $\frac{(\mu-r)^2}{\eta \sigma_*^2}+m\ln\big(\frac{2\pi_\zs{e}m}{\eta \sigma^2_*}\big)+2m\ln\bigg(\Phi\bigg(b\sqrt{\frac{\eta \sigma_*^2}{m}}-\frac{\mu-r}{\sqrt{m\eta\sigma_*^2}}\bigg)-\Phi\bigg(a\sqrt{\frac{\eta \sigma_*^2}{m}}-\frac{\mu-r}{\sqrt{m\eta\sigma_*^2}}\bigg)\bigg)>0$,
		\item [(ii)]$\underset{y\in \bbr}{\sup}\, \sigma^2(y)\le \frac{\q_\zs{0}^m}{\eta} $with $q_\zs{0}^{m}$ given by
		\begin{align}
			q_\zs{0}^m:=\inf\{q>\eta \sigma^2_*: f^{a,b}_m(q)=0\},
		\end{align}
		where 
		\begin{align}\label{Prop_4_eq_1} f^{a,b}_\zs{m}(q)=\frac{(\mu-r)^2}{q}+m\ln\bigg(\frac{2\pi_e m}{q}\bigg)+2m\ln\bigg(\Phi\bigg(b\sqrt{\frac{q}{m}}-\frac{\mu-r}{\sqrt{mq}}\bigg)-\Phi\bigg(a\sqrt{\frac{q}{m}}-\frac{\mu-r}{\sqrt{mq}}\bigg)\bigg).\end{align} 
	\end{itemize}
	Then, under Assumptions $(\gr{SV})$-$(\gr{SB})$, Condition $\C_\zs{2})$ in Theorem \ref{CaPrt00_Theorem} is fulfilled with 
	$$\Psi=\frac{1-\eta}{2\eta}\bigg(\frac{\mu-r}{\sigma_*}\bigg)^2 +m\ln(\frac{2\pi_e m}{\eta \sigma^2_*})+2m\ln\bigg(\Phi\bigg(b\sqrt{\frac{\eta\sigma^2_*}{m}}-\frac{\mu-r}{\sqrt{m\eta\sigma_*^2}}\bigg)-\Phi\bigg(a\sqrt{\frac{\eta\sigma_*^2}{m}}-\frac{\mu-r}{\sqrt{m\eta\sigma_*^2}}\bigg)\bigg)+r(1-\eta).$$
\end{proposition}
\begin{proof}
First, we remark that the function $f^{a,b}_{m}$ is continuous, mapping $[\eta\sigma^2_*, \infty) \to ( -\infty,f^{a,b}_m(\eta\sigma^2_*)]$. From property (i) we have $f^{a,b}_m(\eta\sigma^2_*)>0$. Therefore, by the intermediate value theorem, there exists {$k \ge \eta\sigma^2_*$ such that $f_m^{a,b}(k)=0$. By definition, 	
		}
$f^{a,b}_m(q)\ge 0$ for every $q \in [\eta \sigma^2_*,q_0^m]$. Then condition $\C_{2})$ is ful\-fil\-led with $\Psi=\frac{1-\eta}{2\eta}\bigg(\frac{\mu-r}{\sigma_*}\bigg)^2 +m\ln(\frac{2\pi_e m}{\eta \sigma^2_*})+2m\ln\bigg(\Phi\bigg(b\sqrt{\frac{\eta\sigma^2_*}{m}}-\frac{\mu-r}{\sqrt{m\eta\sigma_*^2}}\bigg)-\Phi\bigg(a\sqrt{\frac{\eta\sigma_*^2}{m}}-\frac{\mu-r}{\sqrt{m\eta\sigma_*^2}}\bigg)\bigg)+r(1-\eta)$. 
\end{proof}

\vspace{2mm}
{Proposition \ref{Prop_4} verifies  that Condition $\C_\zs{2})$ in Theorem \ref{CaPrt00_Theorem} is fulfilled if the (stochastic) volatility $\sigma$ is uniformly bounded from above by $q_\zs{0}^{m}/\eta$ which depends on the exploration parameter $m$. This is not restrictive when exploration is not very rewarding and $m$ is small since we have $q_\zs{0}^{m}\to +\infty$ when $m \to 0$. Verification of Condition $\C_\zs{2})$  and the corresponding upper bound in special cases of Proposition  \ref{Prop_4} are reported in \ref{sec: upperbound}.}

\begin{theorem}\label{thm:HJBEQ}

 
{Under the assumption of Proposition \ref{Prop_4}, there exists a solution  $
 	 	V(t,x,y;m)=\frac{x^{1-\varrho}e^{u(t,y)}-1}{1-\varrho}, $ to the HJB equation \eqref{HJB-substitution}, where   \( u(t, y) \) is a solution to the PDE \eqref{sol_u}.}
   \end{theorem}
 
 \begin{proof}\label{preuve_HJB}
 We find a solution in the form  $V(t,x,y;m)=\frac{{x^{1-\varrho}e^{u(t,y)}}-1}{1-\varrho}$ for some appropriate function $u$ by using Lemma \ref{lem_2}. In particular, direct calculation shows that  the optimal density  is attained with $\lambda^*(\pi|t,y)=\cN(\alpha^*,(\beta^{*})^2)_\zs{[a,b]}$ which does not depend on $x$, where 
  \begin{align}\label{alpha_beta_star}
 	\alpha^*(t,y)= \dfrac{\mu-r+\rho\gamma(y)\sigma(y)u_\zs{y}(t,y)}{\varrho\sigma^2(y)} ,\qquad (\beta^*)^2(t,y)=\frac{m}{\varrho \sigma^2(y)}.
 \end{align} 
 	Now, the HJB equation \eqref{HJB-substitution} boils down to
 	 	\begin{align}
 	\frac{x^{1-\eta}e^{u(t,y)}}{1-\eta}\biggl[\partial_t u+&r(1-\eta)+((\mu-r)(1-\eta)+\rho\delta(y)\sigma(y)u_\zs{y}(1-\eta))\alpha^*(t,y)\nonumber \\
 	&-\frac12 \sigma^2(y)\eta (1-\eta)\bigg(\frac{(\mu-r)^2}{\eta^2 \sigma^4(y)}+2\frac{(\mu-r)\rho \delta(y)u_\zs{y}}{\eta^2\sigma^{3}(y)}+\frac{\rho^2\delta^2(y)u^2_\zs{y}}{\eta^2\sigma^2(y)}+\frac{m^2}{\eta^2\sigma^4(y)}\bigg) \nonumber \\
 	&+\varpi(y)u_\zs{y}+\frac{1}{2}\delta^2(y)(u^2_\zs{y}+u_\zs{yy})+\frac{m(1-\eta)}{2}\left(\ln(\frac{2\pi_\zs{e}e\,m}{\eta\sigma^2(y)})+2\ln Z_\zs{a,b}(y,u_y;m)\right)
 	\biggr]=0. 
 	\end{align}By plugging $\alpha^*(t,y), (\beta^*)^2(t,y)$ from \eqref{alpha_beta_star} into the latest equation, we obtain the PDE \eqref{sol_u}.  

  {The next step is to study solution existence of the PDE \eqref{sol_u}. Equivalently, it suffices to show there exists solutions to the PDE \eqref{Ap_2}. To this end, we verify Conditions $\C_{1})-\C_{5})$ of  Theorem \ref{CaPrt00_Theorem} in \ref{Sec:paraPDE}. First, by Proposition \ref{Prop_4}, we see that condition $\C_{2})$ holds. To verify condition $\C_4)$, let us recall that we have
 	 $$ c_\zs{1}=\frac{\gamma^2(y)}{2}p, \quad \frac{\partial c_1}{\partial l}=0, \quad \frac{\partial c_1}{\partial y}=\gamma(y)\dot{\gamma}(y)p,$$ and 
hence $|c_1|<Kp, |\frac{\partial c_1}{\partial y}|\le Kp$, also $\underset{0\le t\le T,|y|\le N}{\sup}|c(y,t,l,p)|\le K(1+p+p^2)$ and $K$ is a constant in $\bbr$ that depends on $\delta^*$ and $\dot{\gamma}^*$, where
	$$\gamma^*=\underset{0\le t\le T} {\sup}\,\underset{y\in \bbr}{\sup}\, \gamma^2(y),\quad \dot{\gamma}^*=\underset{0\le t\le T}{\sup}\underset{y\in \bbr}{\sup} \,\dot{\gamma}(y).$$  
Therefore, 
 	 \begin{align*}
 	 	\sup_\zs{(y, t)  \leq  \Gamma_\zs{N}}\,\sup_\zs{|l|  \leq  N}\,\sup_\zs{ p \in \bbr}
 	 	\,\frac{ 
 	 		\big( |c_1| + \big| \dfrac{\partial c_1}{\partial l} \big| \big) (1+|p|) +  |\dfrac{\partial c_1}{\partial y}| +|c|}{1+ |p|^2}&=\sup_\zs{ p \in \bbr}\dfrac{Kp(1+|p|)+Kp+K(1+p+p^2)}{1+p^2}<\infty .
 	 \end{align*}
	Moreover, it can be checked directly that Conditions $\C_1)$ and $\C_3) $ hold true, while Condition $\C_5)$ is satisfied for any $ 0 < \varepsilon < 1$. By Theorem \ref{CaPrt00_Theorem} in \ref{Sec:paraPDE} there exists a solution $u(y,t)$ to the PDE \eqref{Ap_2} which is bounded in $\bbr\times[0, T]$ and $u(y,t)$ belongs to $\cH^{2+\varepsilon, 1+\varepsilon/2}(
 	 \Gamma_\zs{N})$ for any $N \ge 1$. }
\end{proof}

{Although the HJB equation has been solved, a verification step is
still required to validate the dynamic programming principle and the optimality
of the associated feedback control.
This step is carried out in the next section.
}
\section{Verification Theorem}
To construct the optimal policy, recall that $u(t,y)$ satisfies the PDE \eqref{sol_u}.
The optimal policy of Lemma~\ref{lem_2} does not depend on wealth,
$\lambda^*(\cdot\,|\,t,x,y)=\lambda^{0,[a,b]}(\cdot\,|\,t,y;m)=\cN\big(\alpha^*(t,y),(\beta^*)^2(t,y)\big)_{[a,b]}$, where
\begin{align}
    \alpha^*(t,y)=\dfrac{\mu-r+\rho\gamma(y)\sigma(y)u_\zs{y}(t,y)}{\varrho\sigma^2(y)},\qquad
    (\beta^*)^2(t,y)=\frac{m}{\varrho\,\sigma^2(y)}.
\end{align}
We denote by $X^{\lambda^*}$ the wealth process \eqref{exploratory_dynamics} generated by
$\lambda^*$. Its coefficients involve only the mean and the variance of
$\lambda^*(\cdot\,|\,t,y_t)$, which are given in \ref{Sec:App-Truncated}. With
$A^*=(a-\alpha^*)/\beta^*$ and $B^*=(b-\alpha^*)/\beta^*$,
\begin{equation}
    \text{Mean}\big(\lambda^*(\cdot\,|\,t,y)\big)=\alpha^{*}+\beta^*\,\dfrac{\vphi(A^*)-\vphi(B^*)}{\Phi(B^*)-\Phi(A^*)}
\end{equation}
and
\begin{equation}
    \text{Var}\big(\lambda^{*}(\cdot\,|\,t,y)\big)=(\beta^*)^2\left[1+\dfrac{A^*\vphi(A^*)-B^*\vphi(B^*)}{\Phi(B^*)-\Phi(A^*)}-\left(\dfrac{\vphi(A^*)-\vphi(B^*)}{\Phi(B^*)-\Phi(A^*)}\right)^2\right],
\end{equation}
where $\alpha^*$, $\beta^*$, $A^*$ and $B^*$ are evaluated at $(t,y)$.
\begin{theorem}\label{Theo_value}
Assume that $u$ satisfies the PDE \eqref{sol_u}.
Then the process $\lambda^*(t,y_t)$ 
is the optimal strategy for Problem \eqref{goal}, i.e.,
\begin{equation*}
\underset{\lambda \in \cH_{[a,b]}}{\sup}J^{\lambda}(t,x,y;m)
=J^{\lambda^*}(t,x,y;m)=V(t,x,y;m).
\end{equation*}
\end{theorem}

\begin{proof}
Define a sequence of stopping times
\[
\tau_n=\inf \Big\{ t\ge 0 :|y_t|\vee (X^{\lambda}_\zs{t})^{1-\eta}\vee e^{M^{\lambda}_\zs{t}}\ge n\Big\}, 
\qquad n\ge 1.
\]
We first prove that the feedback control $\lambda^*(t,y_t)$ belongs to
$\mathcal H_{[a,b]}$. Recall that, for every $(t,y)\in[0,T]\times\mathbb R$,
$\lambda^*(\cdot\mid t,y)$ is a truncated Gaussian density supported on
$[a,b]$, with parameters
\[
\alpha^*(t,y)
=
\frac{\mu-r+\rho\delta(y)\sigma(y)u_y(t,y)}
{\eta\sigma^2(y)},
\qquad
(\beta^*)^2(t,y)
=
\frac{m}{\eta\sigma^2(y)}.
\]

Let
\[
m^*(t,y):=\text{Mean}(\lambda^*(\cdot\mid t,y)),
\qquad
M^*(t,y):=\text{Var}(\lambda^*(\cdot\mid t,y)).
\]
Since $u\in H^{2+\varepsilon,1+\varepsilon/2}(\Gamma_N)$ for every
$N>1$, the function $u_y$ is locally H\"older continuous. Hence,
under Assumption~\textnormal{(SB)}, the mappings $\alpha^*$ and
$\beta^*$, and consequently $m^*$ and $M^*$, are deterministic and
continuous in $(t,y)$, with the regularity required in
Proposition~\ref{Prop1}.

Moreover, since $\lambda^*(\cdot\mid t,y)$ is supported on the compact
interval $[a,b]$, its first and second moments are finite. In particular,
\[
\bigl(m^*(t,y)\bigr)^2+M^*(t,y)
=
\int_a^b \pi^2\lambda^*(\pi\mid t,y)\,d\pi>0,
\]
because $\beta^*(t,y)>0$ and therefore $\lambda^*(\cdot\mid t,y)$ is
non-degenerate. 
All the assumptions of
Proposition~\ref{Prop1} are satisfied. It follows that $\lambda^*\in\mathcal H_{[a,b]}.$ Next, we prove that $V(t,x,y;m)=\frac{x^{1-\varrho}e^{u(t,y)}-1}{1-\varrho}$, where   \( u(t, y) \) satisfies the PDE \eqref{sol_u}, is the optimal value function. Indeed, for any admissible $\lambda$, denote $M^{\lambda}_{t}=\int_{0}^{t}m(1-\varrho)\cE nt(\lambda_\zs{s})\d s$. Applying Itô's formula to $e^{M_{t}}V(t,X^{\lambda}_\zs{t},
	y_\zs{t})$, we get 
	\begin{align}
		&\d \,(e^{M_\zs{t}} V(t,X^{\lambda}_\zs{t},y_\zs{t})) 
		=\biggl\{\partial_t V+[r+(\mu -r)\text{Mean}(\lambda_t)]X^{\lambda}_\zs{t}V_x+\varpi(y_t)V_y+\frac12 \sigma^2(y_t)[\text{Mean}(\lambda_t)^2\nonumber \\
		&+Var(\lambda_t)](X^{\lambda}_\zs{t})^2V_\zs{xx}+\frac12 \gamma^2 (y_t)V_\zs{yy}+\rho\gamma(y_t)\sigma(y_t)\text{Mean}(\lambda_t)X^{\lambda}_\zs{t}V_\zs{xy}\nonumber +m(1-\varrho)\cE nt(\lambda_t)V
		\biggr\}e^{M^{\lambda}_\zs{t}}\d t\nonumber \\
		&+\biggl\{[\sigma(y_\zs{t})\text{Mean}(\lambda_t)X^{\lambda}_\zs{t}V_x+\rho\gamma(y_\zs{t})V_y]\d W_\zs{t}+\rho_\zs{1}\gamma(y_t)V_y \d \bar{W}_\zs{t}+\sqrt{\text{Var}(\lambda_t)}\sigma(y_t)X^{\lambda}_\zs{t}V_x\d \hat{W}_\zs{t}
		\biggr \}e^{M^{\lambda}_\zs{t}}.
	\end{align}
	 By using \eqref{HJB-1}, we have
	\begin{align}
		&e^{M^{\lambda}_\zs{\tau n}}V(\tau_n \wedge T,X^{\lambda}_\zs{\tau_n \wedge T},y_\zs{\tau_n \wedge T})-V(0,X^{\lambda}_\zs{0},y_\zs{0 })
        \le \int_{0}^{\tau_n \wedge T}-m\cE nt(\lambda_s)e^{M^{\lambda}_\zs{s}}\d s\nonumber\\&+\int_{0}^{\tau_n \wedge T}e^{M^{\lambda}_\zs{s}}(X^{\lambda}_\zs{s})^{1-\eta}e^{u(s,y_s)}\biggl\{\big[\sigma(y_s)\text{Mean}(\lambda_s)+\rho\delta(y_s)\frac{u_y(s,y_s)}{1-\eta}\big]\d W_s\nonumber\\
        +&\rho_\zs{1}\delta (y_s)\frac{u_y(s,y_s)}{1-\eta}\d \bar{W}_s+\sqrt{\text{Var}(\lambda_s)}\sigma(y_s)\d \hat{W}_s
		\biggr\},\label{eq: 47}
	\end{align}
	and the equality holds if $\lambda=\lambda^{*}$. The quadratic variation of the stochastic integral above is given by 
	\begin{align}\label{Stochastic_martingale}
		\mathbb{E}\bigg[ \int_{0}^{\tau_n \wedge T}&e^{2M^{\lambda}_s}(X^{\lambda}_\zs{s})^{2(1-\eta)}e^{2u(s,y_s)}\biggl\{\big[\sigma(y_s)\text{Mean}(\lambda_s)+\rho \delta(y_s)\frac{u_y(s,y_s)}{1-\eta}\big]^2 \nonumber \\
		&+(1-\rho^2)\delta^2(y_s)\frac{u^2_y(s,y_s)}{(1-\eta)^2}+\text{Var}(\lambda_s)\sigma^2(y_s)
		\biggr\}\d s
		\bigg].
        \end{align}
On the interval $[0,\tau_n]$, by definition of the stopping time,
\(
|y_s|\vee (X_s^{\lambda})^{1-\eta}\vee e^{M_s^{\lambda}}\le n,
\)
which implies that
\(
e^{2M_s^{\lambda}}\le n^2 \) and \(
(X_s^{\lambda})^{2(1-\eta)}\le n^2.
\)
Moreover, since $u\in H^{2+\varepsilon,\,1+\varepsilon/2}$,
the functions $(t,y)\mapsto u(t,y)$ and $(t,y)\mapsto u_y(t,y)$ are H\"older continuous; hence, they are continuous.
Therefore, they are bounded on the compact set $[0,T]\times[-n,n]$. In particular, there exist constants
\[
U_n:=\sup_{(t,y)\in[0,T]\times[-n,n]} |u(t,y)|<\infty,
\qquad
G_n:=\sup_{(t,y)\in[0,T]\times[-n,n]} |u_y(t,y)|<\infty,
\]
such that for all $0\le s\le \tau_n$,
\(
|u(s,y_s)|\le U_n\) and \(|u_y(s,y_s)|\le G_n.
\)
Consequently, on $[0,\tau_n]$,
\(
e^{2u(s,y_s)}\le e^{2U_n}=:K_n\) and \(
u_y(s,y_s)^2\le G_n^2.
\) 
Using inequality $(a+b)^2\le 2(a^2+b^2)$, we obtain
\[
\big[\sigma(y_s)\text{Mean}(\lambda_s)
+\rho \delta(y_s)\tfrac{u_y(s,y_s)}{1-\eta}\big]^2
\le
2\sigma^2(y_s)\text{Mean}(\lambda_s)^2
+2\rho^2\delta^2(y_s)\tfrac{u_y^2(s,y_s)}{(1-\eta)^2}.
\]
Gathering the remaining terms yields that the terms in brackets of \eqref{Stochastic_martingale} are bounded by
\[
2\sigma^2(y_s)\text{Mean}(\lambda_s)^2
+\delta^2(y_s)\tfrac{u_y^2(s,y_s)}{(1-\eta)^2}
+\sigma^2(y_s)\text{Var}(\lambda_s).
\]
Since $\lambda_s$ has compact support $[a,b]$, we have
$|\text{Mean}(\lambda_s)|< \infty \;\text{and}\;
\text{Var}(\lambda_s)< \infty$.
Furthermore, by continuity, $\sigma$ and $\delta$ are bounded on $\{|y|\le n\}$.
Therefore, there exists a constant $T_n>0$, depending only on $n$, such that for all $s\le\tau_n$, the term inside the integral \eqref{Stochastic_martingale} is bounded by
\[
T_n\Big(\sigma^2(y_s)\big(\text{Mean}(\lambda_s)^2+\text{Var}(\lambda_s)\big)+1\Big),
\]
Therefore, the stochastic integral \eqref{Stochastic_martingale} is bounded by
		$$
        \mathbb{E}\left[\int_{0}^{\tau_n \wedge T}T_n\biggl\{\sigma^2(y_s)[\text{Mean}(\lambda_s)^2+ \text{Var}(\lambda_s)]+1\biggr\}\d s\right]<\infty,$$
	for some constants $T_n>0$. Therefore, the expectations of the corresponding stochastic integrals in \eqref{eq: 47}  are equal to 0.
	For an admissible $\lambda$, by the condition (iii) in Definition \ref{def_1}, we deduce
	\begin{align*}
		&V(0,x_0,y_0)\\
		&\ge \mathbb{E}\biggl\{e^{M^{\lambda}_{\tau_n \wedge T}}V(\tau_n \wedge T,X^{\lambda}_{\tau_n},y_{\tau_n \wedge T})+\int_{0}^{\tau_n \wedge T}m\cE nt(\lambda_s)e^{M^{\lambda}_s}\d s
		\biggr\}\\
		&=\mathbb{E}\left\{e^{M^{\lambda}_{\tau_n \wedge T}}\frac{(X^{\lambda}_{\tau_n \wedge T})^{1-\eta}}{1-\eta}e^{u(\tau_n \wedge T,y_{\tau_n \wedge T})}-\frac{e^{M^{\lambda}_{\tau_n \wedge T}}}{1-\eta}+\frac{e^{M^{\lambda}_{\tau_n \wedge T}}-1}{1-\eta}\right\}\\
		&=\mathbb{E}\left\{e^{M^{\lambda}_{\tau_n \wedge T}}\frac{(X^{\lambda}_{\tau_n \wedge T})^{1-\eta}}{1-\eta}e^{u(\tau_n \wedge T,y_{\tau_n \wedge T})}-\frac{1}{1-\eta}\right\}.
	\end{align*}
	which implies that    
	\begin{align*}
		V(0,x_0,y_0)&\ge \liminf_{n\to\infty} \mathbb{E}\left\{e^{M^{\lambda}_{\tau_n \wedge T}}\frac{(X^{\lambda}_{\tau_n \wedge T})^{1-\eta}}{1-\eta}e^{u(\tau_n \wedge T,y_{\tau_n \wedge T})}-\frac{1}{1-\eta}\right\}.
	\end{align*}
By Fatou's lemma we have 
	\begin{align*}
		V(0,x_0,y_0)&\ge \mathbb{E}\bigg[\frac{e^{M^{\lambda}_{T}}(X^{\lambda}_T)^{1-\eta}-1}{1-\eta}\bigg]
		=J(0,x,y;m),
	\end{align*}
    which concludes the proof.
    	\end{proof}
\begin{remark}
In general, \eqref{sol_u} does not admit a unique solution. However, under the
assumptions of Theorem~\ref{Theo_value} (in particular, the regularity and integrability conditions
ensuring the validity of the verification argument), the solution of \eqref{sol_u} is unique within this verification class. Indeed, let $u^{(1)}$ and $u^{(2)}$ be two functions satisfying \eqref{sol_u} and the assumptions of
Theorem~\ref{Theo_value}. Define for $i=1,2$ the associated candidate value functions
\(
V^{(i)}(t,x,y;m):=\frac{x^{1-\eta}e^{u^{(i)}(t,y)}-1}{1-\eta}.
\)
Applying the verification argument to each $V^{(i)}$ yields, for all $(t,x,y)$,
\[
V^{(i)}(t,x,y;m)=\sup_{\lambda\in\cH_{[a,b]}} J^\lambda(t,x,y;m)=:J^{\lambda^*}(t,x,y;m),
\]
hence $V^{(1)}\equiv V^{(2)}$. Fix any $x>0$. Then
\[
x^{1-\eta}e^{u^{(1)}(t,y)}=(1-\eta)V^{(1)}(t,x,y;m)+1=(1-\eta)V^{(2)}(t,x,y;m)+1=x^{1-\eta}e^{u^{(2)}(t,y)},
\]
which implies $u^{(1)}(t,y)=u^{(2)}(t,y)$ for all $(t,y)$. 
\end{remark}

 \section{Asymptotic Expansion and Derivation of Perturbative PDEs}\label{sec:asymptotic}
 
 \noindent We denote by $u^{(0)}$ the solution to \eqref{sol_u} with $m=0$ (without exploration), and by $V^{(0)}$ the optimal value function for the classical problem (i.e., when $m=0$). It follows from Theorem \ref{thm:HJBEQ} that $V^{(0)}(t,x,y)=\frac{x^{1-\eta}e^{u^{(0)}(t,y)}-1}{1-\eta  }$, where $u^{(0)}$ satisfies
 \begin{align}\label{Expansion_0}
 	&\partial_t u^{(0)}+r(1-\eta)+\frac12\delta^2(y)(u^{(0)}_{yy}+(u^{(0)}_y)^{2})+\varpi(y)u^{(0)}_y\nonumber\\
 	&+\frac{1-\eta}{2\eta}\bigg[\frac{(\mu-r)^2}{\sigma^2(y)}+\frac{2\rho(\mu-r)\delta(y)}{\sigma(y)}u^{(0)}_y+\rho^2\delta^2(y)(u^{(0)}_y)^{2}\bigg]\nonumber \\
 	&+(1-\eta)\varsigma^{(0)}(a,b,u^{(0)}_y)=0,\quad u^{(0)}(T,y)=0,
 \end{align}
 where $\varsigma^{(0)}(a,b,u^{(0)}_y)=\underset{m\to 0}{\lim}m\ln Z_{a,b}(y,u_y;m)$ is defined in Lemma \ref{LEM_Appendix} in the Appendix. 
 
 \noindent Below, we derive an asymptotic expansion of the solution \eqref{sol_u} in a neighborhood of $m=0$. Specifically, we are interested in exploring the behavior of the solution for small values of $m$ to account for the effects of exploration.

\noindent Our asymptotic results involve several functions, among which $u^{(1)}$ and $u^{(2)}$ are, respectively, the solutions of the following PDEs: 
 \begin{align}\label{Expansion_1}
 	&\partial_t u^{(1)}+\frac12 \delta^2(y)(u^{(1)}_{yy}+2u^{(0)}_y u^{(1)}_y)+\varpi(y)u^{(1)}_y+\frac{1-\eta}{2}\ln \frac{2\pi_e}{\eta \sigma^{2}(y)}\nonumber \\
 	&+\frac{1-\eta}{2\eta}\bigg[\frac{2\rho (\mu-r)\delta(y)}{\sigma(y)}u^{(1)}_y+2\rho^{2}\delta^{2}(y)u^{(0)}_yu^{(1)}_y\bigg]=0, \,\, u^{(1)}(T,y)=0,
 \end{align}
 and 
 \begin{align}\label{Expansion_2}
 	&\partial_t u^{(2)}+\frac12\delta^{2}(y)(u^{(2)}_{yy}+(u^{(1)}_y)^{2}+2u^{(0)}_yu^{(1)}_y)+\varpi(y)u^{(2)}_y\nonumber \\
 	&+\frac{1-\eta}{2\eta}\bigg[\frac{2\rho (\mu-r)\delta(y)}{\sigma(y)}u^{(2)}_y+\rho^{2}\delta^2(y)((u^{(1)}_y)^{2}+2u^{(0)}_yu^{(2)}_y)\bigg]=0, \,\, u^{(2)}(T,y)=0.
 	 \end{align}
     The result can be summarized in the following Lemma.
 \begin{lemma}
 	The solution $u$ of \eqref{sol_u} admits the following Taylor expansion with respect to $m$:
 	\begin{align}\label{Lemma_2_representation_sol_u}
 		u(t,y)&=u^{(0)}(t,y) + \frac{(1-\eta)}{2}m\ln m (T-t)+mu^{(1)}(t,y)
 		+m^2u^{(2)}(t,y)+O(m^3),
 	\end{align}
   {where $u^{(1)}$ and $u^{(2)}$ are given by \eqref{Expansion_1}, \eqref{Expansion_2} respectively.}
    \end{lemma}
 	\begin{proof}
 	We first compute the derivatives \eqref{Lemma_2_representation_sol_u} with respect to each argument:
 		\begin{align*}
 			\partial_t u&=\partial_t u^{(0)}-\frac{(1-\eta)}{2}m\ln m+ m \partial_t u^{(1)}+ m^2\partial_t u^{(2)}+\ldots\\
 			u_y&=u^{(0)}_y+mu^{(1)}_y+ m^2u^{(2)}_y+\ldots\\
 			u_{yy}&=u^{(0)}_{yy}+ mu^{(1)}_{yy}+ m^2u^{(2)}_{yy}+\ldots\\
 			u^2_y&=(u^{(0)}_y)^{2}+ m^2(u^{(1)}_y)^{2}+ 2mu^{(0)}_yu^{(1)}_y+2m^2u^{(0)}_yu^{(2)}_y +\ldots \,			 
 		\end{align*} 
 		By substituting these derivatives in \eqref{sol_u} we get:
 		  \begin{align*}
 		  	&\partial_t u^{(0)}-\frac{(1-\eta)}{2}m\ln m+ m\partial_t u^{(1)}+ m^2\partial_t u^{(2)}+ r(1-\eta)+ \frac12\delta^2(y)(u^{(0)}_{yy}+ mu^{(1)}_{yy}+ m^2u^{(2)}_{yy}+(u^{(0)}_y)^{2}\\
 		  	&+ m^2(u^{(1)}_y)^{2}+ 2mu^{(0)}_yu^{(1)}_y+2m^2u^{(0)}_yu^{(2)}_y)+ \varpi(y)(u^{(0)}_y+mu^{(1)}_y+ m^2u^{(2)}_y)\\
 		  	&+\frac{1-\eta}{2\eta}\bigg[\frac{(\mu-r)^2}{\sigma^2(y)}+\frac{2\rho(\mu-r)\delta(y)(u^{(0)}_y+mu^{(1)}_y+ m^2u^{(2)}_y)}{\sigma(y)}+\pc{\rho^{2}}\delta^2(y)((u^{(0)}_y)^{2}+ m^2(u^{(1)}_y)^{2}+ 2mu^{(0)}_yu^{(1)}_y\\
 		  	&+2m^2u^{(0)}_yu^{(2)}_y)\bigg]+\frac{(1-\eta)}{2}m\ln m+ \frac{1-\eta}{2}m\ln \bigg(\frac{2\pi_\zs{e}}{\eta \sigma^2(y)}\bigg)\notag\\
            &+(1-\eta)m\ln Z_{a,b}(y,u^{(0)}_y+mu^{(1)}_y+ m^2u^{(2)}_y;m)=0,
 		  \end{align*}
 		 with $u(T,y)=0.$ 
         Using the expansion
\[
u_y(t,y;m)
=
u_y^{(0)}(t,y)
+
m\,u_y^{(1)}(t,y)
+
m^2\,u_y^{(2)}(t,y)
+
O(m^3),
\]
we define
\(
p_m
:=
u_y^{(0)} + m u_y^{(1)} + m^2 u_y^{(2)} 
\). Obviously, \(
p_m \xrightarrow[m\to 0]{} p_0 := u_y^{(0)}.
\)
Recall that
\[
Z_{a,b}(y,p_m;m)
=
\Phi\!\big(F(y,p_m;m)\big)-\Phi\!\big(D(y,p_m;m)\big),
\]
Applying Lemma \ref{LEM_Appendix}, we obtain
\[
m \ln Z_{a,b}
\!\left(
y,
u_y^{(0)}+m u_y^{(1)}+m^2 u_y^{(2)};
m
\right)
=
\varsigma^{(0)}\!\left(a,b,u_y^{(0)}\right)
+
o(1),
\]
where
\[
\varsigma^{(0)}(a,b,u_y^{(0)})
:=
\begin{cases}
0,
& \pi_0\in[a,b], \\[6pt]
-\tfrac12(a-\pi_0)^2\eta\sigma^2(y),
& \pi_0<a, \\[10pt]
-\tfrac12(b-\pi_0)^2\eta\sigma^2(y),
& \pi_0>b,
\end{cases}
\qquad
\pi_0
=
\frac{\mu-r+\rho\,\delta(y)\sigma(y)\,u_y^{(0)}}{\eta\sigma^2(y)}.
\]
          Now, grouping all the terms of $m^0,\,m,\,m^2$, we obtain \eqref{Expansion_0}, \eqref{Expansion_1}, \eqref{Expansion_2}.
 	\end{proof}
    {
\begin{remark}
The term $\frac{1-\eta}{2}\,m\ln m\,(T-t)$ in the expansion
\eqref{Lemma_2_representation_sol_u} is a purely exploratory effect.
It originates from the entropy-induced variance of the control and reflects the
cost of exploration at the value-function level.
\end{remark}
 \noindent The next theorem shows how the temperature parameter $m$
affects the value function when exploration is removed; i.e., we focus only on exploitation with the mean of the optimal
exploratory policy.
}

\begin{theorem}\label{Theo:3}
Let $\lambda^*(t,y)$ denote the optimal exploratory policy of the entropy--regularized problem with exploration parameter $m$,
and let $\alpha^{(m)}(t,y):=\text{Mean}(\lambda^*(t,y))$.
Then the mean admits the expansion
\begin{align}\label{eq:mean-expansion}
\alpha^{(m)}(t,y)
=
\alpha^*(t,y)
+m\,\frac{\rho\delta(y)}{\eta\sigma(y)}\,u^{(1)}_y(t,y)
+O(m^2),
\end{align}
where
\[
\alpha^*(t,y)
=
\frac{\mu-r+\rho\delta(y)\sigma(y)u^{(0)}_{y}(t,y)}{\eta\sigma^2(y)}
\]
is the optimal policy of the classical (non-exploratory) problem (i.e., $m=0$), and $u^{(1)}$ is the first-order
corrector in the asymptotic expansion of the exploratory value function given in \eqref{Expansion_1}.
{
\noindent Define the deterministic policy obtained by freezing the control at its mean,
\(
\check\lambda^{(m)}(t,y)\equiv \alpha^{(m)}(t,y),
\)
and let $V^{(m)}(t,x,y)$ be the value function of the original (non-regularized) investment problem
evaluated under this policy, i.e.,
\begin{align}\label{def:Vm}
V^{(m)}(t,x,y)
:=
\mathbb E\!\left[
U(X_T^{\check\lambda^{(m)}})\,\big|\,X_t^{\check\lambda^{(m)}}=x,\;y_t=y
\right].
\end{align}
}
Then $V^{(m)}$ admits the representation
\begin{align}\label{eq:Vm-ansatz}
V^{(m)}(t,x,y)
=
\frac{x^{1-\eta}e^{\psi^{(m)}(t,y)}-1}{1-\eta},
\end{align}
where $\psi^{(m)}$ satisfies an HJB equation with zero control variance and a drift depending on $m$
through $\alpha^{(m)}$.
Moreover, $\psi^{(m)}$ admits the expansion
\[
\psi^{(m)}(t,y)
=
u^{(0)}(t,y)
+
m^2k^{(2)}(t,y)
+O(m^3),
\]
where $k^{(2)}$ solves
\begin{align}
\partial_t k^{(2)}
&+
\varpi(y)k^{(2)}_y
+
\frac12\delta^2(y)\big(k^{(2)}_{yy}+2u^{(0)}_y k^{(2)}_y\big)
\nonumber\\
&+
\frac{(1-\eta)\rho\delta(y)}{\eta\sigma(y)}
\big[\mu-r+\rho\delta(y)\sigma(y)u^{(0)}_y\big]k^{(2)}_y
=
0,
\qquad
k^{(2)}(T,y)=0. \label{eq: phi2}
\end{align}
\end{theorem}
\begin{proof}

\textbf{Step 1: Expansion of the mean policy.}
From Theorem~\ref{thm:HJBEQ} and Lemma~\ref{lem_2}, the mean of the optimal exploratory policy
$\alpha^{(m)}(t,y)=\text{Mean}(\lambda^*(t,y))$ admits the expansion
\[
\alpha^{(m)}(t,y)
=
\alpha^*(t,y)
+
m\frac{\rho\delta(y)}{\eta\sigma(y)}u^{(1)}_y(t,y)
+
m^2\frac{\rho\delta(y)}{\eta\sigma(y)}u^{(2)}_y(t,y)
+
O(m^3),
\]
where $u^{(1)}$ and $u^{(2)}$ are the correctors given in
\eqref{Expansion_1}--\eqref{Expansion_2}.
{
These correctors account jointly for the perturbation of the mean control and the contribution of the
exploratory variance.
}
\medskip

\textbf{Step 2: PDE satisfied by $V^{(m)}$.}
Fix the deterministic policy $\check\lambda^{(m)}\equiv\alpha^{(m)}$ and define $V^{(m)}$ by \eqref{def:Vm}.
By the dynamic programming principle {(DPP still holds because we look at the same optimization problem without entropy regularization)}, $V^{(m)}$ satisfies
\begin{align}\label{eq:Vm-PDE}
\partial_t V^{(m)}
&+
\big(r+(\mu-r)\alpha^{(m)}(t,y)\big)xV^{(m)}_x
+
\frac12\sigma^2(y)\alpha^{(m)}(t,y)^2x^2V^{(m)}_{xx}
\nonumber\\
&+
\varpi(y)V^{(m)}_y
+
\frac12\delta^2(y)V^{(m)}_{yy}
+
\rho\delta(y)\sigma(y)\alpha^{(m)}(t,y)xV^{(m)}_{xy}
=
0,
\end{align}
with terminal condition $V^{(m)}(T,x,y)=\frac{x^{1-\eta}-1}{1-\eta}$.
Note that \eqref{eq:Vm-PDE} contains no entropy or variance term: {we look at an exploitation problem but with the mean of the optimal exploratory policy to take into account the effect of the temperature parameter on the value function},
i.e., $m$ enters only through the drift $\alpha^{(m)}$.

\medskip
\textbf{Step 3: Expansion of $\psi^{(m)}$.}
Using the ansatz \eqref{eq:Vm-ansatz}, equation \eqref{eq:Vm-PDE} reduces to a PDE for $\psi^{(m)}$.
Since $\alpha^{(m)}=\alpha^*+O(m)$ and $\psi^{(m)}(T,\cdot)=0$, we have $\psi^{(0)}=u^{(0)}$.
We therefore expand
\[
\psi^{(m)}=u^{(0)}+mk^{(1)}+m^2k^{(2)}+O(m^3),
\]
while keeping the expansion of $\alpha^{(m)}$ obtained in Step~1. Collecting the terms of order $m$ yields: 

	\begin{align}\label{h-i-1}
 			\partial_t k^{(1)}+\frac12\delta^{2}(y)k^{(1)}_{yy}+\bigg[\frac{(1-\eta)\rho\delta(y)(\mu-r)}{\eta\sigma(y)}+\delta^2(y)u^{(0)}_y+\frac{(1-\eta)\rho\delta(y)}{\eta\sigma(y)}u^{(0)}_y\bigg]k^{(1)}_y=0, \,\,k^{(1)}(T,y)=0.
 		\end{align}
 Note that 
 		Theorem \ref{CaPrt00-Theorem-Lin} implies immediately that the equations  \eqref{h-i-1}
 		have the unique solution $k^{(1)}\equiv 0$.
The PDE \eqref{eq: phi2} for $k^{(2)}$ then follows by collecting the terms of order $m^2$.
\end{proof}

{
\begin{remark}
The functions $u^{(1)}, u^{(2)}$ arising from the asymptotic expansion of the \emph{exploratory} HJB equation
solve, respectively, the non-homogeneous linear PDE \eqref{Expansion_1},\eqref{Expansion_2} whose source term originates from
the entropy-induced exploratory variance.
In contrast, $k^{(1)}, k^{(2)}$ are obtained from the HJB equation satisfied by $V^{(m)}$
under the mean policy $\check\lambda^{(m)}$.
As a result, the corresponding first-order equation is homogeneous and admits the trivial solution
$k^{(1)}\equiv 0$.
Therefore, $u^{(1)}, u^{(2)}$ and $k^{(1)}, k^{(2)}$ capture fundamentally different effects and need not be identical. 
\end{remark}
}

 \section{Policy Evaluation and Policy Improvement}
 Since we are operating in an unknown environment, it is necessary to estimate the value function of a given policy using sample data. To achieve this, we adopt the approach from \cite{Jia2022policya,jia2022policyb}, which incorporates the martingale property into the estimation process. For a given policy $\lambda$, we recall that the value function is given by
 	\begin{equation}\label{Policy_value}
 	J^\lambda(t,x,y;m)=\mathbb{E}\left[\int_{t}^{T}m[(1-\varrho)J^\lambda(s,x,y;m)+1]\cE nt(\lambda_\zs{s})\d s+U(X_\zs{T}^{\lambda})|X^{\lambda}_\zs{t}=x,Y_\zs{t}=y
 	\right].
 \end{equation}
 Based on the discussion in Lemma \ref{Lemma1} and Remark \ref{Rem1}, $J^\lambda(t,x,y;m)$ solves the following PDE
 \begin{equation}
 \label{Policy_PDE}
 	\partial_t J + \mathcal{A^{\lambda}} J= 0,
 \end{equation}
  with a terminal condition $J^\lambda(T,x,y;m)=\frac{x^{1-\varrho}-1}{1-\varrho} $.
  \begin{theorem}
  	A function $J^\lambda$ is the value function associated with the policy $\lambda$ if and only if it satisfies the terminal condition $J^\lambda(T,x,y;m)=U(x)$, and for any initial $(t,x,y)\in [0, T)\times \bbr^+\times \bbr$:
  	\begin{align}
  		M_s:=J^\lambda(s,X^{\lambda}_s,y_s;m)+\int_{t}^{s}m[(1-\eta)J^{\lambda}(s',X^{\lambda}_{s'},y_{s'};m)+1]\cE nt(\lambda_{s'})\d s',\quad s\in [t, T]
  	\end{align}
  	is an $\cF_s=\sigma\{X^{\lambda}_u,y_u, \, t\le u\le s\}$ martingale on $[t, T]$.
  \end{theorem}
 \begin{proof}
 	\begin{align}
 		\mathbb{E}[M_T|X^{\lambda}_s,y_s]&=\mathbb{E}\bigg[J^\lambda(T,X^{\lambda}_T,y_T;m)+\int_{t}^{T}m[(1-\eta)J^{\lambda}(s',X^{\lambda}_{s'},y_{s'};m)+1]\cE nt(\lambda_{s'})\d s'|X^{\lambda}_s,y_s\bigg]\nonumber\\
 		&=\mathbb{E}\bigg[J^\lambda(T,X^{\lambda}_T,y_T;m)+\int_{s}^{T}m[(1-\eta)J^{\lambda}(s',X^{\lambda}_{s'},y_{s'};m)+1]\cE nt(\lambda_{s'})\d s'|X^{\lambda}_s,y_s\bigg]\nonumber \\
 		&+\int_{t}^{s}m[(1-\eta)J^{\lambda}(s',X^{\lambda}_{s'},y_{s'};m)+1]\cE nt(\lambda_{s'})\d s'\nonumber\\
 		&=M_s,
 	\end{align}
    which concludes the proof.
  \end{proof}
 
 \noindent By the martingale property of $M$, for any process $g$ satisfying $\mathbb{E}[\int_{0}^{T}g_s^2\d\langle M \rangle_s]<\infty$, we have  $\mathbb{E}\int_{0}^{T}g_s\d M_s=0$. Equivalently,
 \begin{align}
 	\mathbb{E}\bigg[\int_{0}^{T}g_t\bigg(\d J^\lambda(t,X^{\lambda}_t,y_t;m)+m[(1-\eta)J^\lambda(t,X^{\lambda}_{t},y_{t};m)+1]\cE nt(\lambda_{t})\,\d t\bigg)\bigg]=0.
 \end{align}
 Such a process $g$ is called a test function. For a given $m$, let $J^{\theta}(t,X^{\lambda}_t,y_t;\lambda)$ be a parameterized family that is used to approximate $J^\lambda$, where $\theta \in \Theta \subset \bbr ^n, n\ge 1$. Our goal now is to find the best $\theta^*$. As seen above, the process 
 \begin{align*}
 	M^{\theta^*}_t:=J^{\theta^*}(t,X^{\lambda}_t,y_t;\lambda)+\int_{0}^{t}m[(1-\eta)J^{\theta^*}(s',X^{\lambda}_{s'},y_{s'};\lambda)+1]\cE nt(\lambda_{s'})\d s',
 \end{align*}
 is a martingale in $[0, T]$. Moreover, a fundamental property of the conditional expectation implies that \( M_t \) is the orthogonal projection of \( M_T \) onto the space of all \( \mathcal{F}_t \)-measurable random variables. This means that \( M_t \) minimizes the \( L^2 \)-error between \( M_T \) and any \( \mathcal{F}_t \)-measurable random variable; in particular,
 \(
 M^{\theta}_t = \mathbb{E}[M^{\theta}_T \mid \mathcal{F}_t] = \text{argmin}_{\xi \text{ is } \mathcal{F}_t\text{-measurable}} \mathbb{E}\left[|M^{\theta}_T - \xi|^2\right]\) for  $t \in [0, T].$ Our objective is therefore to minimize the martingale loss function $\mathbb{E}[\int_0^T(M^{\theta}_T-M^{\theta}_t)^2\d t]$ over admissible $\theta\in\Theta$. In other words, the agent considers the following minimization 
 \begin{align}
 	\underset{\theta\in \Theta}{\min}\mathbb{E}\bigg[\int_{0}^{T}\bigg(U(X^{\lambda}_T)-J^{\theta}(t,X^{\lambda}_t,y_t;\lambda)+\int_{t}^{T}m[(1-\eta)J^{\theta}(s',X^{\lambda}_{s'},y_{s'})+1]\cE nt(\lambda_{s'})\d s')\bigg)^2\d t\bigg].
 \end{align}
This approach requires the entire trajectory to predict the value function, making it an offline learning procedure. In such an offline learning framework, the estimate is updated only after a complete trajectory is observed.   In contrast, in an online learning framework, the value function is updated continuously as time progresses, rather than waiting for the entire trajectory.  The martingale orthogonality property now reads: \begin{align}\label{orthonal_conditions}
	\mathbb{E}\bigg[\int_{0}^{T}g_t\bigg(\d J^{\theta}(t,X^{\lambda}_t,y_t;\lambda)+m[(1-\eta)J^{\theta}(t,X^{\lambda}_{t},y_{t};\lambda)+1]\cE nt(\lambda_{t})\,\d t\bigg)\bigg]=0.
\end{align}
A common choice of test process is $g_t=\frac{\partial}{\partial \theta}J^{\theta}(t,X^{\lambda}_t,y_t;\lambda)\in \bbr^n$. Remark that to fully determine $\theta$, we need  at least $n$ equations in \eqref{orthonal_conditions}. We make the following assumptions about these functional approximators.
From here on we fix the convention used throughout the paper and its appendices: $\psi$ denotes the parameters of the critic, that is of the approximate value function, which we write $V^{\psi}$, while $\theta$ denotes the parameters of the actor, that is of the exploratory policy $\lambda^{\theta}$. The family $J^{\theta}$ above is the critic evaluated along a fixed policy, so that $J^{\theta}$ and $V^{\psi}$ denote the same object under the two parametrizations.
\begin{Assumption}
	For all $\psi \in \Psi \subset \bbr^n$, $V^{\psi}, \frac{\partial V^{\psi}}{\partial \psi}\in C^{1,2,2}([0,T)\times \bbr^+\times \bbr)\cap C([0, T]\times\bbr^+\times \bbr)$ and satisfy the polynomial growth conditions in $x$ and $y$.
	\end{Assumption}
\noindent We begin with an arbitrary admissible policy $\lambda(\cdot|t,x,y)$, associated with the value function $V^{\lambda}$, which satisfies the following PDE: 
	\begin{align*}
	 \partial _tV^{\lambda} + \mathcal{A^{\lambda}} V^{\lambda}+ m\cE nt(\lambda_t)[(1-\varrho)V^{\lambda}+1] = 0,
	\end{align*}
	with terminal condition $V^{\lambda}(T,x,y;m)=\frac{x^{1-\varrho}-1}{1-\varrho} $. As shown in Theorem \ref{thm:HJBEQ}, we can represent $V^{\lambda}(t,x,y;m)=\frac{x^{1-\eta}e^{u^{\lambda}(t,y)}-1}{1-\eta}$, where $u^{\lambda}$ satisfies the following PDE:
	\begin{align}\label{eq:u_lambda}
		\partial_t u^{\lambda}&+(1-\eta)\int_{a}^{b}(r+(\mu-r)\pi)\lambda(\pi)\d \pi+\varpi(y)u^{\lambda}_y-\frac{\eta(1-\eta)}{2}\sigma^2(y)\int_{a}^{b}\pi^2\lambda(\pi)\d \pi\nonumber \\ &+\rho\delta(y)\sigma(y)(1-\eta)u^{\lambda}_y\int_{a}^{b}\pi\lambda(\pi)d\pi+\frac12\delta^2(y)(u^{\lambda}_{yy}+u^{\lambda^2}_y)+m(1-\eta)\cE nt(\lambda)=0, \quad u(T,y)=0.
	\end{align}
    Inspired by Theorem \ref{thm:HJBEQ}, the policy is now updated to
	$\hat{\lambda}(\pi|t,y)=\cN(\hat{\alpha},(\hat{\beta})^2)|_\zs{[a,b]}$, where	\begin{align}\label{policy_eval_improve}
		\hat{\alpha}(t,y)= \dfrac{\mu-r+\rho\gamma(y)\sigma(y)u^{{\lambda}}_\zs{y}(t,y)}{\varrho\sigma^2(y)} ,\qquad (\hat{\beta})^2(t,y)=\frac{m}{\varrho \sigma^2(y)}.
	\end{align}
\begin{theorem}\label{Theo:5}
Let $\lambda\in\mathcal H_{[a,b]}$ be an admissible policy such that
\[
V^\lambda\in C^{1,2,2}([0,T)\times \bbr^+\times \bbr)\cap C([0, T]\times\bbr^+\times \bbr).
\]
Let $\hat{\lambda}(\pi|t,y)$ be the updated admissible policy with
parameters given in \eqref{policy_eval_improve}, and let
$V^{\hat{\lambda}}$ be the cor\-res\-ponding value function. Then
\begin{align}\label{THeo5}
V^{\hat{\lambda}}(t,x,y;m)\ge
V^\lambda(t,x,y;m),
\qquad \forall
(t,x,y)\in[0,T]\times\mathbb R_+\times\mathbb R.
\end{align}
\end{theorem}

		\begin{proof}
			Recall first that $V^{\lambda}(t,x,y;m)$ solves the following PDE:
			\begin{align*}
			 \partial_t V^{\lambda} + \mathcal{A^{\lambda}} V^{\lambda} + m\cE nt(\lambda_t)[(1-\varrho)V^{\lambda}+1]= 0,
			\end{align*}
			 with terminal condition $V^{\lambda}(T,x,y;m)=\frac{x^{1-\varrho}-1}{1-\varrho} $. It follows that 
			 	\begin{align*}
			 	\partial_tV^{\lambda}(t,x,y;m)+\underset {\tilde{\lambda} \in \cH_{[a,b]}}{\sup}\{\mathcal{A^{\tilde{\lambda}}}V^{\lambda}+ m\cE nt(\tilde{\lambda}_t)[(1-\varrho)V^{\lambda}+1]\}\ge0,
			 \end{align*}
			Similarly to Lemma \ref{lem_2},
            the above supremum is attained at  the updated policy $\hat{\lambda}$ defined by \eqref{policy_eval_improve}. In other words,
		\begin{align}\label{policy_improve_1}
			\partial_t V^{\lambda}(t,x,y;m)+\mathcal{A^{\hat{\lambda}}}V^{\lambda}+ m\cE nt(\hat{\lambda}_t)[(1-\varrho)V^{\lambda}+1]\ge0,
		\end{align}
		Now for the policy $\hat{\lambda}$, the corresponding value function is given by
		\begin{align}
			V^{\hat{\lambda}}(t,x,y;m)=\E \bigg[\int_{t}^{T}e^{\int_{t}^{s}m(1-\eta)\cE nt(\hat{\lambda}_u)\d u}m\cE nt(\hat{\lambda}_s)\d s+e^{\int_{t}^{T}m(1-\eta)\cE nt(\hat{\lambda}_u)\d u}U(X^{\hat{\lambda}}_T)| \cF_\zs{t}\bigg].
		\end{align}
		We define $M^{\hat{\lambda}}_t=\int_{0}^{t}m(1-\eta)\cE nt(\hat{\lambda}_s)\d s$. By Itô's formula we obtain 
		\begin{align}\label{lemm_ito_6.12}
			&e^{M^{\hat{\lambda}}_s}V^{\lambda}(s,X^{\hat{\lambda}}_s,y_s)=e^{M^{\hat{\lambda}}_0}V^{\lambda}(0,X^{\hat{\lambda}}_0,y_0)
			+\int_{0}^{s}\bigg\{ V^{\lambda}_u+\big[r+(\mu-r)\text{Mean}(\hat{\lambda}_u)\big]X_u^{\hat{\lambda}}V^{\lambda}_x+\varpi(y_u)V^{\lambda}_y\nonumber \\
			&+\frac{1}{2}\sigma^2(y_u)[\text{Mean}(\hat{\lambda}_u)^2+var(\hat{\lambda}_u)](X^{\hat{\lambda}}_u)^2V^{\lambda}_{xx}+\frac12\delta^2(y_u)V^{\lambda}_{yy}+\rho\delta(y_u)\sigma(y_u)\text{Mean}(\hat{\lambda}_u)X^{\hat{\lambda}}_uV^{\lambda}_{xy}\nonumber \\
			&+m(1-\eta)\cE nt(\hat{\lambda}_u)V^{\lambda}\bigg\}e^{M^{\hat{\lambda}}_u}\d u 
            +\bigg\{[\sigma(y_u)\text{Mean}(\hat{\lambda}_u)X^{\hat{\lambda}}_uV^{\lambda}_x+\rho\delta(y_u)V^{\lambda}_y]\d W_u\nonumber\\
            &+\sqrt{1-\rho^2}\delta(y_u)V^{\lambda}_y\d \bar{W}_u
			+\sqrt{\text{Var}(\hat{\lambda}_u)}\sigma(y_u)X^{\hat{\lambda}}_uV^{\lambda}_x\d \hat{W}_u\bigg\}e^{M^{\hat{\lambda}}_u}.
		\end{align}
		Let $\{\tau_n\}$ be a sequence of stopping times $\tau_n=\inf\{r\ge t\ : (X^{\hat{\lambda}}_r)^{1-\eta}\vee|y_r|\vee e^{M^{\hat{\lambda}}_r}\ge n\}\wedge T$. Clearly $\lim_{n\to \infty}\tau_n=T$. Then by \eqref{policy_improve_1} we have

	\begin{align}\label{al_}
			\partial_tV^{\lambda}(t,x,y;m)+\mathcal{A^{\hat{\lambda}}}V^{\lambda}_t
			&+m\cE nt(\hat{\lambda}_s)[(1-\varrho)V^{\lambda}]\ge -m\cE nt(\hat{\lambda}_s),
		\end{align}
		taking the expectation of \eqref{lemm_ito_6.12}, using the inequality \eqref{al_} and taking $s=T\wedge\tau_n$ and sending $n\to \infty$, we obtain
		\begin{align*}
			e^{M^{\hat{\lambda}}_0}V^{\lambda}(0,X^{\hat{\lambda}}_0,y_0)&\le \mathbb{E}\bigg[\int_{0}^{T}e^{M^{\hat{\lambda}}_u}m\cE nt(\hat{\lambda}_u)\d u+e^{M^{\hat{\lambda}}_T}U(X^{\hat{\lambda}}_T)| \cF_\zs{t}\bigg],\nonumber			
		\end{align*}
        hence $V^{\lambda}(0,x,y)\le V^{\hat{\lambda}}(0,x,y).$ It can be done similarly for $t\in(0,T]$. 
		\end{proof}
        
        \noindent 
        Although Theorem \ref{Theo:5} holds for arbitrary admissible policies, convergence towards
the optimal policy cannot be expected without additional structure.
In particular, the choice of the parametric form of the policy plays a crucial
role.
The next theorem shows that the iteration converges to the optimal policy when the
initial policy has the same parametric form as the optimal one.
\begin{theorem}\label{convergence}
 Let $\lambda^{(0)}(\pi;t,y;m)=\cN\left(\frac{\kappa}{\sigma^2(y)},\frac{\chi^2}{\sigma^2(y)}\right)_{\zs{[a,b]}}$ with some real number $\kappa$, $\chi>0$. Define the sequence of feedback policies $(\lambda^{(n)}(\pi;t,y;m))$ updated by the policy improvement scheme \eqref{policy_eval_improve}, i.e.,
        \begin{align}
            \lambda^{(n)}(\pi;t,y;m)=\cN\bigg(\pi\bigg|\dfrac{\mu-r+\rho\gamma(y)\sigma(y)u^{\lambda^{(n-1)}}_\zs{y}(t,y)}{\varrho\sigma^2(y)} , \frac{m}{\varrho \sigma^2(y)}\bigg)_\zs{[a,b]}, \quad n=1,2\ldots,
        \end{align}
        where $u^{\lambda^{(n-1)}}$ is associated to the value function $V^{\lambda^{(n-1)}}$ corresponding to the policy $\lambda^{(n-1)} $ defined by
        \begin{align*}\label{Value_function_lambda_n-1}
                V^{\lambda^{(n-1)}}(t,x,y;m)=\E \bigg[\int_{t}^{T}e^{\int_{t}^{s}m(1-\eta)\cE nt(\lambda^{(n-1)}_u)\d u}m\cE nt(\lambda^{(n-1)}_s)\d s+e^{\int_{t}^{T}m(1-\eta)\cE nt(\lambda^{(n-1)}_u)\d u}U(X^{\lambda^{(n-1)}}_T)| \cF_\zs{t}\bigg].
        \end{align*}
         Then,
    $$
        \lim_{n\to \infty} \lambda^{(n)}(\cdot,t,y,m)=\lambda^*(\cdot,t,y,m).
$$
    \end{theorem}
        \begin{proof}
        Observe first that $V^{\lambda^{(0)}}$ solves the following PDE
        \begin{align}
            \partial_t V^{\lambda^{(0)}}+\mathcal{A}^{\lambda^{(0)}}V^{\lambda^{(0)}}
            +m\cE nt(\lambda^{(0)})[(1-\varrho)V^{\lambda^{(0)}}+1]=0.
        \end{align}
        Solving the above PDE with terminal conditions $V^{\lambda^{(0)}}(T,x,y;m)=\frac{x^{1-\eta}-1}{1-\eta}$ and $\lambda^{(0)}(\pi;t,y;m)=\cN(\frac{\kappa}{\sigma^2(y)},\frac{\chi^2}{\sigma^2(y)})_\zs{[a,b]}$, we can show that $V^{\lambda^{(0)}}(t,x,y;m)=\frac{x^{1-\eta}e^{u^{\lambda^{(0)}}}-1}{1-\eta}$, where $u^{\lambda^{(0)}}$ satisfies the following PDE:  \begin{align}\label{u_lambda_0}
            u^{\lambda^{(0)}}_t&+\frac12\delta^2(y)(u^{\lambda^{(0)}}_{yy}+u^{\lambda^{(0)^2}}_y)+\varpi(y)u^{\lambda^{(0)}}_y+(1-\eta)\rho\delta(y)(\frac{\kappa}{\sigma(y)}+\chi\frac{\varphi(A)-\varphi(B)}{Z})u^{\lambda^{(0)}}_y+r(1-\eta)\nonumber \\
            &+(\mu-r)(\frac{\kappa}{\sigma^2(y)}+\frac{\chi}{\sigma(y)}\frac{\vphi(A)-\vphi(B)}{Z})(1-\eta)\nonumber \\
            &-\frac{\eta(1-\eta)}{2}(\frac{\kappa^2}{\sigma^2(y)}+\chi^2-\chi^2\frac{B\vphi(B)-A\vphi(A)}{Z}+2\frac{\kappa\chi}{\sigma(y)}\frac{\vphi(A)-\vphi(B)}{Z})\nonumber\\
            &+m(1-\eta)[\ln \sqrt{2\pi_e e}\frac{\chi}{\sigma(y)}Z+\frac{A\vphi(A)-B\vphi(B)}{2Z}]=0,\quad u^{\lambda^{(0)}}(T,y)=0,
        \end{align}
        where $A=\dfrac{a-\frac{\kappa}{\sigma^2(y)}}{\frac{\chi}{\sigma(y)}},\, B=\dfrac{b-\frac{\kappa}{\sigma^2(y)}}{\frac{\chi}{\sigma(y)}}$, $Z=\Phi(B)-\Phi(A)$. Now we need to study the PDE \eqref{u_lambda_0}. As before, we set $l(y,t)=u(T-t,y)$; we obtain the following PDE:
        \begin{align}\label{l_lambda_0}
            \begin{cases}
                l_t-\frac{\delta^2(y)}{2}l_{yy}+G(y,t,l,p)=0\\
                l(y,0)=0,
            \end{cases}
        \end{align}
        where
        \begin{align}G(y,t,l,l_y)&=-\frac12\delta^2(y)l^2_y-\Big(\varpi(y)+(1-\eta)\rho\delta(y)\big(\frac{\kappa}{\sigma(y)}+\chi\frac{\varphi(A)-\varphi(B)}{Z}\big)\Big)l_y-r(1-\eta)\nonumber\\
        &-(\mu-r)(\frac{\kappa}{\sigma^2(y)}+\frac{\chi}{\sigma(y)}\frac{\vphi(A)-\vphi(B)}{Z})(1-\eta)
        +\frac{\eta(1-\eta)}{2}(\frac{\kappa^2}{\sigma^2(y)}+\nonumber\\
        &\chi^2-\chi^2\frac{B\vphi(B)-A\vphi(A)}{Z}+2\frac{\kappa\chi}{\sigma(y)}\frac{\vphi(A)-\vphi(B)}{Z})\nonumber\\
        &-m(1-\eta)[\ln \sqrt{2\pi_e e}\frac{\chi}{\sigma(y)}Z+\frac{A\vphi(A)-B\vphi(B)}{2Z}]. \end{align}
        We have that $C(y, t, l, p)=G(y,t,l,p)-\delta(y)\dot{\delta}(y)p$
        and $C(y,t,l,0)=G(y,t,l,0)$. Since $G(y,t,l,0)$ is bounded, it follows that
            \(
        g_* = \sup_{0 \leq t \leq T} \sup_{y \in \bbr } |G(y,t,l,0)| < \infty.
        \) This implies that Condition \( \C_2) \) in Theorem \ref{CaPrt00_Theorem} is satisfied with \( \psi = g_* \), and all other conditions of the theorem are also fulfilled. Hence, by Theorem \ref{CaPrt00_Theorem}, \( l(y,t) \) or \( u^{\lambda^{(0)}}  \in \cH^{2+\epsilon,1+\epsilon/2}(\Gamma_N)\) for any $N>1$ . Therefore, the next  truncated Gaussian policy is fully defined and given by
        \begin{align}
            \lambda^{(1)}(\pi;t,y;m)=\cN\bigg(\pi\bigg|\dfrac{\mu-r+\rho\gamma(y)\sigma(y)u^{\lambda^{(0)}}_\zs{y}(t,y)}{\varrho\sigma^2(y)} , \frac{m}{\varrho \sigma^2(y)}\bigg)_\zs{[a,b]},
        \end{align}
        $V^{\lambda^{(1)}}(t,x,y;m)\ge V^{\lambda^{(0)}}(t,x,y;m)$ where $V^{\lambda^{(1)}}$ is the value function corresponding to this new policy $\lambda^{(1)}(\pi,t,y;m)$. Again $V^{\lambda^{(1)}}(t,x,y;m)$ solves the following PDE
        \begin{align}
            \partial_t V^{\lambda^{(1)}}+\mathcal{A}V^{\lambda^{(1)}}
            +m\cE nt(\lambda^{(1)})[(1-\varrho)V^{\lambda^{(1)}}+1]=0.
        \end{align}
        Direct computations related to the distribution of $\lambda^{(1)}$ lead to the following PDE
         \begin{align}\label{PDE_V_lamba_1}
            &\partial_t V^{\lambda^{(1)}}+rxV^{\lambda^{(1)}}_x+(\mu-r)xV^{\lambda^{(1)}}_x(\alpha_1+\beta\frac{\phi(\bar{A}(x,y;m))-\phi(\bar{B}(x,y;m))}{Z_{a,b}(x,y;m)})+\frac12\sigma^2(y)x^2V^{\lambda^{(1)}}_{xx}(\alpha_1^2+\beta^2\nonumber \\
            &-\beta^2\frac{\bar{B}(x,y;m)\phi(\bar{B}(x,y;m))-\bar{A}(x,y;m)\phi(\bar{A}(x,y;m))}{Z_{a,b}(x,y;m)}
            +2\alpha_1\beta\frac{\phi(\bar{A}(x,y;m))-\phi(\bar{B}(x,y;m))}{Z_{a,b}(x,y;m)})+\varpi(y)V^{\lambda^{(1)}}_y\nonumber \\        &+\frac12\delta^2(y)V^{\lambda^{(1)}}_{yy}+\rho\delta(y)\sigma(y)xV^{\lambda^{(1)}}_{xy}(\alpha_1+\beta\frac{\phi(\bar{A}(x,y;m))-\phi(\bar{B}(x,y;m))}{Z_{a,b}(x,y;m)})\nonumber \\
            &+\bigg(m\ln(\sqrt{2\pi_e e\beta^2}Z_{a,b}(x,y;m))+m\frac{\bar{A}(x,y;m)\phi(\bar{A}(x,y;m))-\bar{B}(x,y;m)\phi(\bar{B}(x,y;m))}{2Z_{a,b}(x,y;m)}\bigg)((1-\eta)V^{\lambda^{(1)}}+1)=0,
        \end{align}
        where $Z_{a,b}(x,y;m)=\Phi(\bar{B}(x,y;m))-\Phi(\bar{A}(x,y;m))$, $\bar{B}(x,y;m)=\frac{b-\alpha_1}{\beta}$, $\bar{A}(x,y;m)=\frac{a-\alpha_1}{\beta}$,
        with $V^{\lambda^{(1)}}(T,x,y)=\frac{x^{1-\eta}-1}{1-\eta}$, where $\alpha_1=\dfrac{\mu-r+\rho\gamma(y)\sigma(y)u^{\lambda^{(0)}}_\zs{y}(t,y)}{\varrho\sigma^2(y)}$ and $\beta^2=\frac{m}{\eta\sigma^2(y)}$. We find a solution of PDE \eqref{PDE_V_lamba_1} of the form  $V^{\lambda^{(1)}}(t,x,y;m)=\frac{x^{1-\eta}e^{u^{\lambda^{(1)}}}-1}{1-\eta}$, where $u^{\lambda^{(1)}}$ satisfies the following PDE:

        \begin{align}\label{related_PDE}
            \partial_t u^{\lambda^{(1)}}&+r(1-\eta)+\varpi(y)u^{\lambda^{(1)}}_y+\frac12\delta^2(y)(u^{\lambda^{(1)^2}}_y+u^{\lambda^{(1)}}_{yy})
            +(1-\eta)\frac{m}{2}\ln(\frac{2\pi_e m}{\eta\sigma^2(y)})+(1-\eta)m \ln Z_{a,b}(x,y;m) \nonumber\\
            &+\frac{1-\eta}{2\eta}\bigg(\frac{(\mu-r)^2}{\sigma^2(y)}+\frac{2\rho\delta(y)(\mu-r)u^{\lambda^{(1)}}_y}{\sigma(y)}-\rho^2\delta^2(y)u^{\lambda^{(0)^2}}_y
            +2\rho^2\delta^2(y)u^{\lambda^{(0)}}_yu^{\lambda^{(1)}}_y\bigg)\notag\\
            &+{\sqrt{\frac{m}{\eta\sigma^2(y)}}}(\frac{\vphi(\bar{A}(y,u^{\lambda^{(0)}}_y;m))-\vphi(\bar{B}(y,u^{\lambda^{(0)}}_y;m))}{Z_{a,b}(y,u^{\lambda^{(0)}}_y;m)})(1-\eta)\rho\delta(y)\sigma(y)(u^{\lambda^{(1)}}_y-u^{\lambda^{(0)}}_y)=0,
        \end{align}
with terminal condition $\,\, u^{\lambda^{(1)}}(T,y)=0.$ Using the same arguments as in the proof of Theorem \ref{thm:HJBEQ}, we have $G(y,t,l,l_y)=-\frac12\delta^2(y)l_y^2-\Big(\varpi(y)+\frac{1-\eta}{\eta}\big(\frac{\rho\delta(y)(\mu-r)}{\sigma(y)}+\rho^2\delta^2(y)u_y^{\lambda^{(0)}}\big)\Big)l_y-\frac{1-\eta}{2}m\ln(\frac{2\pi_e m}{\eta\sigma^2(y)})-\frac{1-\eta}{2\eta}\frac{(\mu-r)^2}{\sigma^2(y)}-r(1-\eta)+\frac{(1-\eta)}{2\eta}\rho^2\delta^2(y)(u_y^{\lambda^{(0)}})^2-(1-\eta)m\ln Z_{a,b}(x,y;m)-\sqrt{\frac{m}{\eta}}(\frac{\vphi(\bar{A}(x,y;m))-\vphi(\bar{B}(x,y;m))}{Z_{a,b}(x,y;m)})(1-\eta)\rho\delta(y)(l_y-u^{\lambda^{(0)}}_y)$. Moreover, $C(y, t, l, p)=G(y,t,l,p)-\delta(y)\dot{\delta}(y)p$ and $C(y,t,l,0)=G(y,t,l,0)$. Using the fact that $G(y,t,l,0)$ is bounded, one can show that \eqref{related_PDE} admits a  solution \(u^{\lambda^{(1)}}\), which is bounded along with its first and second derivatives with respect to $y$ and belongs to any cylinder $\Gamma_N$ for any $N>1$.

 The next-step truncated Gaussian policy is given by
\begin{align}
    \lambda^{(2)}(\pi;t,y;m)=\cN\bigg(\pi\bigg|\dfrac{\mu-r+\rho\gamma(y)\sigma(y)u^{\lambda^{(1)}}_\zs{y}(t,y)}{\varrho\sigma^2(y)} , \frac{m}{\varrho \sigma^2(y)}\bigg)_\zs{[a,b]},
\end{align}
It is clear that
$V^{\lambda^{(2)}}(t,x,y;m)\ge V^{\lambda^{(1)}}(t,x,y;m)$ where $V^{\lambda^{(2)}}$ is the value function corresponding to this new policy and $V^{\lambda^{(2)}}(t,x,y;m)=\frac{x^{1-\eta}e^{u^{\lambda^{(2)}}}-1}{1-\eta}$, where $u^{\lambda^{(2)}}$ satisfies the following PDE:

\begin{align}\label{related_PDE_2}
\partial_t u^{\lambda^{(2)}}&+r(1-\eta)+\varpi(y)u^{\lambda^{(2)}}_y+\frac12\delta^2(y)(u^{\lambda^{(2)^2}}_y+u^{\lambda^{(2)}}_{yy})+(1-\eta)\frac{m}{2}\ln(\frac{2\pi_e m}{\eta\sigma^2(y)})+(1-\eta)m \ln Z_{a,b}(x,y;m)\notag\\
&+\frac{1-\eta}{2\eta}\bigg(\frac{(\mu-r)^2}{\sigma^2(y)}+\frac{2\rho\delta(y)(\mu-r)u^{\lambda^{(2)}}_y}{\sigma(y)}-\rho^2\delta^2(y)u^{\lambda^{(1)^2}}_y+2\rho^2\delta^2(y)u^{\lambda^{(1)}}_yu^{\lambda^{(2)}}_y\bigg)\nonumber\\
&+\beta(\frac{\phi(\bar{A}(x,y;m))-\phi(\bar{B}(x,y;m))}{Z_{a,b}(x,y;m)})(1-\eta)\rho\delta(y)\sigma(y)(u^{\lambda^{(2)}}_y-u^{\lambda^{(1)}}_y)=0,\,\, u^{\lambda^{(2)}}(T,y)=0.
\end{align}
With the same arguments we have that $u^{\lambda^{(2)}}$ is a  solution of PDE \eqref{related_PDE_2} with $u^{\lambda^{(2)}}_y$ bounded. We can continue to update the policy to the step $n$ and we obtain that $u^{\lambda^{(n)}}$ satisfies the following PDE:
\begin{align}\label{related_PDE_n}
\partial_t u^{\lambda^{(n)}}&+r(1-\eta)+\varpi(y)u^{\lambda^{(n)}}_y+\frac12\delta^2(y)(u^{\lambda^{(n)^2}}_y+u^{\lambda^{(n)}}_{yy})\notag\\
&+\frac{1-\eta}{2\eta}\bigg(\frac{(\mu-r)^2}{\sigma^2(y)}+\frac{2\rho\delta(y)(\mu-r)u^{\lambda^{(n)}}_y}{\sigma(y)}-\rho^2\delta^2(y)u^{\lambda^{{(n-1)}^2}}_y+2\rho^2\delta^2(y)u^{\lambda^{(n-1)}}_yu^{\lambda^{(n)}}_y\bigg)\notag\\
&
+(1-\eta)\frac{m}{2}\ln(\frac{2\pi_e m}{\eta\sigma^2(y)})+(1-\eta)m \ln Z_{a,b}(x,y;m) \nonumber\\
&+\beta(\frac{\phi(\bar{A}(x,y;m))-\phi(\bar{B}(x,y;m))}{Z_{a,b}(x,y;m)})(1-\eta)\rho\delta(y)\sigma(y)(u^{\lambda^{(n)}}_y-u^{\lambda^{(n-1)}}_y)=0,\notag\\ &u^{\lambda^{(n)}}(T,y)=0.
\end{align}
With the same arguments we have that $u^{\lambda^{(n)}}$ is a solution of PDE \eqref{related_PDE_n} and belongs to $\cH^{2+\epsilon,1+\epsilon/2}(\Gamma_N)$ for any $N>1$, and $\|u^{\lambda^{(n)}}\|_\zs{2+\epsilon,1+\epsilon/2,\Gamma_N}\le C_N$ with $C_N$ independent of $n$. Now, by employing the usual diagonal process (see \ref{Diag_process}), we can extract from $u^{\lambda^{(n)}}$ a subsequence $u^{\lambda^{(n_k)}}$ that converges together with the derivatives $u^{\lambda{(n_k)}}_y, u^{\lambda{(n_k)}}_{yy}, u^{\lambda_{(n_k)}}_{t}$ at each point of $\bbr\times [0, T]$ to some function $u$ and its corresponding derivatives. It is clear that $u(y,t)$ belongs to $\cH^{2+\epsilon,1+\epsilon/2}(\Gamma_N)$ in each finite cylinder $\Gamma_N$ and satisfies PDE \eqref{related_PDE_n} with the same terminal condition, i.e.,
\begin{align}\label{related_PDE_n0}
\partial_t u&+r(1-\eta)+\varpi(y)u_y+\frac12\delta^2(y)(u^{2}_y+u_{yy})+\frac{1-\eta}{2\eta}\bigg(\frac{(\mu-r)^2}{\sigma^2(y)}+\frac{2\rho\delta(y)(\mu-r)u_y}{\sigma(y)}
+\rho^2\delta^2(y)u_y^2\bigg)\nonumber\\
&+(1-\eta)\frac{m}{2}\ln(\frac{2\pi_e m}{\eta\sigma^2(y)})+(1-\eta)m \ln Z_{a,b}(y,u_y;m)=0,\quad
u(T,y)=0.
\end{align}
By unicity within the verification class, we conclude that $u(t,y)=u^{\lambda^*}(t,y)$. Since $u^{\lambda^{(n)}}$ is nondecreasing in $n$, the sequence converges to $u^{\lambda^*}$, and hence $\lambda^{(n)}\to\lambda^*$.
 \end{proof}

\begin{theorem}\label{thm:PG-BSDE}
	Let $\lambda^\theta$ be an admissible feedback policy and 
	$V^\psi(t,x,y;\lambda^\theta)$ be the (approximate) value function associated with
	$\lambda^\theta$. Define the policy gradient
	\[
	h(t,x,y;\lambda^\theta):=\frac{\partial}{\partial\theta}
	V^\psi(t,x,y;\lambda^\theta),
	\qquad (t,x,y)\in[0,T]\times\R_+\times\R ,
	\]
	and set
	\[
	V^\psi(t):=V^\psi\!\big(t,X_t^{\lambda^\theta},y_t;\lambda^\theta\big),
	\qquad 
	\cE nt(\lambda_t^\theta):=-\widetilde{\E}\!\left[\log\lambda^\theta(\pi_t\mid t,y_t)\,\Big|\,\cF_t\right],
	\]
    where \(\widetilde{\mathbb E}\) is the expectation
with respect to the enlarged probability space supporting the
state dynamics and the randomization \footnote{For any integrable measurable function \(g\),
\[
\widetilde{\mathbb E}\!\left[g(\pi_t)\mid \cF_t\right]
=
\int_a^b g(\pi)\,\lambda^\theta(\pi\mid t,y_t)\,d\pi.
\]}.
	Assume that for the policy $\lambda^\theta$, the wealth--value pair
	$(X^{\lambda^\theta},V^\psi)$ satisfies the BSDE
	\begin{equation}\label{eq:V-BSDE}
		\d V^\psi(t)
		=
		-\,m\big[(1-\eta)V^\psi(t)+1\big]\cE nt\, \d t
		+Z_t^\psi\cdot \d \gr{W}_t,
		\qquad V_T^\psi = U(X_T^{\lambda^\theta}),
	\end{equation}
	for some progressively measurable process $Z^\psi$. Assume moreover that all objects are
	differentiable in $\theta$, that differentiation and conditional expectation may be interchanged,
	and that
	\[
	\widetilde{\E}\!\left[\int_0^T
	\Big\|
	\frac{\partial}{\partial\theta}\log\lambda^\theta(\pi_t\mid t,y_t)\,Z_t^\psi
	\Big\|^2\,\d t\right]<\infty.
	\]
	Then, $h(0,x,y;\lambda^\theta)$ admits the representation
	\begin{align}
		h(0,x,y;\lambda^\theta)
		=&\widetilde{\mathbb E}\Bigg[
		\int_0^T
		\frac{\partial}{\partial\theta}
		\log\lambda^\theta(\pi_t\mid t,y_t)\,
	 \d V^\psi(t,X_t^{\lambda^\theta},y_t;\lambda^\theta)
\label{eq:h-final}\\
		&\qquad
		+\int_0^T
		m\big[(1-\eta)V^\psi(t,X_t^{\lambda^\theta},y_t;\lambda^\theta)+1\big]
		\frac{\partial}{\partial\theta}\cE nt(\lambda_t^\theta)\,\d t
		\nonumber\\
		&\qquad
		+\int_0^T
		m(1-\eta)
		\Big(
		\cE nt(\lambda_t^\theta)\,
		h(t,X_t^{\lambda^\theta},y_t;\lambda^\theta)
		\Big)\d t
		\Bigg].
		\nonumber
	\end{align}
\end{theorem}

\medskip

\begin{proof}
For notational simplicity, we write
\[
V(t):=V^\psi(t,x,y;\lambda^\theta),
\qquad
h(t,x,y;\lambda^\theta):=\partial_\theta V^\psi(t,x,y;\lambda^\theta).
\]
Recall that $V$ solves
\begin{equation}
\partial_t V+\mathcal A^{\lambda^\theta}V
+m\,\mathcal E nt(\lambda_t^\theta)\big[(1-\eta)V+1\big]=0,
\qquad
V(T,x,y)=U(x),
\label{eq:V-PDE}
\end{equation}
where
\[
\mathcal A^{\lambda^\theta}V
=
\int_a^b \mathcal A^\pi V\,\lambda^\theta(\pi\mid t,y)\, \d\pi.
\]
We write \eqref{eq:V-PDE} in the following form:
\begin{equation}\label{eq:V-PDE2}
    \int_{a}^{b}\left(\partial_t V+\cA^{\pi}V\right)\lambda^{\theta}(\pi|t,y)\,\d\pi
    +m\,\mathcal E nt(\lambda_t^\theta)\big[(1-\eta)V+1\big]=0.
\end{equation}
Differentiating \eqref{eq:V-PDE2} with respect to $\theta$, while keeping $(t,x,y)$ fixed, gives
\begin{align}
h_t
&+\mathcal A^{\lambda^\theta}h+m\big[(1-\eta)V+1\big]\partial_\theta\mathcal E nt(\lambda_t^\theta)
+m(1-\eta)\mathcal E nt(\lambda_t^\theta)\,h\notag\\
&
+\int_a^b \left(\partial_t V+
\mathcal A^\pi V\right)\,
\lambda^\theta(\pi\mid t,y)\,
\partial_\theta\log\lambda^\theta(\pi\mid t,y)\,\d\pi
=0.
\label{eq:h-PDE}
\end{align}
We obtain, by using the Duhamel formula applied to \eqref{eq:h-PDE}, 
\begin{align}\label{eq:h-FK-2}
h(0,x,y;\lambda^\theta)
=
\mathbb E\Bigg[
&\int_0^T
\int_a^b
\left(\partial_t V+A^\pi V^\psi(t,X_t^{\lambda^\theta},y_t;\lambda^\theta)\right)\,
\lambda^\theta(\pi\mid t,y_t)\,
\partial_\theta\log\lambda^\theta(\pi\mid t,y_t)\,\d\pi\,\d t
\nonumber\\
&\quad+
\int_0^T
m\big[(1-\eta)V^\psi(t)+1\big]\partial_\theta\mathcal E nt(\lambda_t^\theta)\,\d t
\nonumber\\
&\quad+
\int_0^T
m(1-\eta)\mathcal E nt(\lambda_t^\theta)\,
h(t,X_t^{\lambda^\theta},y_t;\lambda^\theta)\,\d t
\Bigg];
\end{align}
The above representation corresponds to a mild 
formulation associated with the Markov transition semigroup
induced by the controlled diffusion $(X^{\lambda^{\theta}}_t,y_t)$.
Such formulations are standard in the theory of semilinear evolution equations;
see, e.g., Chapter~6 of \cite{pazy1983applications}.
We may rewrite the first term in \eqref{eq:h-FK-2} as
\[
\widetilde{\mathbb E}\Bigg[
\int_0^T
\left(
\partial_t V+\mathcal A^{\pi_t}V^\psi(t,X_t^{\lambda^\theta},y_t;\lambda^\theta)
\right)
\partial_\theta\log\lambda^\theta(\pi_t\mid t,y_t)\,\d t
\Bigg].
\]
 Hence,
\eqref{eq:h-FK-2} becomes
\begin{align*}
h(0,x,y;\lambda^\theta)
=
\widetilde{\mathbb E}\Bigg[
&\int_0^T
\left(
\partial_t V+\mathcal A^{\pi_t}V^\psi(t,X_t^{\lambda^\theta},y_t;\lambda^\theta)
\right)
\partial_\theta\log\lambda^\theta(\pi_t\mid t,y_t)\,\d t
\\
&\quad+
\int_0^T
m\big[(1-\eta)V^\psi(t)+1\big]\partial_\theta\mathcal E nt(\lambda_t^\theta)\,\d t
\\
&\quad+
\int_0^T
m(1-\eta)\mathcal E nt(\lambda_t^\theta)\,
h(t,X_t^{\lambda^\theta},y_t;\lambda^\theta)\,\d t
\Bigg].
\end{align*}
Finally, applying It\^o's formula to
\(
V^\psi(t,X_t^{\lambda^\theta},y_t;\lambda^\theta)
\),
we obtain
\[
\d V^\psi(t)
=
\Big(
\partial_t V^{\psi}(t)
+
\mathcal A^{\pi_t}V^\psi(t,X_t^{\lambda^\theta},y_t;\lambda^\theta)
\Big)\,\d t
+Z_t^\psi\cdot \d W_t.
\]
Hence
\[
\left(
\partial_t V^{\psi}(t)
+
\mathcal A^{\pi_t}V^\psi(t,X_t^{\lambda^\theta},y_t;\lambda^\theta)
\right)\,\d t
=
\d V_t^\psi-Z_t^\psi\cdot \d W_t.
\]
Substituting this identity into the previous expression, we get
\begin{align*}
h(0,x,y;\lambda^\theta)
=
\widetilde{\mathbb E}\Bigg[
&\int_0^T
\partial_\theta\log\lambda^\theta(\pi_t\mid t,y_t)\,\d V_t^\psi
-
\int_0^T
\partial_\theta\log\lambda^\theta(\pi_t\mid t,y_t)\,Z_t^\psi\cdot \d W_t
\\
&\quad+
\int_0^T
m\big[(1-\eta)V^\psi(t)+1\big]\partial_\theta\mathcal E nt(\lambda_t^\theta)\,\d t
\\
&\quad+
\int_0^T
m(1-\eta)\mathcal E nt(\lambda_t^\theta)\,
h(t,X_t^{\lambda^\theta},y_t;\lambda^\theta)\,\d t
\Bigg].
\end{align*}
By the square-integrability assumption, the stochastic integral above is a martingale with zero expectation, and \eqref{eq:h-final} follows.
\end{proof}

Theorem~\ref{thm:PG-BSDE} provides a representation of the policy gradient associated with the parametric exploratory family $\{\lambda^\theta\}_{\theta\in \Theta}$. This representation immediately yields a first-order condition for any interior maximizer of the value function with respect to the policy parameter.

\begin{corollary}
Assume that there exists $\theta^\ast$ in the interior of the admissible parameter set $\Theta$ such that
\[
V^\psi(0,x,y;\lambda^{\theta^\ast})
=
\max_{\theta} V^\psi(0,x,y;\lambda^\theta).
\]
If the mapping $\theta \mapsto V^\psi(0,x,y;\lambda^\theta)$ is differentiable at $\theta^\ast$, then
\[
h(0,x,y;\lambda^{\theta^\ast})
=
\partial_\theta V^\psi(0,x,y;\lambda^{\theta^\ast})
=
0.
\]
\end{corollary}

\begin{remark}
According to the verification result established above, the optimal exploratory policy $\lambda^\ast$ maximizes the value function $V$. Therefore, if $\lambda^\ast=\lambda^{\theta^\ast}$ belongs to the parametric family $\{\lambda^\theta\}_\theta$ and corresponds to an interior parameter $\theta^\ast$, the first-order condition of the corollary is satisfied by $\lambda^\ast$.
\end{remark}

\section{Numerical experiments} \label{sec: num}
In this section, we numerically explore the learning using the policy-improvement scheme of Theorem~\ref{Theo:5}. The closed-form solution suggests a suitable parameterization detailed in \ref{sec:NumVeri} for the value function and the policy. Our theoretical results also allow us to design an actor-critic algorithm (see Algorithm \ref{alg:AC-compact}). 
\subsection{Setup}
\label{sec:num-market}
We use the specification
\begin{equation}\label{eq:num-spec}
\sigma(y)=\sqrt{y^{2}+\sigma_*^{2}},\qquad
\mu(t,y)-r=k_{2}\sqrt{y+\delta_*}\,\sigma(y),\qquad
\delta(y)=k_{1}\sqrt{y+\delta_*},\qquad
\varpi(y)=c\,(\bar y-y),
\end{equation}
where $\sigma_*>0$ keeps volatility away from zero and $\delta_*>0$ keeps the square
root well defined. The market price of risk is $(\mu-r)/\sigma=k_2\sqrt{y+\delta_*}$
and the factor is a mean-reverting square-root process, so \eqref{eq:num-spec} is
consistent with the Heston-type and $3/2$-type models used in the literature. The
functions $\varpi$ and $\delta$ satisfy $(\gr{SV})$ and $(\gr{SB})$, and
\ref{sec:App-num} verifies that the parameter values below meet the conditions of
Proposition~\ref{Prop_4}. Note that $y_t$ is observable, because $\sigma(\cdot)$ is a strictly increasing
bijection of $[0,\infty)$ onto $[\sigma_*,\infty)$, so the state is recovered from any
volatility proxy through $y_t=\sqrt{(\sigma_t^{\mathrm{obs}})^{2}-\sigma_*^{2}}$. Below, we fix
\begin{equation}\label{eq:num-par}
k_1=0.381,\quad k_2=1.235,\quad c=2.446,\quad \bar y=0.143,\quad
\delta_*=0.079,\quad \sigma_*=0.054,\quad \rho=-0.779,\quad r=2\%,
\end{equation}
with $\eta=0.7$, $T=1$ year and a daily time step $\Delta t=1/252$. The model parameters are calibrated using a Bayesian approach based on the Euler-Maruyama approximation of the joint transition density of the asset returns and the volatility factor. Following \cite{dai2023learning}, we adopt diffuse priors and sample from posterior distributions using Hamiltonian Monte Carlo. The calibration uses daily S\&P 500 total returns from 1988 to 2010, ensuring that the estimation sample precedes the backtesting period considered later in Section \ref{sec:num-backtest}.  
Among the estimated parameters, $k_2$ is least precisely identified by the data, which is consistent with the well-known difficulty of estimating drift parameters in continuous-time diffusion models  \cite{merton1980estimating}. While higher-frequency observations improve the estimation of diffusion coefficients, they provide relatively little information about the drift parameters. 
We take $a=0$ throughout and report the learning result for both
$b=1$  and $b=3$. 
\subsection{Policy improvement and learning}
\label{sec:num-learn}

\begin{figure}[h]
	\centering
	\includegraphics[width=\linewidth]{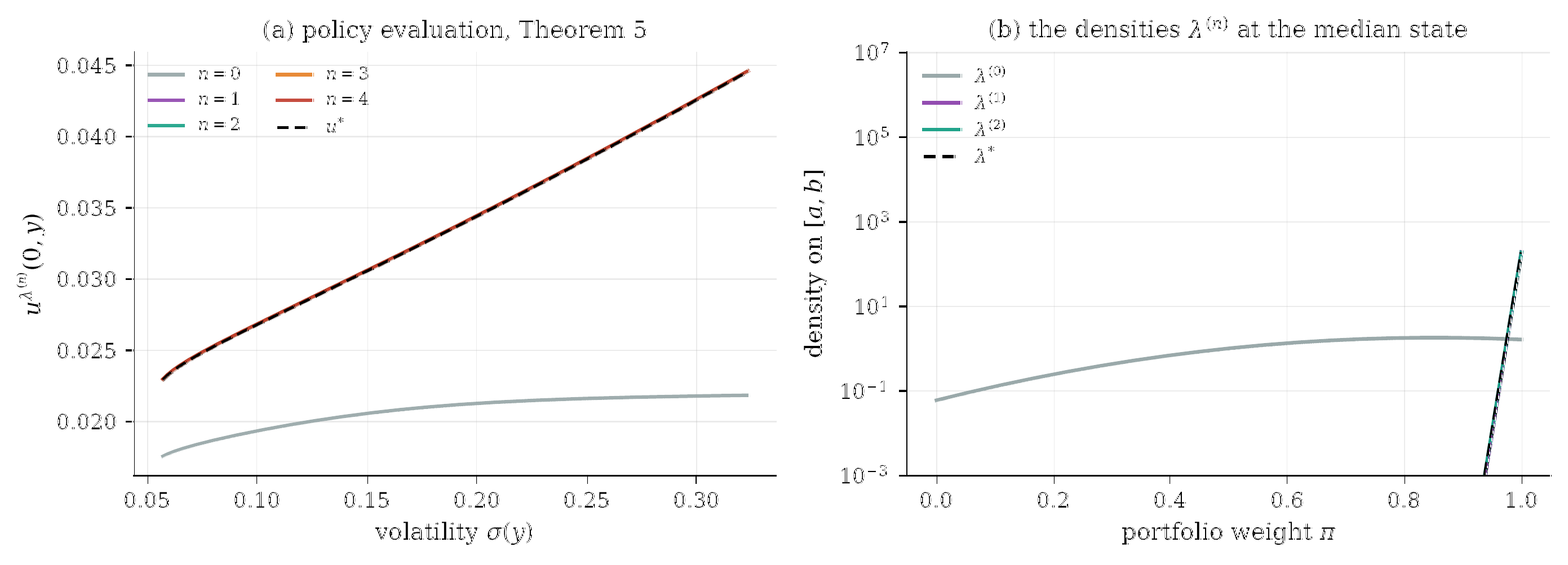}
	\caption{The model-based policy-improvement sequence, started from the initial
	policy prescribed by the convergence theorem. Panel (b) shows the densities
	$\lambda^{(n)}(\cdot\mid 0,\bar y)$ on a logarithmic scale, those for $n\ge1$ being
	indistinguishable from $\lambda^{*}$.}
	\label{fig:improve}
\end{figure}

\begin{table}[htbp]
\centering
\small
\caption{The policy-improvement sequence of the convergence theorem, $\eta=0.7$, $a=0$, $b=1$, $\epsilon=0.15$, started from $\lambda^{(0)}=\mathcal N(\kappa/\sigma^{2}(y),\chi^{2}/\sigma^{2}(y))_{[a,b]}$ with $\kappa=0.02$ and $\chi^2=0.05$.} 

\label{tab:improve}
\begin{tabular}{lrrr}
\hline
$n$ & $u^{\lambda^{(n)}}(0,\bar y)$ & $\|u^{\lambda^{(n)}}-u^{*}\|_\infty$ & $\|\mathbb{E}[\lambda^{(n)}]-\mathbb{E}[\lambda^{*}]\|_\infty$ \\
\hline
0 & 0.020627 & 2.27e-02 & 7.74e-01 \\
1 & 0.030826 & 5.58e-05 & 1.29e-04 \\
2 & 0.030826 & 5.58e-05 & 1.30e-08 \\
3 & 0.030826 & 5.58e-05 & 1.29e-08 \\
4 & 0.030826 & 5.58e-05 & 1.29e-08 \\
\hline
\end{tabular}
\end{table}

Figure~\ref{fig:improve} reports the policy-improvement sequence of
Theorem~\ref{Theo:5}. Given $\lambda^{(n)}$, we solve the PDE \eqref{eq:u_lambda} for
$u^{\lambda^{(n)}}$ and use it in \eqref{policy_eval_improve} to obtain the improved
policy $\lambda^{(n+1)}$, and so on. 
The initial policy is chosen by
$\lambda^{(0)}=\cN(\kappa/\sigma^{2}(y),\chi^{2}/\sigma^{2}(y))_{[a,b]}$ with $\kappa=0.02$ and $\chi^2=0.05$, 
which is far
from optimal. 
Panel (a) of Figure~\ref{fig:improve} 
shows that the value increases at every step, which is consistent with Theorem~\ref{Theo:5} while Panel
(b) shows convergence of 
densities. At the median volatility state, $\lambda^{(0)}$ is a broad density spread over the whole of $[a,b]$, while
$\lambda^{(1)}$, $\lambda^{(2)}$ and $\lambda^{*}$ are indistinguishable. Convergence of the policy sequence is reported in 
Table~\ref{tab:improve}. 

We turn to the learning procedure with Algorithm~\ref{alg:AC-compact} which observes realized wealth increments and updates the
critic along the martingale-orthogonality direction and the actor along the policy
gradient.
\begin{algorithm}[h]
	\caption{Entropy-regularized actor--critic learning}
	\label{alg:AC-compact}
	\begin{algorithmic}[1]
		\State \textbf{Initialize} critic parameters $\psi^{(0)}$, actor parameters $\theta^{(0)}$.
		\For{episode $j=0,\dots,J_{\max}-1$}
		\State \textbf{Simulate $N$ trajectories.} For $k=0,\dots,K-1$ and $i=1,\dots,N$:
		\State \quad sample $v_{t_{k+1},i}=y_{t_{k+1},i}+\delta_*$ exactly from the
		non-central chi-square transition law of the CIR factor, and set
		\[
		\d U_{k,i}
		=\frac{v_{t_{k+1},i}-\E[v_{t_{k+1},i}\mid v_{t_k,i}]}
		{\sqrt{\text{Var}[v_{t_{k+1},i}\mid v_{t_k,i}]}}\sqrt{\Delta t},
		\qquad
		\d W_{k,i}=\rho\,\d U_{k,i}+\sqrt{1-\rho^{2}}\,\d W^{\perp}_{k,i};
		\]
		\State \quad draw the action
		$\pi_{k,i}\sim\lambda^{\theta^{(j)}}(\cdot\mid t_k,y_{t_k,i})$;
		\State \quad propagate wealth with \emph{that} action, coefficients evaluated
		at $y_{t_k,i}$,
		\[
		\Delta\log X^{\lambda}_{k,i}
		=\Big(r+(\mu-r)\pi_{k,i}-\tfrac12\sigma^{2}\pi_{k,i}^{2}\Big)\Delta t
		+\sigma\,\pi_{k,i}\,\d W_{k,i}.
		\]

		\State For all $(k,i)$, compute
		\[
		\Delta V_{k,i}
		=
		V^{\psi^{(j)}}(t_{k+1},X^\lambda_{t_{k+1},i},y_{t_{k+1},i})
		-
		V^{\psi^{(j)}}(t_k,X^\lambda_{t_k,i},y_{t_k,i}),
		\]
		the entropy $\cE nt_{k,i}$, its gradient $\nabla_\theta \cE nt_{k,i}$ and the
		score $\nabla_\theta \log f_{\theta^{(j)}}(\pi_{k,i}\mid t_k,y_{t_k,i})$, all
		evaluated through $\log Z_{a,b}$.

		\State Critic gradient accumulator:
		\[
		H_\psi
		=
		\sum_{k,i}
		\partial_\psi V^{\psi^{(j)}}(t_k,X^\lambda_{t_k,i},y_{t_k,i})\,
        \big(		\Delta V_{k,i}
		+
		\Delta t\,m\big[(1-\eta)V^{\psi^{(j)}}(t_k,X^\lambda_{t_k,i},y_{t_k,i})+1\big]\cE nt_{k,i}\big).
		\]

		\State Actor gradient accumulator, with $\bar\Delta_{k,i}$ the leave-one-out
		batch mean of $\Delta V_{k,\cdot}$ as control variate:
		\[
		H_\theta
		=
		\sum_{k,i}
		\left[
		\nabla_\theta \log f_{\theta^{(j)}}(\pi_{k,i}\mid t_k,y_{t_k,i})\,
		\big(\Delta V_{k,i}-\bar\Delta_{k,i}\big)
		+
		m\,\Delta t\,
		\big((1-\eta)V^{\psi^{(j)}}(t_k,X^\lambda_{t_k,i},y_{t_k,i})+1\big)\,
		\nabla_\theta \cE nt_{k,i}
		\right].
		\]

		\State Learning rate $\ell(j)$ (e.g., $\ell(j)=j^{-0.51}$).

		\State Update:
		\[
		\psi^{(j+1)}=\psi^{(j)}+\ell(j)\frac{H_\psi}{N},\qquad
		\theta^{(j+1)}=\theta^{(j)}+\ell(j)\frac{H_\theta}{N}.
		\]

		\EndFor
		\State \textbf{Return} $(\psi^{(J)},\theta^{(J)})$.
	\end{algorithmic}
\end{algorithm}
The algorithm is implemented on a uniform daily grid over a one-year horizon,
$\Delta t=1/252$, with batches of $N=96$ trajectories per episode and the step
$\ell(j)=j^{-0.51}$ and we report the final iteration. The critic and the actor are
parametrized as in \ref{sec:NumVeri}, following
the form of the solution under the ansatz \eqref{Value_solution}. The temperature is
reported through $\epsilon:=\sqrt{m/(\eta\sigma^{2}(\bar y))}$ which is the standard deviation
of the randomization. 
\begin{table}[htbp]
\centering
\small
\caption{Learning curves of Algorithm~\ref{alg:AC-compact} for $\eta=0.7$ and $[a,b]=[0,1]$.}
\label{tab:learning}
\begin{tabular}{lrr}
\hline
episode & $\theta_4$ & $\|\mathbb{E}[\lambda^{\theta}]-\mathbb{E}[\lambda^{*}]\|_\infty$ \\
\hline
1 & 0.2900 & 0.3292 \\
201 & 0.4523 & 0.0426 \\
1001 & 0.6309 & 0.0202 \\
3501 & 0.7692 & 0.0097 \\
10000 & 1.0227 & 0.0034 \\
\hline
\end{tabular}
\end{table}
Table~\ref{tab:learning} reports the learned values of the market price of risk parameter $\theta_4$ and the mean learned policy against the number of episodes, confirming the convergence of the learned policy to the true one. 


\subsection{Comparison with standard rules on simulated paths}
\label{sec:num-compare}

Table~\ref{tab:horse} compares the exploratory policy and the learned policy with the alternatives that an investor would otherwise use. Following \cite{dai2023learning}, we evaluate each rule by three criteria. The first is the average realized utility $\mathbb{E}[U]$.  The second is the misallocation $\frac1N\sum_{i=1}^{N}\sum_k|\pi^{*}(t_k,y^i_k)-\pi^i_k|^2\Delta t$, with $\Delta t=1/252$, where $\pi^{*}$ is the constrained optimum and $\pi^i_k$ is the allocation of the rule on path $i$ at date $t_k$. Thus, the misallocation measures the $L^{2}$ distance between the rule's allocation and the optimal one along the same state paths. The third is the equivalent relative wealth loss, namely the fraction of initial wealth that an investor holding the constrained optimum could give up and still reach the same expected utility as the rule. All rules are applied to the same $100{,}000$ simulated one-year trajectories.
\begin{table}[htbp]
\centering
\small
\caption{Comparison of investment rules on simulated paths for the action sets $[0,1]$ and $[0,3]$, with $\eta=0.7$, $\epsilon=0.15$ and $T=1$, on the same 100,000 one-year paths for every rule, drawn independently of those used for training. 
}
\label{tab:horse}
\begin{tabular}{lrrrrrr}
\hline
 & $\mathbb{E}[U]$ ($b{=}1$) & misalloc. & loss (\%) & $\mathbb{E}[U]$ ($b{=}3$) & misalloc. & loss (\%) \\
\hline
Constrained optimum (HJB) & 0.10649 & 0.00e+00 & 0.000 & 0.21796 & 0.00e+00 & 0.000 \\
Optimal exploratory policy & 0.10611 & 5.92e-05 & 0.037 & 0.21757 & 1.59e-04 & 0.036 \\
Learned actor--critic & 0.10615 & 5.09e-05 & 0.033 & 0.21754 & 2.25e-04 & 0.039 \\
Buy and hold ($\pi\equiv1$) & 0.10649 & 0.00e+00 & 0.000 & 0.10649 & 4.00e+00 & 10.085 \\
Clipped Gaussian & 0.10649 & 5.17e-21 & 0.000 & 0.21796 & 3.71e-07 & 0.000 \\
Estimate-and-plug-in & 0.07964 & 1.97e-01 & 2.579 & 0.12917 & 2.92e+00 & 8.094 \\
\hline
\end{tabular}

\end{table}
We observe that under $b=1$, the constrained optimum and the buy-and-hold (BH) strategy are the same
policy. The misallocation of the BH strategy is zero, and so is its wealth loss. 
Under $b=3$, the two separate. In particular, BH costs $10.09\%$ of initial wealth, while the learned policy $0.039\%$. It is worth emphasizing that the comparison with the plug-in investor highlights the distinction between learning and estimating. Both investors have the same information. However, the plug-in investor first estimates the price of
risk and then substitutes it into the optimal rule, whereas the learning investor directly learns the investment policy. As shown in Table \ref{tab:horse}, the plug-in investment incurs wealth losses of $2.58\%$ and $8.09\%$ of initial wealth for $b=1$ and $b=3$,  respectively, compared to $0.033\%$ and $0.039\%$ for the learned policy. These results suggest that learning may be more effective than the conventional estimate-and-optimize approach, particularly when the model parameters are difficult to estimate. 

\subsection{Out of sample on the realized $S\&P$~500 path}
\label{sec:num-backtest}
We apply the same evaluation framework to the $S\&P$~500 index. Data up to 2010 are used for
estimation and training, while the period 2011-2025 is used for out-of-sample testing. No
information from the testing period is used to determine the learned policy. The plug-in investor re-estimates
the model at the end of every calendar year and re-solves \eqref{sol_u} under the
current estimates, while the learned investor trains Algorithm~\ref{alg:AC-compact} only once using the pre-2011 data without knowledge about the market price of risk supplied, and subsequently applies the resulting policy without further training. 
\begin{table}[htbp]
\centering
\small
\caption{Backtest on the realized S\&P~500 path, 2011--2025. The model is re-estimated at the end of each year on past data only. $\eta=0.7$, $a=0$.}
\label{tab:backtest}
\begin{tabular}{lrrrrrrrr}
\hline
policy & return (\%) & vol.\ (\%) & Sharpe & Sortino & Calmar & max DD (\%) & $\bar\pi$ & sd$(\pi)$ \\
\hline
estimate-and-plug-in, $b=1$ & 13.23 & 17.47 & 0.64 & 0.77 & 0.33 & 33.8 & 1.00 & 0.00 \\
estimate-and-plug-in, $b=3$ & 27.11 & 42.92 & 0.58 & 0.73 & 0.41 & 61.3 & 2.84 & 0.35 \\
exploratory, $b=3$, $\epsilon=0.15$ & 26.55 & 42.51 & 0.58 & 0.72 & 0.40 & 61.4 & 2.80 & 0.34 \\
learned actor--critic, $b=1$ & 13.11 & 17.37 & 0.64 & 0.76 & 0.33 & 33.7 & 0.99 & 0.01 \\
learned actor--critic, $b=3$ & 27.32 & 41.88 & 0.60 & 0.75 & 0.44 & 57.0 & 2.77 & 0.34 \\
buy and hold ($\pi\equiv1$) & 13.23 & 17.47 & 0.64 & 0.77 & 0.33 & 33.8 & 1.00 & 0.00 \\
constant mix ($\pi\equiv0.6$) & 9.11 & 10.47 & 0.68 & 0.82 & 0.33 & 21.3 & 0.60 & 0.00 \\
risk-free & 2.00 & 0.00 & -- & -- & -- & 0.0 & 0.00 & 0.00 \\
\hline
\end{tabular}
\end{table}
Table \ref{tab:backtest} shows that when $b=1$, the learned rule closely replicates the BH strategy, obtaining an annualized return of $13.11\%$, compared with $13.23\%$ for BH. Both have a Sharpe ratio of $0.64$ and a maximum drawdown of approximately $34\%$.
When leverage is allowed, i.e., $b=3$, the learned policy achieves a higher annualized return of $27.32\%$ compared with $13.23\%$ for the unleveraged BH strategy. More importantly, compared to the plug-in policy under the same leverage constraint, the learned policy delivers a slightly better risk profile with a Sharpe ratio of $0.60$ (versus $0.58$), a higher Calmar ratio of $0.44$ (versus $0.41$), and a lower maximum drawdown of $57.0\%$ (versus $61.3\%$). These results suggest that the learned policy remains competitive with the annually recalibrated plug-in approach, despite being trained only once on the historical data and subsequently implemented without model reestimation.  
Figure \ref{fig:backtest} plots realized wealth under the learned actor-critic, estimate-and-plug-in, and the BH strategy over the out-of-sample period from 2011 to 2025. It can be observed that for $b=1$ the learned rule closely replicates the BH strategy, which is consistent with Table \ref{tab:backtest}. When leverage is allowed, i.e., $b=3$, both the learned and plug-in policies generate substantially greater terminal wealth than the BH. Interestingly, the learned policy closely tracks the plug-in strategy despite being trained only once on the past data, whereas the latter requires annual parameter reestimation and reoptimization. Together with Table \ref{tab:backtest}, these findings suggest that the actor–critic approach can learn a competitive investment policy and remains effective over an extended out-of-sample period without further retraining.

\begin{figure}[H]
	\centering
	\includegraphics[width=\linewidth]{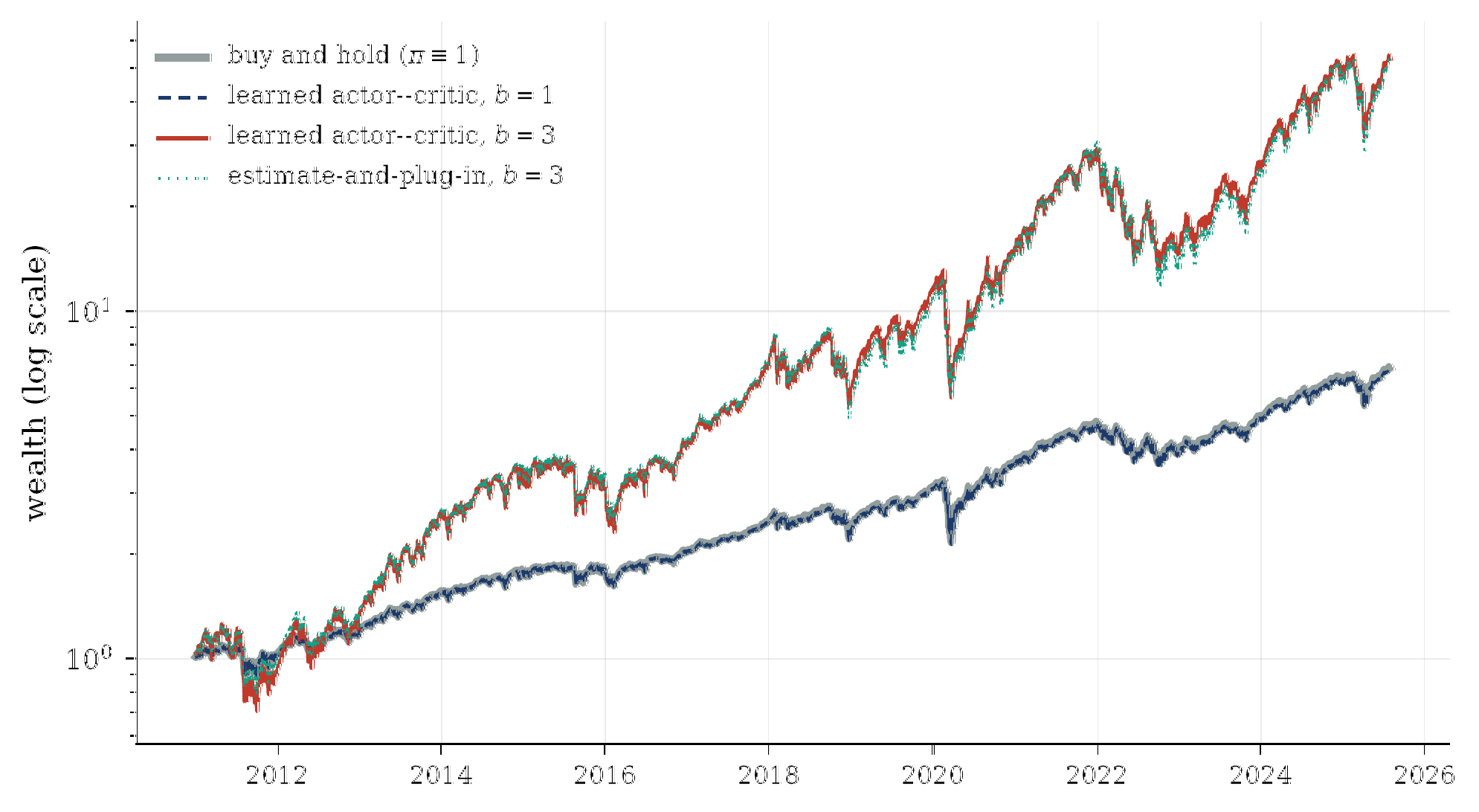}
	\caption{Realised wealth on the S\&P~500, 2011--2025. 
    }
	\label{fig:backtest}
\end{figure}

\section{Conclusion}
\noindent	We study a continuous-time portfolio optimization problem under stochastic volatility, portfolio constraints, and entropy-regularized exploration.
	We proved the well-posedness of the associated Hamilton–Jacobi–Bellman equation by establishing the existence and uniqueness of a classical solution in the spirit of \cite{ladyzhenskaia1968linear}, allowing us to characterize the value function in semi--closed form.
	This analytical structure enabled us to design an interpretable actor–critic algorithm that learns both the value function and the exploratory policy.
	The experiments, calibrated on S\&P~500 data, show when learning is useful. Under a no-leverage mandate, the constrained optimum is simply to hold the index. 
	Our framework opens several directions for future work, including the incorporation of additional sources of risk such as stochastic interest rates or ambiguity aversion, settings in which the volatility factor is only partially observed, and models with multi-agents e.g. \cite{liu2026partially}.
	These aspects will be addressed in subsequent research.

    \section*{Acknowledgements}
 \noindent The authors acknowledge the support from the Natural Sciences and Engineering Research Council of Canada (Grant No. RGPIN-2021-02594) and the Chair of Actuary, Laval University. 
 \bibliographystyle{elsarticle-harv}
\bibliography{literature}

\appendix
\renewcommand{\thesubsection}{\Alph{subsection}}

\section{Auxiliary lemmas}\label{sec:App-A}

\begin{lemma}[Donsker--Varadhan variational formula  \cite{donsker2006large}]
\label{lem:DV}
Let $(\Theta,\mathcal F,P)$ be a probability space, and let $h:\Theta\to\mathbb R$ be a measurable function such that $e^{h}\in L^1(P)$. Then
\begin{equation}
\label{eq:DV}
\ln \E_{P}\!\left[e^{h}\right]
=
\sup_{Q\ll P}
\left\{
\E_{Q}[h]-\KL(Q\|P)
\right\},
\end{equation}
where $\KL(Q\|P)$ is the Kullback-Leibler divergence defined by
\[
\KL(Q\|P):=
\begin{cases}
\displaystyle \E_Q\!\left[\ln\!\left(\frac{dQ}{dP}\right)\right], & Q\ll P,\\[1.2ex]
+\infty, & \text{otherwise}.
\end{cases}
\]
Moreover, the supremum is attained uniquely at the probability measure $Q^*$ defined by
\begin{equation}
\label{eq:QstarDV}
\frac{dQ^*}{dP}
=
\frac{e^{h}}{\E_P[e^{h}]}.
\end{equation}
\end{lemma}

\begin{proof}
Let $Q\ll P$. Then
\begin{align*}
\E_Q[h]-\KL(Q\|P)
=
\E_Q\!\left[h-\ln\!\left(\frac{dQ}{dP}\right)\right] =
\E_Q\!\left[\ln\!\left(e^{h}\frac{dP}{dQ}\right)\right].
\end{align*}
By Jensen's inequality, since $\ln$ is concave,
\begin{align*}
\E_Q\!\left[\ln\!\left(e^{h}\frac{dP}{dQ}\right)\right]
&\le
\ln \E_Q\!\left[e^{h}\frac{dP}{dQ}\right] =
\ln \E_P[e^{h}].
\end{align*}
Hence
\(\E_Q[h]-\KL(Q\|P)\le \ln \E_P[e^{h}]
\) 
for every $Q\ll P$, which proves
\[
\sup_{Q\ll P}\left\{\E_Q[h]-\KL(Q\|P)\right\}\le \ln \E_P[e^{h}].
\]
Now define $Q^*$ by
\[
\frac{dQ^*}{dP}=\frac{e^{h}}{\E_P[e^{h}]}.
\]
Since $\E_P[e^{h}]<\infty$, $Q^*$ is a probability measure and $Q^*\ll P$. Moreover,
\(
\ln\!\left(\frac{dQ^*}{dP}\right)=h-\ln \E_P[e^{h}],
\)
so that
\begin{align*}
\E_{Q^*}[h]-\KL(Q^*\|P)
&=
\E_{Q^*}[h]
-
\E_{Q^*}\!\left[\ln\!\left(\frac{dQ^*}{dP}\right)\right] \\
&=
\E_{Q^*}[h]
-
\E_{Q^*}[h-\ln \E_P[e^{h}]] =
\ln \E_P[e^{h}].
\end{align*}
Therefore, the supremum is attained at $Q^*$, and \eqref{eq:DV} follows.

Finally, since equality in Jensen's inequality holds if and only if
\(
e^{h}\frac{dP}{dQ}
\)
is $Q$-almost surely constant, there exists a constant $C>0$ such that
\(
e^{h}\frac{dP}{dQ}=C
\quad Q\text{-a.s.},
\)
which is equivalent to
\[
\frac{dQ}{dP}=\frac{e^{h}}{C}.
\]
Since $Q$ is a probability measure, integrating with respect to $P$ yields
\[
1=\int_\Theta \frac{e^{h}}{C}\,dP=\frac{\E_P[e^{h}]}{C},
\]
hence $C=\E_P[e^{h}]$, and therefore $Q=Q^*$. This proves uniqueness.
\end{proof}
        
    \begin{lemma}\label{LEM_Appendix}
    Consider $Z_{a,b}(y,u_y;m)=\Phi(B(y,u_y;m))-\Phi(A(y,u_y;m))$ where $A(y,u_y;m),\,B(y,u_y;m)$ are defined respectively by \eqref{A_y}, \eqref{B_y}. We have
    \begin{align}
    	\underset{m\to 0}{\lim}m\ln Z_{a,b}(y,u_y;m)=\begin{cases}
    		0 & \text{if}\,\, \pi^{Merton} \in [a,b]\\
    		-\frac{1}{2}(a-\pi^{Merton})^2\eta \sigma^2(y)& \text{if}\,\, \pi^{Merton}<a\\
    		-\frac{1}{2}(b-\pi^{Merton})^2\eta \sigma^2(y)& \text{if}\,\, \pi^{Merton}>b,
    	\end{cases}
    \end{align}
    with $\pi^{Merton}=\frac{\mu-r+\rho\delta(y)\sigma(y)u_y}{\eta \sigma^2(y)}.$
     \end{lemma}
    \begin{proof}
    	First we notice that \begin{align*}
    		\underset{m\to 0}{\lim}\ln Z_\zs{a,b}(y,u_y;m)=\begin{cases}
    			0 & \text{if}\,\, \pi^{Merton} \in [a,b]\\
    			-\infty &\text{if}\,\, \pi^{Merton}<a\\
    			-\infty &\text{if}\,\, \pi^{Merton} >b.
    		\end{cases}
    	\end{align*}
    	For $\pi^{Merton} \in [a,b]$, $\underset{m\to 0}{\lim} m\ln Z_\zs{a,b}(y,u_y;m)=0$. 
        Now we assume $ \pi^{Merton}<a$, and apply  L'Hopital's rule:
    	\begin{align}
    		\underset{m\to 0}{\lim}\frac{\ln Z_{a,b}(y,u_y;m)}{\frac{1}{m}}&=\underset{m\to 0}{\lim}\dfrac{\frac12 \sqrt{m\eta \sigma^2(y)}\varphi((b-\pi^{Merton})\sqrt{\frac{\eta \sigma^2(y)}{m}})-\frac12 \sqrt{m\eta \sigma^2(y)}\varphi((a-\pi^{Merton})\sqrt{\frac{\eta \sigma^2(y)}{m}})}{\Phi((b-\pi^{Merton})\sqrt{\frac{\eta \sigma^2(y)}{m}})-\Phi((a-\pi^{Merton})\sqrt{\frac{\eta \sigma^2(y)}{m}})}\nonumber\\
    		&\overset{(*)}{=}\underset{m\to 0}{\lim}
    		\dfrac{\frac14 mf(b)+\frac14(b-\pi^{Merton})^2\eta \sigma^2(y)f(b)-\frac14 mf(a)-\frac14(a-\pi^{Merton})^2\eta \sigma^2(y)f(a)}{-\frac12f(b)+\frac12 f(a)}\nonumber\\
    		&=	-\frac{1}{2}(a-\pi^{Merton})^2\eta \sigma^2(y)
    	\end{align}
    	where $ f$ in $(*)$ is given by $f(k)=\sqrt{\eta\sigma^2(y)}(k-\pi^{Merton})\vphi\bigg((k-\pi^{Merton})\sqrt{\frac{\eta\sigma^2(y)}{m}}\bigg)$, {the proof of the case $\pi^{Merton}>b$ is similar and hence is omitted}.
    \end{proof}
    

\begin{theorem}[SDE under Condition \textbf{SV}]\label{Appendix_Existence_SDE}
	Suppose that $\varpi$ and $\gamma$ are continuous and there exists $K<\infty$ such that
	\begin{enumerate}
		\item $|\varpi(x)-\varpi(y)|+|\gamma(x)-\gamma(y)|\le K|x-y|$
		\item 
		$|\varpi(y)|+|\gamma(y)|\le K(1+|y|)$
	\end{enumerate}
	Then, for all $T>0$, there exists a unique solution to the SDE
    \begin{equation}
	\label{AD}
	\begin{cases}
		
		\d y_\zs{t}=\varpi (y_\zs{t})\d t+ \gamma(y_{t})\d W_\zs{t}\\
		y_0 \in \mathbb{R}
	\end{cases}
\end{equation}
 Moreover this solution $(y_\zs{s})_{0\le s\le T}$ satisfies
	\(
	\E\left(\sup_{0 \le s\le T}|y_\zs{s}|^2\right) <\infty\).
\end{theorem}
\proof
The proof can be found e.g., in \cite{lamberton1997introduction} pages 57-60.
\endproof

\begin{theorem}[Generalized Feynman-Kac's formula] \label{Thm:FeymanKac}Consider the partial differential equation 
\begin{align*}
	\frac{\partial}{\partial t}u(x,t)+\mu(x,t)\frac{\partial}{\partial x}u(x,t)+\frac12\sigma^2(x,t)\frac{\partial^2}{\partial x^2}u(x,t)-r(t,x)u(x,t)+f(x,t)=0 \,\, (t,x)\in [0, T)\times \bbr,
\end{align*}
with $u(x, T)=g(x)$, where $\mu, \sigma, g, f$ are known functions, and $u: \bbr\times [0, T]\to \bbr$ is the unknown. Then the Feynman-Kac formula expresses $u(x,t)$ as 
\begin{align}
	u(x,t)=\mathbb{E}\bigg[\int_{t}^{T}e^{-\int_{t}^{s}r(u,X_u)\d u}f(s,X_s)\d s+e^{-\int_{t}^{T}r(u,X_u)\d u}g(X_T)\bigg|X_t=x\bigg]
	,\end{align}
where $X$ satisfies $\d X_t=\mu(X_t,t)\d t+\sigma(X_t,t)\d W_t$.

\end{theorem}
\proof See e.g., \cite{pardoux2005backward, peng1991probabilistic,peng1992generalized}. \endproof
\section{Truncated Gaussian distribution} \label{Sec:App-Truncated}
Let $X$ be $\cN(\alpha,\beta^2)$ and $Y$ a truncated normal random variable of $X$ on $[a,b]$. With $\vphi$ the standard normal probability density function, the truncated normal PDF and CDF are, respectively, given by
\[f(y,\alpha,\beta^2,a,b)=\frac{\vphi(\frac{y-\alpha}{\beta})}{\beta[\Phi(\frac{b-\alpha}{\beta})-\Phi(\frac{a-\alpha}{\beta})]}\mathds{1}_{[a,b]}(y)
\]
\[F(y,\alpha,\beta^2,a,b)=\frac{[\Phi(\frac{y-\alpha}{\beta})-\Phi(\frac{a-\alpha}{\beta})]}{[\Phi(\frac{b-\alpha}{\beta})-\Phi(\frac{a-\alpha}{\beta})]}\mathds{1}_{[a,b]}(y)
\]
The mean and variance of a truncated Gaussian distribution are, respectively, given by
\[
\mathbb{E}(Y)=\alpha +\beta \frac{\vphi(\frac{a-\alpha}{\beta})-\vphi(\frac{b-\alpha}{\beta})}{\Phi(\frac{b-\alpha}{\beta})-\Phi(\frac{a-\alpha}{\beta})}
\]
\[ \text{Var}(Y)=\beta^2\left[1+\dfrac{\frac{a-\alpha}{\beta}\vphi(\frac{a-\alpha}{\beta})-\frac{b-\alpha}{\beta}\vphi(\frac{b-\alpha}{\beta})}{\Phi(\frac{b-\alpha}{\beta})-\Phi(\frac{a-\alpha}{\beta})}-\left(\frac{\vphi(\frac{a-\alpha}{\beta})-\vphi(\frac{b-\alpha}{\beta})}{\Phi(\frac{b-\alpha}{\beta})-\Phi(\frac{a-\alpha}{\beta})} \right)^2
\right].
\]
The differential entropy of the truncated normal distribution is given by
\[
\mathcal{H}(Y)
= \frac{1}{2}\log(2\pi_e e \beta^2)
+ \log\!\Big(\Phi\big(\tfrac{b-\alpha}{\beta}\big)-\Phi\big(\tfrac{a-\alpha}{\beta}\big)\Big)
+ \frac{\frac{a-\alpha}{\beta}\vphi\big(\tfrac{a-\alpha}{\beta}\big)
      -\frac{b-\alpha}{\beta}\vphi\big(\tfrac{b-\alpha}{\beta}\big)}
       {2\Big[\Phi\big(\tfrac{b-\alpha}{\beta}\big)-\Phi\big(\tfrac{a-\alpha}{\beta}\big)\Big]},
\]
see, e.g., \cite{shah1966, johnson1972continuous} for properties of the truncated normal distribution.

\section{Quasilinear parabolic PDEs}\label{Sec:paraPDE}
\noindent 
Here we recall the existence theorem from \cite{ladyzhenskaia1968linear} for the Cauchy problem in $\bbr^{n}$. To do this, we set the differential operator as the operator.
\begin{equation}\label{Operator-Gen-1-1}
	\cL(x,t,\frac{\partial }{\partial x})u
	=
	\sum_\zs{1  \leq  i, j  \leq  n} c_\zs{ij}(x,t, u)u_\zs{x_\zs{i} x_\zs{j} }- c(x, t, u, u_\zs{x})\,.
\end{equation}
\noindent
Now we study the following Cauchy problem.
\begin{align}\label{sec:CauchyEq.110}
	\begin{cases}
		u_\zs{t}  -\cL(x,t,\frac{\partial }{\partial x})u
		= 0 \,, \\[3mm]
		u|_\zs{t = 0} = u(x, 0) = \psi_\zs{0}(x)\,, \quad x\in\bbr^{n}\,.
	\end{cases}
\end{align}
We assume that there exist some functions $(c_\zs{1}, c_\zs{2}, \ldots, c_\zs{n}) ,$ such that
\begin{equation}\label{sec:CaPr.00}
	c_\zs{ij}(x, t, u, p)  \equiv \frac{\partial c_\zs{i}(x, t,  u, p)}{\partial p_\zs{j}} 
	\,,\quad p\in\bbr^{n}\,.
\end{equation}
Using these functions, we set
$$
C(x, t, u, p)  \equiv c(x, t, u, p)- \sum^{n}_\zs{i = 1} \frac{\partial c_\zs{i}}{\partial u} p_i- \sum^{n}_\zs{i = 1} \frac{\partial c_\zs{i}}{\partial x_\zs{i}}\,.
$$
\noindent
Next, we also need the differential operator in $\bbr^{n}$ to be defined as
\begin{equation}
	\label{Dif-Opr--1-x}
	D_\zs{l_\zs{1},\dots,l_\zs{n}}\,r(x)
	=\frac{\partial^{l_\zs{1}+\dots+l_\zs{n}}}{(\partial x_\zs{1})^{l_\zs{1}}\dots (\partial x_\zs{n})^{l_\zs{n}}}\,r(x)\,.
\end{equation}
\noindent Let now $\V$ be some closed convex subset in $\bbr^{n}$, i.e., $\V\subseteq \bbr^{n}$.   Now we recall
the definitions of the H\"older spaces. First of all, for any $\V\to \bbr$ function $r$, we set the H\"older constant of  order $0<\alpha<1$ as
\begin{equation}
	\label{holder-const-1}
	<r>^{(\alpha)}_\zs{\V}=
	\sup_\zs{|x-y|\le 1, x,y\in\V}
	\frac{|r(x)-r(y)|}{|x-y|^{\alpha}}\,.
\end{equation}
\noindent
We denote by $\cH^{\mu}(\V)$ the H\"older space of order $\mu>0$, i.e., the Banach space
of $\V\to\bbr$ functions $r$, which are continuous along with all their derivatives
$D_\zs{l_\zs{1},\dots,l_\zs{n}}\,r(x)$ for $l_\zs{1}+\dots+l_\zs{n}\le [\mu]$
, and
$
\|r\|_\zs{\mu,\cH}<\infty\,,
$
where $[\mu]$ is the integer part of $\mu$, and in the sequel, we denote the fractional part by $\{\mu\}$, and $|r|_\zs{\mu,\cH}$ is the H\"older norm, which
in this case
is defined as
\begin{align}
	\|r\|_\zs{\mu,\cH}&=
	\sum^{[\mu]}_\zs{j=0}
	\sum_\zs{l_\zs{1}+\ldots+l_\zs{n}=j}
	\sup_\zs{x\in\V}
	\big|
	D_\zs{l_\zs{1},\dots,l_\zs{n}}\,r(x)
	\big| 
	+
	\sum_\zs{l_\zs{1}+\ldots+l_\zs{n}=[\mu]}
	<D_\zs{l_\zs{1},\dots,l_\zs{n}}\,r>^{\{\mu\}}\,.\label{Hld-nprm--1}
\end{align}
\noindent
It should be noted that in the scalar case, i.e., when $n=1$, this norm is represented as
\begin{equation}
	\label{Hld-nprm--1-n-1}
	\|r\|_\zs{\mu,\cH}=
	\sum^{[\mu]}_\zs{j=0}
	\sup_\zs{x\in\V}
	\big|
	r^{(j)}(x)
	\big|
	+
	<r^{([\mu])}>^{\{\mu\}}_\zs{\V}\,,
\end{equation}
\noindent where $r^{(l)}$ is the derivative of order $l$.

\noindent
Moreover, similar to the definition \eqref{Dif-Opr--1-x}
we set
for any $\bbr^{n}\times [0,T]\to\bbr$ function $u$  the corresponding differential operator as
\begin{equation}
	\label{Dif-Opr--2-x-t}
	D^{k}_\zs{l_\zs{1},\dots,l_\zs{n}}\,u(x,t)
	=
	\frac{\partial^{l_\zs{1}+\dots+l_\zs{n}}}{(\partial x_\zs{1})^{l_\zs{1}}\dots (\partial x_\zs{n})^{l_\zs{n}}}\,
	\frac{\partial^{k}}{(\partial t)^{k}}
	\,u(x,t)\,.
\end{equation}
\noindent
Moreover, 
let now $\Gamma=\V\times [0,T]$. 
We denote by $\cH^{\mu,\mu/2}(\Gamma)$ the H\"older space of orders $(\mu,\mu/2)$, i.e.,
this is the Banach space  of the $\bbr^{n}\times [0,T]\to\bbr$ functions $u$ which are continuous together with all their derivatives 
$D^{k}_\zs{l_\zs{1},\dots,l_\zs{n}}\,u(x,t)$ for $2k+l_\zs{1}+\ldots+l_\zs{n}\le [\mu]$ and 
$$
\|u\|_\zs{\mu,\mu/2,\cH}<\infty\,.
$$
\noindent Here $\|u\|^{\mu}_\zs{\cH}$ is the H\"older norm, which is defined as
\begin{align}
	\nonumber
	\|u\|_\zs{\mu,\mu/2,\cH}&=
	\sum^{[\mu]}_\zs{j=0}
	\sum_\zs{2k+l_\zs{1}+\ldots+l_\zs{n}=j}
	\sup_\zs{x\in\V}
	\sup_\zs{0\le t\le T}
	\big|
	D^{k}_\zs{l_\zs{1},\dots,l_\zs{n}}\ u(x,t)
	\big|\\[4mm] \nonumber
	&+
	\sum_\zs{2k+l_\zs{1}+\ldots+l_\zs{n}=[\mu]}
	<D^{k}_\zs{l_\zs{1},\dots,l_\zs{n}}\,u>^{\{\mu\}}_\zs{x}\\[4mm] \label{Holder-norm-2}
	&+
	\sum_\zs{2k+l_\zs{1}+\ldots+l_\zs{n}=[\mu]}
	<D^{k}_\zs{l_\zs{1},\dots,l_\zs{n}}\,u>^{\{\mu\}}_\zs{t}
	\,,
\end{align} 
where
$$
<u>^{\alpha}_\zs{\V,x}=
\sup_\zs{|x-y|\le 1, x,y\in\V\,t\in [0,T]}
\frac{|u(x,t)-u(y,t)|}{|x-y|^{\alpha}}
$$
and 
$$
<u>^{\alpha}_\zs{\V,t}=
\sup_\zs{|t-s|\le 1, x\in\V\,s,t\in [0,T]}
\frac{|u(x,t)-u(x,s)|}{|t-s|^{\alpha}}\,.
$$
\noindent 
Moreover, for any  $N \geq 1$ we set
$$
\Gamma_\zs{N} =  \{ (x, t)\in\bbr^{n}\times [0,T]:  \,  |x|  \leq  N \}\,.
$$
We introduce the following conditions for ensuring the existence of  at least one  solution $u(x, t)$ for the problem \eqref{sec:CauchyEq.110}.


\begin{itemize}
	\item[$\C_\zs{1}$)] \hspace{1mm} {\sl 
		The operator \eqref{Operator-Gen-1-1} is uniformly elliptic, i.e.,  for any $ t \in (0,T]$, for  arbitrary $ x,\;  u,\;  p \in \bbr^n $, and any $z =  (z_\zs{1}, z_\zs{2}, ..., z_\zs{n}) \in \bbr^{n}$, there exists 
		$ 0 < \nu_\zs{1} \le \nu_\zs{2}<\infty$  such that
		\begin{align*}
			\nu_\zs{1} |z|^2  \leq 
			\sum^{n}_\zs{i,j=1}
			c_\zs{ij} (x, t,  u, p) z_\zs{i} z_\zs{j}  \leq  \nu_\zs{2} |z|^2\,,
		\end{align*}
		\noindent where $|z|^{2}=\sum^{n}_\zs{j=1}\,z^{2}_\zs{j}$.
	}
	\item[$\C_\zs{2}$)] \hspace{1mm} { There exists $ b \geq 0 $ and some $ \bbr_\zs{+} \rightarrow \bbr_\zs{+}$ function $ \Psi $,
		such that for all $ x \in \bbr^n$, $u \in \bbr$ and for all $ 0 \le t \le T$,
		\begin{equation}\label{sec:CaPr.1}
			C(x, t, u, 0) u \geq - \Psi (|u|) |u| - b ,
			\quad
			\mbox{and}
			\quad
			\int_\zs{0}^\infty \frac{d \tau}{ \Psi (\tau)} =  \infty\,.
		\end{equation}
		
	}
	\item[$\C_\zs{3}$)] \hspace{1mm} {\sl There exists $\varepsilon>0$ such that for all $ N \geq 1$,  
		\begin{align*}
			\psi_\zs{0}(x) \in \cH^{2+\varepsilon}(\Gamma_\zs{N})
			\quad
			\mbox{and}
			\quad 
			\max_\zs{x\in\bbr^{n}} \mid \psi_\zs{0}(x) \mid <\infty\,.
	\end{align*} }

	\item[$\C_\zs{4}$)] \hspace{1mm}
	{\sl The functions $ c_\zs{i} (x, t, u, p) $ and $ c(x, t, u, p)$ are continuous, 
		the functions $ (c_\zs{i})_\zs{1  \leq  i  \leq  n} $ are differentiable with respect to $ x, u$ and $ p \in \bbr^n$ ,
		and for any $ N \geq 1$
		$$
		\sup_\zs{(x, t)  \leq  \Gamma_\zs{N}}\,\sup_\zs{|u|  \leq  N}\,\sup_\zs{ p \in \bbr^n}
		\,\frac{\sum^{n}_\zs{i = 1} 
			\big( |c_i| + \big| \dfrac{\partial c_i}{\partial u} \big| \big) (1+|p|) + \sum^n_\zs{i, j = 1} |\dfrac{\partial c_i}{\partial x_\zs{j}}| +|c|}{1+ |p|^2} <\infty\,.
		$$
	} 
	
	\item[$\C_\zs{5}) $] \hspace{1mm}
	{\sl  For all $ N \geq 1,$ 
		and for all $(x,t)\in\Gamma_\zs{N}$,
		$ |u|  \leq  N$ 
		and $ |p|  \leq  N $,
		the functions $c_\zs{i}$, $c$,  $\partial c_\zs{i} / \partial p_\zs{j}$, $\partial c_\zs{i} / \partial u$,
		and $ \partial c_\zs{i} / \partial x_\zs{i}$ are continuous functions satisfying a H$\ddot{o}$lder condition in 
		$x$, $t$, $u$ and $p$ with exponents $\varepsilon$, $\varepsilon/ 2$, $\varepsilon$ and $\varepsilon$ respectively for $\varepsilon>0$ from 
		the condition $\C_\zs{1})$.
	}    
\end{itemize}

\begin{theorem}[ See Theorem 8.1, p. 495 of \cite{ladyzhenskaia1968linear}]\label{CaPrt00_Theorem}
	Assume that the conditions $ \C_\zs{1}$)--$ \C_\zs{5}$) hold. Then there exists at least one solution $ u(x,t)$ of Cauchy problem (\ref{sec:CauchyEq.110})
	which is bounded in $\bbr^n \times [0, T]$ and  for any $N \geq 1$ belongs to $\cH^{2+\varepsilon, 1+\varepsilon/2}(
	\Gamma_\zs{N})$.It will be an element of $H^{2+\epsilon,1+\epsilon/2}(\bbr^n \times [0, T])$ if, in addition, it is assumed that the constants in conditions $\gr{C}_4), \gr{C}_5)$do not depend on $\Gamma_N$.
	
	If the functions $c_{ij}(x,t,u,p)$ and $C(x,t,u,p)$ are differentiable with respect to $u$ and $p$ and
	\[
	\max_{(x,t) \in \bbr^n\times [0, T], |u,p| \leq N} \left| \frac{\partial c_{ij}(x,t,u,p)}{\partial u},\frac{\partial c_{ij}(x,t,u,p)}{\partial p}, \frac{\partial C(x,t,u,p)}{\partial p}   \right| \leq \mu_1(N),
	\]
	\[
	\min_{(x,t) \in \bbr^n\times [0, T], |u,p| \leq N} \frac{\partial C(x,t,u,p)}{\partial p} \geq -\mu_2(N),
	\]
	for an arbitrary $N$ and some constants $\mu_1, \mu_2$ depending possibly on $N$, then problem \eqref{sec:CauchyEq.110} in  $\bbr^n\times [0, T]$ has no more than one classical solution $u(x,t)$ that is bounded in $\bbr^n\times [0, T]$ together with its derivatives of first and second orders.
	
\end{theorem}
\section{Linear parabolic PDEs}

\noindent
In this section we study the problem \eqref{sec:CauchyEq.110} for the linear operator 

\begin{align}
\label{Operator-Linear-1-1}
	\cL(x,t,\frac{\partial }{\partial x})u
	&=
	\sum_\zs{1  \leq  i, j  \leq  n} a_\zs{ij}(x,t, u)u_\zs{x_\zs{i} x_\zs{j} }
	-\sum^{n}_\zs{j=1}\,a_\zs{i}(x,t)\,u_\zs{x_\zs{i}}
	-
	a(x, t) u\,.
\end{align}

\noindent 
Now for this operator we consider 
the following Cauchy problem in $\bbr^{n}\times [0,T]$, i.e.,
\begin{align}\label{sec:CauchyEq.110-LinEq}
	\begin{cases}
		u_\zs{t}  -\cL(x,t,\frac{\partial }{\partial x})u
		= f \,, \\[3mm]
		u|_\zs{t = 0} = u(x, 0) = \psi_\zs{0}(x)\,, \quad x\in\bbr^{n}\,.
	\end{cases}
\end{align}
\noindent where $f$ is some $\bbr^{n}\times [0,T]\to\bbr$ function.

\begin{theorem}[ See Theorem 5.1, p. 320 in \cite{ladyzhenskaia1968linear}]\label{CaPrt00-Theorem-Lin}
	Assume that the coefficients of the operator \eqref{Operator-Linear-1-1} belong to the space 
	$\cH^{\mu,\mu/2}(\bbr^{n}\times [0,T])$, where $\mu=l+\epsilon$
	for some integer $l\ge 0$ and $0<\epsilon<1$. Then for any $f\in\cH^{\mu,\mu/2}(\bbr^{n}\times [0,T])$ and 
	$\psi_\zs{0}\in \cH^{\mu+2}(\bbr^{n})$ the problem \eqref{sec:CauchyEq.110-LinEq} with the differential operator \eqref{Operator-Linear-1-1}
	has a unique solution in 
	$\cH^{\mu+2,\mu/2+1}(\bbr^{n}\times [0,T])$
	such that
	$$
	\sup_\zs{f\in\cH^{\mu,\mu/2}(\bbr^{n}\times [0,T])}
	\,
	\sup_\zs{\psi_\zs{0}\in \cH^{\mu+2}(\bbr^{n})}
	\frac{
		\|u\|_\zs{\mu+2,\mu/2+1,\cH}}{\|f\|_\zs{\mu,\mu/2,\cH}
		+
		\|\psi_\zs{0}\|_\zs{\mu+2,\cH}}
	<\infty\,.
	$$
\end{theorem}

\section{Solution existence and volatility upper bound for special cases}\label{sec: upperbound}

\begin{proposition}[No exploration, without portfolio constraints]\label{Prop_1}
	{In the absence of exploration and portfolio constraints, i.e., $m=0$ and $a=-\infty,\,b=+\infty$ in PDE \eqref{sol_u}, the condition $\C_{2})$ in Theorem \ref{CaPrt00_Theorem} is always satisfied with $\Psi=\frac{1-\eta}{2\eta}\bigg(\frac{\mu-r}{\sigma_*}\bigg)^2+r(1-\eta) >0$.}
\end{proposition}
\begin{proof}
{It can be seen directly that}
	\begin{align}\label{proof_Lemma1}
		C_\zs{0}(y,t,l,0)&=-\frac{1-\eta}{2}\bigg[\frac{(\mu-r)^2}{\eta \sigma^2(y)}\bigg]-r(1-\eta) \ge -\frac{1-\eta}{2\eta}\bigg(\frac{\mu-r}{\sigma_*}\bigg)^2-r(1-\eta),\nonumber \\
	\end{align}
    and, for any $l \in \bbr$,
	$
		C_\zs{0}(y,t,l,0)l\ge -\Psi(|l|)|l|,
	$ 
	where $\Psi=\frac{1-\eta}{2\eta}\bigg(\frac{\mu-r}{\sigma_*}\bigg)^2+r(1-\eta) >0$. Hence, $\C_2)$ is fulfilled.
\end{proof}

{The following is a direct result of Proposition \ref{Prop_4} for the case without portfolio constraints. }

\begin{proposition}[{Exploration without portfolio constraints}]
	Let $m\ne 0$ and $a=-\infty,\,b=+\infty$ and assume that $\frac{(\mu-r)^2}{\eta \sigma_*^2}+m\ln\big(\frac{2\pi_\zs{e}m}{\eta \sigma^2_*}\big)>0$ and $\underset{y \in \bbr}{\sup} \,\sigma^2(y)\le \frac{q_\zs{0}^{m}}{\eta}$ with $q_\zs{0}^{m}$ is the unique solution of
		\begin{align}\label{Prop_2_eq_1}
			f_\zs{m}(q)=\frac{(\mu-r)^2}{q}+m\ln\bigg(\frac{2\pi_e m}{q}\bigg)=0.
		\end{align}
		Then $\C_2)$ in Theorem \eqref{CaPrt00_Theorem} is fulfilled.
\end{proposition}
\begin{proof}
{In this case, we observe that }
	\begin{align}
		C_\zs{m}(y,t,l,0)&=-\frac{1-\eta}{2}\bigg[\frac{(\mu-r)^2}{\eta \sigma^2(y)}+m\ln\bigg(\frac{2\pi_e m}{\eta \sigma^2(y)}\bigg)\bigg]-r(1-\eta)\nonumber \\
		&=-\frac{1-\eta}{2}f_\zs{m}(\eta\sigma^2(y))-r(1-\eta)
	\end{align}
    where $f_\zs{m}(q)$ is defined by \eqref{Prop_2_eq_1}. It is straightforward to see that $f_\zs{m}(q)$ is a decreasing function from $(\eta\sigma^2_*,\infty)$ to $( -\infty,f^{a,b}_m(\eta\sigma^2_*)]$, so there exists a unique $\q_\zs{0}^{m}$ such that $f_\zs{m}(q_\zs{0}^m)=0$ and $f_\zs{m}(q)>0$ for $q\in (\eta\sigma^2_*,q_\zs{0}^m)$ i.e., $\sigma^2(y) \in (\sigma^2_*,\frac{q_\zs{0}^m}{\eta})$. Note that if $m=0$, $q_\zs{0}^{0}=+\infty$. Then, condition $C_2)$ is fulfilled with $\Psi=\frac{1-\eta}{2\eta}\bigg(\frac{\mu-r}{\sigma_*}\bigg)^2+m\ln(\frac{2\pi_e m}{\eta \sigma^2_*})+r(1-\eta) >0$.
\end{proof}
\begin{proposition}[{Exploration with shortselling constraints}]\label{prop_3}
	Let $m\ne 0$, $a=0, b=+\infty$ and assume that $\frac{(\mu-r)^2}{\eta \sigma_*^2}+m\ln\big(\frac{2\pi_\zs{e}m}{\eta \sigma^2_*}\big)+2m\ln(\Phi\big(\frac{\mu-r}{\sqrt{m\eta\sigma_*^2}}\big))>0$ and $\underset{y\in \bbr}{\sup}_, \sigma^2(y)\le \frac{\q_\zs{0}^m}{\eta}$, with $q_\zs{0}^{m}$ is the unique solution of \begin{align}\label{Prop_3_eq_1}
			f^a_\zs{m}(q)=\frac{(\mu-r)^2}{q}+m\ln\bigg(\frac{2\pi_e m}{q}\bigg)+2m\ln\bigg(\Phi\bigg(\frac{\mu-r}{\sqrt{mq}}\bigg)\bigg)=0.\end{align}  
Then $\C_2)$ in Theorem \eqref{CaPrt00_Theorem} is fulfilled.
\end{proposition}
\begin{proof}
	\begin{align}
		C_\zs{m}(y,t,l,0)&=-\frac{1-\eta}{2}\bigg[\frac{(\mu-r)^2}{\eta \sigma^2(y)}+m\ln\bigg(\frac{2\pi_e m}{\eta \sigma^2(y)}\bigg)+2m\ln\bigg(\Phi\bigg(\frac{\mu-r}{\sqrt{m\eta \sigma^2(y)}}\bigg)\bigg)\bigg]-r(1-\eta)\nonumber \\
		&=-\frac{1-\eta}{2}f^a_\zs{m}(\eta\sigma^2(y))-r(1-\eta)
	\end{align}
	We have $f^a_m(q)=\frac{(\mu-r)^2}{q}+m\ln\bigg(\frac{2\pi_e m}{q}\bigg)+2m\ln\bigg(\Phi\bigg(\frac{\mu-r}{\sqrt{mq}}\bigg)\bigg)$. Let us compute the derivative of $f^a_m(q)$: 
	$$\dot{f}_m^{a}(q)=-\frac{\mu-r}{q^2}-\frac{m}{q}-\frac{\frac{m}{q}\frac{\mu-r}{\sqrt{mq}}\vphi(\frac{\mu-r}{\sqrt{mq}})}{\Phi(\frac{\mu-r}{\sqrt{mq}})},$$
	$f^a_m(q)<0$ for all $q\in [\eta\sigma^2_*,\infty)$. Hence $f^a_m(q)$ is a decreasing function from $[\eta\sigma^2_*,\infty)$ to $( -\infty,f^{a,b}_m(\eta\sigma^2_*)]$, so there exists a unique $\q_\zs{0}^{m}$ such that $f_\zs{m}(q_\zs{0}^m)=0$ and $f_\zs{m}(q)>0$ for $q\in (\eta\sigma^2_*,q_\zs{0}^m)$ i.e., $\sigma^2(y) \in (\sigma^2_*,\frac{q_\zs{0}^m}{\eta})$. Note that if $m=0$, $q_\zs{0}^{0}=+\infty$. Then, condition $\C_2)$ is fulfilled with $\Psi=\frac{1-\eta}{2\eta}\bigg(\frac{\mu-r}{\sigma_*}\bigg)^2+m\ln(\frac{2\pi_e m}{\eta \sigma^2_*})+2m \ln\bigg(\Phi\bigg(\frac{\mu-r}{\sqrt{m\eta \sigma^2_*}}\bigg)\bigg)+r(1-\eta) >0$.	
\end{proof}

\section{Diagonal process}\label{Diag_process}
\noindent By Theorem \ref{CaPrt00_Theorem}, the solution \( u^{\lambda^n} \) of the partial differential equation \eqref{related_PDE_n} belongs to the function space \( \mathcal{H}^{2+\epsilon,1+\epsilon/2}(\Gamma_N) \). To simplify notation, we denote \( u^{\lambda^n} \) by \( u^n \) in the following. This implies that \( u^n \) is uniformly bounded, i.e., there exists a constant \( c(N) \), which depends on \( N \) but not on \( n \), such that  
\[
\sup_{(t,y) \in \Gamma_N} |u^n(t,y)| \leq c(N).
\]  

\noindent Moreover, since \( u^n \in \mathcal{H}^{2+\epsilon,1+\epsilon/2}(\Gamma_N) \), it satisfies Hölder continuity conditions in both variables. In particular, there exists a constant \( C > 0 \) such that  
\[
|u^n(t,y) - u^n(s,y)| \leq C |t - s|,
\]  
which implies equicontinuity in \( t \). By the regularity of \( u^n \) in the space \( \mathcal{H}^{2+\epsilon,1+\epsilon/2} \), similar Hölder estimates hold for the first and second derivatives with respect to \( y \), ensuring equicontinuity in those variables as well. By the Ascoli–Arzelà theorem, we can extract a subsequence \( u^{n_k} \) that converges uniformly. For each fixed \( y \), the sequence \( u^{n_k}_t \) is also uniformly bounded and equicontinuous due to the Hölder conditions. Applying Ascoli–Arzelà again, we extract a further subsequence \( u^{n_{kl}}_t \) that converges uniformly. By the same argument, using the Hölder continuity of \( u^n \) in \( y \) and \( yy \), we can successively extract subsequences  
$
u^{n_{klr}}_y$ and $u^{n_{klrp}}_{yy}
$ 
that converge uniformly.  
Thus, we obtain a nested sequence of subsequences satisfying  
\[
u^n \supseteq u^{n_k} \supseteq u^{n_{kl}} \supseteq u^{n_{klr}} \supseteq u^{n_{klrp}}.
\]  
Now, consider the diagonal subsequence \( u^p \), whose \( p \)-th term is the \( p \)-th element in the nested sequence of subsequences \( u^{n_{klrp}} \). By construction, \( u^p \) converges uniformly for every \( t, y \) to a limit function \( u \), along with its corresponding derivatives.

\section{Parameterization and supplementary details for numerical example in Section \ref{sec: num}}\label{sec:NumVeri}

Under the ansatz \eqref{Value_solution}, and in the regime where the portfolio bound is
slack, the value function of the non-exploratory problem is
\begin{equation}\label{V0_form}V^{(0)}(t,x,y)=\frac{x^{1-\eta}e^{L(t)y+M(t)}-1}{1-\eta},
\qquad u^{(0)}(t,y)=L(t)y+M(t),
\end{equation}
with $L$ and $M$ solving the coupled ordinary differential equations obtained by
identifying the coefficients of $y$ and of $1$ in \eqref{sol_u}. This motivates the
parametrization used by the critic,
\begin{equation}\label{L_param}
	L^\Psi(t)=\frac{-\Psi_1+\Psi_1 e^{\Psi_0 (T-t)}}{-\Psi_2+\Psi_3 e^{\Psi_0 (T-t)}},
	\qquad
	M^\Psi(t)=\Psi_4 (T-t)+\Psi_5\log\!\Big(
	\frac{-\Psi_2+\Psi_3 e^{\Psi_0 (T-t)}}{-\Psi_2+\Psi_3}\Big),
\end{equation}
with $\psi=(\Psi_0,\dots,\Psi_5)\in\bbr^6$, and, following Lemma~\ref{lem_2}, the actor
$\lambda^\theta_t\sim\cN(\alpha_\theta(t,y),\beta^{2}(y))_{[a,b]}$ with
\[
\alpha_\theta(t,y)=\frac{\sqrt{y+\delta_*}}{\sqrt{y^{2}+\sigma_*^{2}}}\,
\frac{\theta_4+\theta_5 L^\theta(t)}{\eta},
\qquad
\beta^{2}(y)=\frac{m}{\eta\,\sigma^{2}(y)},
\]
where $L^\theta$ has the form \eqref{L_param} in the parameters $\theta_0,\dots,\theta_3$.
At the true parameters $\theta_4=k_2$, so both classes are exact
wherever $u_y$ does not depend on $y$; once the bound binds we enrich them by one term
each, $u^\psi=L^\psi(t)y+N^\psi(t)y^{2}+M^\psi(t)$ and
$\alpha_\theta\mapsto\alpha_\theta+\theta_6\,y$.

From \eqref{V0_form} and \eqref{L_param} we have 
\[
E = e^{\psi_0(T-t)}, \qquad
D = -\psi_2 + \psi_3 E, \qquad
B = -\psi_2 + \psi_3 ,
\]
\[
L(t)=\frac{\psi_1(E-1)}{D},\qquad
M(t)=\psi_4(T-t)+\psi_5\log\!\left(\frac{D}{B}\right),
\]
Hence, 

\[
V(t,x,y)=\frac{x^{1-\eta}e^{\mathrm{exp\_tot}} - 1}{1-\eta},
\]
where \[
\mathrm{exp\_tot}
= L(t)y + M(t) + \frac12 m(1-\eta)(T-t)\log m,
\]

\begin{align*}
\frac{\partial V}{\partial \psi_0}
&=
\frac{x^{1-\eta} e^{\mathrm{exp\_tot}}}{1-\eta}
(T-t)E
\left[
y\,\frac{\psi_1(\psi_3-\psi_2)}{D^2}
+ \psi_5\frac{\psi_3}{D}
\right], \\[0.5em]
\frac{\partial V}{\partial \psi_1}
&=
\frac{x^{1-\eta} e^{\mathrm{exp\_tot}}}{1-\eta}
\,y\,\frac{E - 1}{D}, \\[0.5em]
\frac{\partial V}{\partial \psi_2}
&=
\frac{x^{1-\eta} e^{\mathrm{exp\_tot}}}{1-\eta}
(E-1)
\left[
y\,\frac{\psi_1}{D^2}
+ \frac{\psi_3\psi_5}{D\,B}
\right], \\[0.5em]
\frac{\partial V}{\partial \psi_3}
&=
\frac{x^{1-\eta} e^{\mathrm{exp\_tot}}}{1-\eta}
(E-1)
\left[
-\,y\,\frac{\psi_1 E}{D^2}
-\,\frac{\psi_2\psi_5}{D\,B}
\right], \\[0.5em]
\frac{\partial V}{\partial \psi_4}
&=
\frac{x^{1-\eta} e^{\mathrm{exp\_tot}}}{1-\eta}\,(T-t), \\[0.5em]
\frac{\partial V}{\partial \psi_5}
&=
\frac{x^{1-\eta} e^{\mathrm{exp\_tot}}}{1-\eta}
\log\!\left(\frac{D}{B}\right).
\end{align*}
For a control $\pi\in[0,1]$, the policy $\lambda_t$ follows a truncated normal
\[
f_{\theta}(\pi \mid t,y)
=
\frac{
	\varphi\!\left(\dfrac{\pi-\mu_{\theta}(t,y)}{\sigma_{\theta}(t,y)}\right)
}{
	\sigma_{\theta}(t,y)\,
	\Big[
	\Phi(\hat b)-\Phi(\hat a)
	\Big]
}\,\mathbf 1_{\{0\le \pi\le 1\}},
\]
where $\varphi$ and $\Phi$ denote the standard normal pdf/cdf, and
\[
z=\frac{\pi-\mu_{\theta}}{\sigma_{\theta}},\qquad
\hat a=-\frac{\mu_{\theta}}{\sigma_{\theta}},\qquad
\hat b=\frac{1-\mu_{\theta}}{\sigma_{\theta}}.
\]
Thus, the log-density is
\(
\log f_{\theta}(\pi\mid t,y)
=
-\frac{1}{2}z^{2}
-\log\sigma_{\theta}
-\frac{1}{2}\log(2\pi)
-
\log\big(
\Phi(\hat b)-\Phi(\hat a)
\big).
\)
A direct computation yields
\[
\frac{\partial \log f_{\theta}}{\partial \mu_{\theta}}
=
\frac{\pi-\mu_{\theta}}{\sigma_{\theta}^{2}}
+
\frac{\varphi(\hat b)-\varphi(\hat a)}
{\sigma_{\theta}\big[
	\Phi(\hat b)-\Phi(\hat a)
	\big]}=:g_{\mu}(t,y,\pi;\theta).
\]
The mean of the truncated normal is
\(
\mu_{\theta}(t,y)
=
C(y)\,\big(\theta_{4}+\theta_{5}L^{\theta}(t)\big),
\;
L^{\theta}(t)=\frac{\theta_{1}(E_{0}-1)}{D_{\theta}(t)},
\)
where 
\(
C(y):=\frac{\sqrt{y+\delta_*}}{\eta\sqrt{y^{2}+\sigma_*^{2}}},\,
E_{0}:=e^{\theta_0(T-t)},
\,
D_{\theta}(t):=-\theta_{2}+\theta_{3}E_{0}.
\)

Since $\sigma_{\theta}(t,y)$ does not depend on~$\theta$,
the chain rule gives, for each parameter $\theta_k$,
\[
\frac{\partial\log f_{\theta}}{\partial\theta_k}
=
g_{\mu}(t,y,\pi;\theta)\,
\frac{\partial\mu_{\theta}(t,y)}{\partial\theta_k}.
\]
The gradients with respect to $(\theta_0,\theta_1,\theta_2,\theta_3,\theta_4,\theta_5)$ is then given by
\begin{align*}
\frac{\partial L^{\theta}(t)}{\partial \theta_0}
&= \frac{(T-t)E_{0}\,\theta_{1}(\theta_{3}-\theta_{2})}{D_{\theta}(t)^{2}},\quad
\frac{\partial \log f_{\theta}}{\partial \theta_0}
= g_{\mu}(t,y,\pi;\theta)\,
C(y)\,\theta_{5}
\frac{(T-t)E_{0}\,\theta_{1}(\theta_{3}-\theta_{2})}
{D_{\theta}(t)^{2}}.
\end{align*}
\begin{align*}
\frac{\partial L^{\theta}(t)}{\partial \theta_1}
&= \frac{E_{0}-1}{D_{\theta}(t)}, \quad
\frac{\partial \mu_{\theta}}{\partial \theta_1}= C(y)\,\theta_{5}\frac{E_{0}-1}{D_{\theta}(t)}, \quad
\frac{\partial \log f_{\theta}}{\partial \theta_1}
= g_{\mu}(t,y,\pi;\theta)\,
C(y)\,\theta_{5}\frac{E_{0}-1}{D_{\theta}(t)}.
\end{align*}
\begin{align*}
\frac{\partial L^{\theta}(t)}{\partial \theta_2}
= \frac{\theta_{1}(E_{0}-1)}{D_{\theta}(t)^{2}},
\;
\frac{\partial \mu_{\theta}}{\partial \theta_2}
= C(y)\,\theta_{5}\frac{\theta_{1}(E_{0}-1)}{D_{\theta}(t)^{2}}, \;
\frac{\partial \log f_{\theta}}{\partial \theta_2}
= g_{\mu}(t,y,\pi;\theta)\,
C(y)\,\theta_{5}\frac{\theta_{1}(E_{0}-1)}{D_{\theta}(t)^{2}}.
\end{align*}

\begin{align*}
\frac{\partial L^{\theta}(t)}{\partial \theta_3}
= -\frac{\theta_{1}(E_{0}-1)E_{0}}{D_{\theta}(t)^{2}},
\,\frac{\partial \mu_{\theta}}{\partial \theta_3}
= -C(y)\,\theta_{5}\frac{\theta_{1}(E_{0}-1)E_{0}}{D_{\theta}(t)^{2}}, \\[0.5em]
\frac{\partial \log f_{\theta}}{\partial \theta_3}
= -\,g_{\mu}(t,y,\pi;\theta)\,
C(y)\,\theta_{5}
\frac{\theta_{1}(E_{0}-1)E_{0}}{D_{\theta}(t)^{2}}.
\end{align*}
\begin{align*}
\frac{\partial \mu_{\theta}}{\partial \theta_4}
= C(y),
\frac{\partial \log f_{\theta}}{\partial \theta_4}
= g_{\mu}(t,y,\pi;\theta)\,C(y), \;
\frac{\partial \mu_{\theta}}{\partial \theta_5}
= C(y)\,L^{\theta}(t),\;
\frac{\partial \log f_{\theta}}{\partial \theta_5}
= g_{\mu}(t,y,\pi;\theta)\,C(y)\,L^{\theta}(t).
\end{align*}

\section{Verification of the conditions in Proposition~\ref{Prop_4} under the simulation setup} \label{sec:App-num}

\noindent The experiments of Section~\ref{sec: num} use the calibration of
\eqref{eq:num-par}, namely $k_1=0.381$, $k_2=1.235$, $\delta_*=0.079$,
$\sigma_*=0.054$, $c=2.446$, $\bar y=0.143$, $\rho=-0.779$ and $r=2\%$, with
$\eta=0.7$ and the temperature $m=\epsilon^{2}\eta\sigma^{2}(\bar y)$; at
$\epsilon=0.15$ this is $m=3.69\cdot10^{-4}$. The reduced equation is solved on the
truncated factor domain $y\in[-\delta_*,\,0.955]$, ten standard deviations of the
stationary law above its mean. On that domain
\[
\sup_{y}\sigma^{2}(y)=0.9142,
\qquad
\bar\mu-r:=\sup_{y}\big(\mu(t,y)-r\big)=1.2003 .
\]

Condition (i) of Proposition~\ref{Prop_4} requires
\[
\frac{(\bar\mu-r)^{2}}{\eta\sigma_{*}^{2}}
+ m\ln\!\Big(\tfrac{2\pi_\zs{e} m}{\eta\sigma_{*}^{2}}\Big)
+2m\ln\!\Big(\Phi(B_{0})-\Phi(A_{0})\Big)>0 ,
\qquad
A_0=a\sqrt{\tfrac{\eta\sigma_*^{2}}{m}}-\tfrac{\bar\mu-r}{\sqrt{m\eta\sigma_*^{2}}},
\quad
B_0=b\sqrt{\tfrac{\eta\sigma_*^{2}}{m}}-\tfrac{\bar\mu-r}{\sqrt{m\eta\sigma_*^{2}}},
\]
which evaluates to $2.39$ for $[a,b]=[0,1]$ and to $7.18$ for $[a,b]=[0,3]$. The first
term alone is $695.5$; the truncation term $2m\ln(\Phi(B_{0})-\Phi(A_{0}))$ cancels
almost all of it, so the margin is small but positive. Condition (ii) requires
$\sup_y\sigma^{2}(y)\le q^{m}_{0}/\eta$ with $q^{m}_{0}$ the smallest positive root of
$f^{a,b}_{m}$. Numerical root finding gives $q^{m}_{0}=356.65$, hence
$q^{m}_{0}/\eta=509.50$, against $\sup_y\sigma^{2}(y)=0.9142$. Both conditions are
therefore met for every configuration reported in
Section~\ref{sec: num}.

\end{document}